\newif\ifshaphered

\ifdefined\isshaphered
\shapheredtrue
\fi

\documentclass[sigconf]{acmart}

\AtBeginDocument{%
  \providecommand\BibTeX{{%
    \normalfont B\kern-0.5em{\scshape i\kern-0.25em b}\kern-0.8em\TeX}}}

    \makeatletter
\def\@ACM@checkaffil{
    \if@ACM@instpresent\else
    \ClassWarningNoLine{\@classname}{No institution present for an affiliation}%
    \fi
    \if@ACM@citypresent\else
    \ClassWarningNoLine{\@classname}{No city present for an affiliation}%
    \fi
    \if@ACM@countrypresent\else
        \ClassWarningNoLine{\@classname}{No country present for an affiliation}%
    \fi
}
\makeatother

\setcopyright{acmcopyright}
\copyrightyear{2018}
\acmYear{2018}
\acmDOI{XXXXXXX.XXXXXXX}

\acmConference[Conference acronym 'XX]{Make sure to enter the correct
  conference title from your rights confirmation emai}{June 03--05,
  2018}{Woodstock, NY}
\acmPrice{15.00}
\acmISBN{978-1-4503-XXXX-X/18/06}

\usepackage{algorithm}
\usepackage{algorithmicx}
\usepackage[noend]{algpseudocode}
\usepackage{bbm}
\usepackage{soul}
\usepackage{diagbox}
\usepackage{enumerate}

\ifshaphered
\usepackage{longtable}
\newcommand\revision[1]{\textcolor{blue}{#1}}

\else
\newcommand\revision[1]{#1}

\fi

\begin{document}

\title{Local Differentially Private Heavy Hitter Detection in Data Streams with Bounded Memory}

\settopmatter{printacmref=false}


\author{Xiaochen Li}
\authornote{Work at The State Key Laboratory of Blockchain and Data Security.}
\email{xiaochenli@zju.edu.cn}
\affiliation{%
  \institution{Zhejiang University}
}

\author{Weiran Liu}
\email{weiran.lwr@alibaba-inc.com}
\affiliation{%
  \institution{Alibaba Group}
}

\author{Jian Lou}
\authornotemark[1]
\email{jian.lou@zju.edu.cn}
\affiliation{%
  \institution{Zhejiang University}
}

\author{Yuan Hong}
\email{yuan.hong@uconn.edu}
\affiliation{%
  \institution{University of Connecticut}
}

\author{Lei Zhang}
\email{zongchao.zl@taobao.com}
\affiliation{%
  \institution{Alibaba Group}
}

\author{Zhan Qin}
\authornotemark[1]
\authornote{Zhan Qin is the corresponding author.}
\email{qinzhan@zju.edu.cn}
\affiliation{%
  \institution{Zhejiang University}
}

\author{Kui Ren}
\authornotemark[1]
\email{kuiren@zju.edu.cn}
\affiliation{%
  \institution{Zhejiang University}
}

\renewcommand{\shortauthors}{Anonymous.}

\begin{abstract}
Top-$k$ frequent items detection is a fundamental task in data stream mining. Many promising solutions are proposed to improve memory efficiency while still maintaining high accuracy for detecting the Top-$k$ items. Despite the memory efficiency concern, the users could suffer from privacy loss if participating in the task without proper protection, since their contributed local data streams may continually leak sensitive individual information. However, most existing works solely focus on addressing either the memory-efficiency problem or the privacy concerns but seldom jointly, which cannot achieve a satisfactory tradeoff between memory efficiency, privacy protection, and detection accuracy.

In this paper, we present a novel framework HG-LDP to achieve accurate Top-$k$ item detection at bounded memory expense, while providing rigorous local differential privacy (LDP) protection. Specifically, we identify two key challenges naturally arising in the task, which reveal that directly applying existing LDP techniques will lead to an inferior ``accuracy-privacy-memory efficiency'' tradeoff. Therefore, we instantiate three advanced schemes under the framework by designing novel LDP randomization methods, which address the hurdles caused by the large size of the item domain and by the limited space of the memory. We conduct comprehensive experiments on both synthetic and real-world datasets to show that the proposed advanced schemes achieve a superior ``accuracy-privacy-memory efficiency'' tradeoff, saving $2300\times$ memory over baseline methods when the item domain size is $41,270$. Our code is open-sourced via the link.\footnote{\url{https://github.com/alibaba-edu/mpc4j/tree/main/mpc4j-dp-service}}
\end{abstract}




\received{20 February 2007}
\received[revised]{12 March 2009}
\received[accepted]{5 June 2009}

\maketitle

\section{Introduction}
Detecting Top-$k$ frequent items in data streams is one of the most fundamental problems in streaming data analysis \cite{DBLP:journals/vldb/JinYCYL10, DBLP:journals/tods/LiuWY10, DBLP:conf/vldb/DasGKS07, DBLP:conf/sigmod/ChenC15, DBLP:journals/ton/BasatEKOVW22}. It forms the foundation for a multitude of critical applications across various domains, such as anomaly detection in data mining \cite{DBLP:journals/csur/ChandolaBK09}, click analysis in web analysis \cite{DBLP:journals/electronicmarkets/SteinfieldAL05}, and topic mining in social networks \cite{DBLP:conf/www/ZhangCYNL13}. In the typical decentralized setting as illustrated in Figure \ref{fig:heavy_exp}, the users send local item counts to the server in a streaming fashion, and the server continuously finds hot items (i.e., items with high-frequencies, as depicted in Figure \ref{fig:heavy_exp}) in the item domain based on all users' local streams. 
Apparently, the na\"ive solution to count and store all items ever appearing in the data streams will incur an \emph{overwhelming memory burden} for large domain sizes that are commonly encountered in practice.
\revision{
For example, as of 2023, there are over $1$ billion videos uploaded on the Youtube platform, with over $500$ hours of videos uploaded every minute \cite{youtube}.
It becomes evident that maintaining a histogram of the expanding item domain on the server side for identifying hot items is impractical.
} 
Many existing works thus focus on improving memory efficiency by designing advanced data structures, especially for applications where the domain size is too large to efficiently fit in the memory \cite{DBLP:conf/infocom/Ben-BasatEFK16, DBLP:conf/pods/YiZ09, DBLP:journals/pvldb/CormodeH08, DBLP:conf/vldb/CormodeKMS03, DBLP:conf/stoc/BravermanCIW16}. 

Furthermore, users' submitted streaming data often contain sensitive individual information, e.g., click analysis may reveal online behavior and topic mining may reveal political opinions. The privacy of users is under severe threat if they submit local data streams without proper privacy protection. 
In particular, the privacy concern has a unique characteristic in the Top-$k$ detection problem.
That is, the cold items (i.e., items with low frequencies, as depicted in Figure \ref{fig:heavy_exp}) are not statistical targets, but constitute the majority of the data domain and are particularly sensitive, as they reveal highly personal information specific to certain user groups.
Due to its central role in streaming data analysis, Top-$k$ frequent items detection has attracted significant research attention in recent years. 
However, most existing works pursue the memory efficiency or privacy protection goals separately but seldom jointly. 

\setlength{\textfloatsep}{0.1cm}
\begin{figure}[t]
	\centering
	\includegraphics[width=0.43\textwidth]{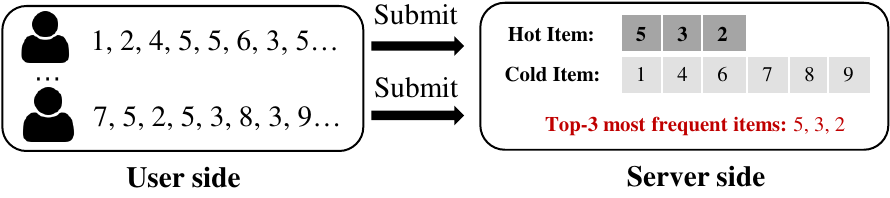}
 \vspace{-0.1in}
	\caption{An example of Top-$3$ frequent items detection.}
	\label{fig:heavy_exp}
\end{figure}

On the memory efficiency side, a series of approaches have been proposed to improve the memory efficiency with decent accuracy for detecting the Top-$k$ items  \cite{DBLP:conf/kdd/0003GZZSL18, DBLP:conf/sigmod/Zhou0J0YLU18, DBLP:conf/kdd/LiLXJ00DZ20, DBLP:conf/vldb/MankuM02, DBLP:conf/icdt/MetwallyAA05, DBLP:journals/pvldb/LiCZYC22}. 
The key rationale of the memory-saving stems from the fact that most items are cold while only a few items are hot in practical data streams \cite{DBLP:conf/sigmod/RoyKA16, cormode2011sketch}. Accurately recording the information of massive cold items not only wastes much memory, but also incurs non-trivial errors in hot item estimation when the memory is tight.  
Thus, existing methods seek to design a compact data structure to keep and guard the items and their frequencies of hot items, while possibly evicting cold items. 
One of the most widely adopted and effective data structures addressing this challenge is \emph{HeavyGuardian} \cite{DBLP:conf/kdd/0003GZZSL18}. 
It introduces the \emph{separate-and-guard-hot} design principle, which effectively segregates hot items from cold items, preserving the accuracy of hot item estimations.
\emph{HeavyGuardian} further delineates a specific strategy called \emph{Exponential Decay} (ED) to guard the hot items by exponentially decreasing the probability that the possible cold items remain in the heavy part of the data structure. However, despite achieving a promising balance between accuracy and memory efficiency, none of these methods simultaneously account for privacy concerns.

On the privacy protection side, Differential Privacy (DP) has been regarded as a \emph{de facto} standard by both academia and industry \cite{DBLP:conf/icalp/Dwork06, DBLP:journals/fttcs/DworkR14}.
In the decentralized data analytics setting, 
Local Differential Privacy (LDP) is the state-of-the-art approach extended from DP to the local setting, which has been widely deployed in industry, e.g., Google Chrome browser \cite{DBLP:conf/ccs/ErlingssonPK14} and Apple personal data collection \cite{greenberg2016apple}. 
In LDP, each user perturbs his/her data with a local randomization mechanism before sending it to the server. 
The server could still derive general statistics from the perturbed submissions with a certain accuracy decrease. General randomization mechanisms for frequency estimation such as Generalized Randomized Response (GRR), Optimal Local Hash (OLH) \cite{wang2017locally}, and Hadamard Response (HR) \cite{acharya2019hadamard}, can be applied to Top-$k$ items detection as baseline methods. 
There also exist many works designed specifically for heavy hitter estimation under LDP, including estimates over the single-valued data \cite{DBLP:conf/icml/GillenwaterJMD22, DBLP:conf/nips/BassilyNST17, DBLP:journals/tdsc/0001LJ21, DBLP:journals/pvldb/CormodeMM21}, and set-valued data \cite{DBLP:conf/ccs/QinYYKXR16, DBLP:conf/sp/WangLJ18}. 
However, it is noteworthy that these works neither address the data stream setting nor tackle the issue of memory efficiency.

In this paper, our objective is to bridge the gap between memory-efficient heavy hitter tracking in data streams and LDP privacy protection. To achieve this, we introduce the HG-LDP framework designed for tracking the Top-$k$ heavy hitters within data streams. This framework comprises three essential modules. 
First, the \emph{randomization module} is responsible for randomizing the streaming data generated by users, ensuring event-level LDP privacy that is more suitable for the streaming data \cite{DBLP:conf/stoc/DworkNPR10, perrier2018private}. 
Second, the \emph{storage module} records the incoming data on the server side. 
\revision{To this end, we integrate the $\emph{HeavyGuardian}$ data structure, and significantly optimize its implementations, i.e., dynamic parameter configuration, and sampling optimization (see details in Appendix \ref{app:imple-detail}) to facilitate the heavy hitter tasks and processing of LDP-protected noisy data. Finally, the \emph{response module} processes and publishes the statistical results of heavy hitters.}


\revision{
It is worth noting that directly applying existing LDP techniques cannot achieve satisfactory accuracy or would be even functionally infeasible, primarily due to the following two new challenges.} 

\revision{
\emph{Challenge (1): Incompatibility of Space-Saving Strategy and Large Domain Size for LDP.} 
To highlight this challenge, we instantiate a basic scheme BGR as a baseline (detailed in Section \ref{sec:basicldppd}), which directly uses the Generalized Randomized Response (GRR) mechanism \cite{wang2017locally} in the randomization module. The large domain size incurs two problems that jointly fail BGR: 1) the noise variance introduced by the GRR will increase as the data domain increases; 2) the space-saving strategy of the data structure introduces additional underestimation error to the noise items, which will be further amplified by the debiasing operation, required by LDP. 
Although existing mechanisms such as Optimal Local Hash (OLH) \cite{wang2017locally} and Hadamard Response (HR) \cite{acharya2019hadamard} in the LDP field aim to alleviate the impact of large data domains on randomized results' accuracy, it is crucial to emphasize that we still confront a unique and unaddressed challenge. 
We identified that the core idea of the LDP field in addressing this problem is to encode the large data domain into a smaller one for randomization.
However, the decoding of randomized data on the server side inevitably produces a multiple of diverse collision data, which can significantly disrupt the decision-making of the space-saving strategy.}

\revision{
\emph{Challenge (2): Dynamically Changing Hot/Cold Items}. 
Notably, cold items often constitute the majority of the data domain, and indiscriminately randomizing data across the entire domain can result in an unnecessary waste of privacy budget.
The ability to distinguish between hot items and cold items during the randomization process is crucial for enhancing the accuracy of hot item estimation. 
However, since the labels of hot and cold items may dynamically change as the data stream evolves, randomizing data based on the previous timestamp's state may introduce a huge bias towards the prior state. 
This poses several new challenges, e.g., how to strike a balance between reducing unnecessary privacy budget expenditure on cold items, and how to manage such dynamically emerging bias. 
Addressing this challenge also mandates novel LDP mechanism designs.
}


\vspace{0.05in}

\noindent\textbf{Contribution.} 
In this paper, we initiate a baseline method and propose three novel advanced LDP designs under the hood of a framework HG-LDP to address these hurdles. 
First, we present a baseline method that directly combines the GRR mechanism with \emph{HeavyGuardian} data structure.
Second, we propose a newly designed LDP mechanism. It is based on the observation that the ED strategy does not need to know the specific item of the incoming data in most cases if it is not recorded in the data structure. Third, we adjust the noise distribution by dividing the privacy budget to achieve higher accuracy. Finally, we utilize the light part of \emph{HeavyGuardian} to elect current cold items before they become new hot items, which further improves the accuracy of the estimated result. The main contributions are summarized as follows.
\begin{itemize}
  \item To our best knowledge, this paper is the first to track the Top-$k$ frequent items from data streams in a bounded memory space while providing LDP protection for the sensitive streaming data. 
  We present a general framework called HG-LDP to accommodate any proper LDP randomization mechanisms on the users' side into the space-saving data structures on the server side for the task.

\vspace{0.05in}
  
  \item By investigating the failure of na\"ively combining existing LDP techniques with HG-LDP, we design three new LDP schemes, which achieve a desired tradeoff performance between accuracy, privacy, and memory efficiency.

\vspace{0.05in}
  
  \item We comprehensively evaluate the proposed schemes on both synthetic and real-world datasets in terms of accuracy and memory consumption, which shows that the proposed schemes achieve higher accuracy and higher memory efficiency than baseline methods. For instance, when the size of the domain size reaches $41,270$, the proposed schemes save about $2300\times$ size of memory over baselines.
\end{itemize}

\section{Preliminaries}
\subsection{Problem Statement}
We consider the setting of finding Top-$k$ items in data streams under Local Differential Privacy (LDP).
Given $n$ users, each user generates a private infinite data stream.
Denote $v_i^t\in \Omega$ as the data generated by the user $u_i$ at timestamp $t$. The user only sends data at the timestamp when data is generated. 
A server collects values from users at each timestamp $t$.
Note that the server can only maintain a data structure with a length much smaller than the size $d$ of data domain $\Omega$ due to its limited memory space.
Whenever a query is received, the server needs to publish the Top-$k$ items up to the latest timestamp and their counts.

\subsection{Privacy Definitions}
In this paper, we provide event-level privacy guarantee \cite{DBLP:conf/stoc/DworkNPR10,perrier2018private, wang2021continuous,DBLP:conf/ccs/ChenMHM17, DBLP:conf/icalp/ChanSS10}.
Specifically, the event-level LDP ensures the indistinguishability of any pairs of elements in streams, e.g., every single transaction remains private in a user's long-term transactions:
\begin{definition}[Local Differential Privacy (LDP) \cite{DBLP:conf/focs/KasiviswanathanLNRS08}]
  An algorithm $\mathcal{M}$ satisfies $\epsilon$-LDP, where $\epsilon \geq 0$, if and only if for any input $v, v' \in \mathbb{D}$, and any output $y \in Range(\mathcal{M})$, we have 
  \[\setlength\abovedisplayskip{0.5ex}
     \setlength\belowdisplayskip{0.5ex}
  \Pr \left[\mathcal{M}(v) = y \right] \leq e^{\epsilon} \Pr \left[\mathcal{M}(v') = y \right].\]
  \label{def:local-differential-privacy} 
\end{definition}
The parameter $\epsilon$ is called the \emph{privacy budget}, whereby smaller $\epsilon$ reflects stronger privacy guarantees. We say $\mathcal{M}$ satisfies $\epsilon$-LDP if for different data $v$ and $v'$, the ratio of distribution of output $\mathcal{M}(v)$ and that of $\mathcal{M}(v')$ are not greater than $e^{\epsilon}$. 

\subsection{LDP Mechanisms}
The Randomized Response (RR) mechanism \cite{warner1965randomized} is considered to be the first LDP mechanism that supports binary response.
It allows each user to provide a false answer with a certain probability so as to provide plausible deniability to users. The Generalized Randomized Response (GRR) mechanism \cite{wang2017locally} is an extension of Randomized Response (RR) \cite{warner1965randomized}, which supports multi-valued domain response.
Denote $d$ as the size of the domain $\mathbb{D}$.
Each user with private value $v \in \mathbb{D}$ reports the true value $v' = v$ with probability $p$ and reports a randomly sampled value $v' \in \mathbb{D}$ where $v' \neq v$ with probability $q$.
The probability $p$ and $q$ are defined as follows
\begin{equation}
	\label{eq:rr}
	\left\{
	\begin{aligned}	
		p & = \frac{e^{\epsilon}}{e^{\epsilon} + d - 1}, \\
		q & = \frac{1}{e^{\epsilon} + d - 1}.
	\end{aligned}
	\right. 
\end{equation}
where $d$ is the size of the data domain.
It is straightforward to prove $\epsilon$-LDP for GRR, i.e., $p/q\le e^\epsilon$ \cite{wang2017locally}.
Assuming that each of the $n$ users reports one randomized value. Let $\hat c_i$ be the number of value $i$ occurs in randomized values, the estimation of true number $\tilde c_{i}$ of value $i$ can be computed with 
$\tilde c_i = \frac{\hat c_i - nq}{p-q}$. 
The variance of the estimated result $\tilde c_{i}$ is 
$Var[\tilde c_{i}] = n\cdot\frac{d-2+e^\epsilon}{(e^\epsilon-1)^2}$.

As shown above, the variance of the estimation result of the GRR mechanism increases linearly with the increase of $d$.
Some other mechanisms, e.g., Optimal Local Hash (OLH) \cite{wang2017locally} and Hadamard Response (HR) \cite{acharya2019hadamard}, are proposed to randomize data in a large data domain.
Essentially, they map the data to a smaller domain before randomizing it to avoid the large variance caused by a large data domain.
We defer their details to Appendix \ref{app:ldp-mechanism}.

\subsection{Space-Saving Data Structure}
\label{sec:hg}
\revision{
Counter-based data structures \cite{DBLP:conf/vldb/MankuM02, DBLP:conf/icdt/MetwallyAA05, DBLP:conf/kdd/0003GZZSL18, DBLP:conf/sigmod/Zhou0J0YLU18} and sketches \cite{alon1999space, cormode2005improved, DBLP:conf/kdd/LiLXJ00DZ20, chi2004moment} are two kinds of mainstream memory-efficient data structures.
While sketches have been extensively studied as compressed data structures for frequency estimation, they may not be the optimal choice when it comes to heavy hitter estimation in data streams, particularly in scenarios characterized by limited storage space and real-time response requirements. 
This preference is underpinned by two key reasons: 
Firstly, sketches record counts for all items, whereas heavy hitter tasks only concern hot items. 
This equally treated recording of all counts results in unnecessary memory consumption. 
For example, the Count-Min sketch (CMS) necessitates a minimum of $O(\frac{N}{\alpha} \times \log(1/\delta))$ space to guarantee that the probability of error in the estimated count of each item being less than $\alpha$ is no less than $1-\delta$, with $N$ representing the total data count \cite{cormode2005improved}. 
Furthermore, as highlighted by Cormode and Hadjieleftheriou in \cite{DBLP:journals/pvldb/CormodeH08}, sketches require additional storage for finding the counts of hot items. 
For instance, $O(\frac{N}{\alpha} \log d \log\delta)$ space increase is incurred when using group testing to find hot items, or a minimum of $O(d)$ computational overhead is needed for hot item retrieval.}

\revision{
Thus, in this paper, we choose to employ a counter-based data structure called \emph{HeavyGuardian} proposed by Yang et al. \cite{DBLP:conf/kdd/0003GZZSL18} as the foundation for our framework.
It identifies and records the high-frequency items in subsequent data streams based on observations of historical streaming data.
}
The basic version of \emph{HeavyGuardian} is a hash table with each bucket storing several KV pairs ($\langle ID, count\rangle$) and small counters. Specifically, each bucket is divided into two parts: a heavy part with a length of $\lambda_h$ ($\lambda_h>0$) to precisely store counts of hot items, and a light part with a length of $\lambda_l$ ($\lambda_l$ can be $0$) to approximately store counts of cold items. For each incoming item $e$, \emph{HeavyGuardian} needs to decide whether and how to insert it into the heavy part of a bucket according to a strategy called Exponential Decay (ED).
There are three cases when inserting an item $e$ into the heavy part of \emph{HeavyGuardian}. 
\begin{enumerate}[\emph{Case} 1:]
  \item The KV pair of $e$ has been stored in the heavy part, it increments the corresponding $count$ by $1$.

\vspace{0.05in}
  
  \item The KV pair of $e$ is not in the heavy part, and there are still empty buckets. 
        It inserts the KV pair of $e$ into the heavy part and sets the $count$ to $1$.

\vspace{0.05in}
        
  \item The KV pair of $e$ is not in the heavy part, and there is no empty bucket.
        It decays $1$ from the current least $count$ in the heavy part with probability $\mathcal{P}=b^{-c}$, where $b$ is a predefined constant number (b=1.08 in \cite{DBLP:conf/kdd/0003GZZSL18}), and $c$ is the $count$ value.
        After decay, if the $count$ becomes $0$, it replaces this KV pair (the weakest KV pair) with $e$'s KV pair, and sets the $count$ to $1$.
\end{enumerate}

If $e$ is not successfully inserted into the heavy part, it is recorded in the light part.
Since the heavy hitter tasks only focus on Top-$k$ items and their counts, we set the parameters of \emph{HeavyGuardian} as the number of buckets $w=1$, the length of the heavy part $\lambda_h=k$, and the length of light part $\lambda_l=0$ (except in one of the proposed scheme CNR).
For simplicity of description, we denote the data structure of \emph{HeavyGuardian} as $\mathcal{HG}$ in the following sections.
We use $\mathcal{HG}[i]$ to denote the $i^\text{th}$ key pair in $\mathcal{HG}$, and use $\mathcal{HG}[i].ID$ and $\mathcal{HG}[i].C$ to denote the ID and the count of an item, respectively. 

\section{HG-LDP for Heavy Hitters Tracking}
\revision{
In this section, we first introduce the HG-LDP framework for tracking heavy hitters in data streams with bounded memory space.
Then, we instantiate a baseline to highlight key obstacles for achieving a satisfactory ``\emph{accuracy-privacy-memory efficiency}'' tradeoff.
}
\subsection{Overview}

\begin{figure}[!h]
	\centering
	\includegraphics[width=0.43\textwidth]{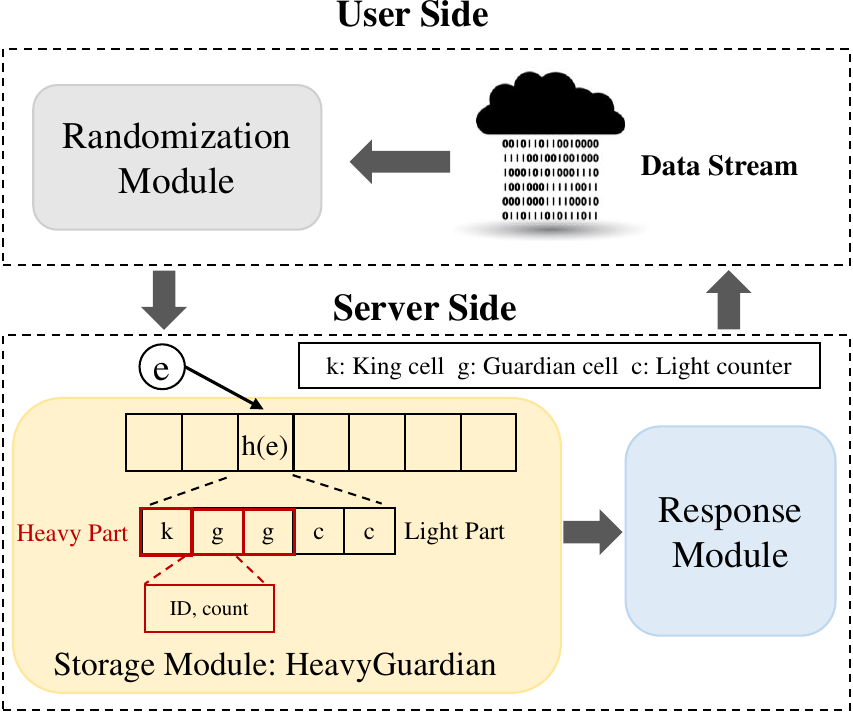}
	\setlength{\belowcaptionskip}{2ex}
	\caption{The overview of HG-LDP.}\vspace{-0.15in}
	\label{fig:heavyguardsys}
\end{figure}

Figure \ref{fig:heavyguardsys} illustrates the framework for HG-LDP, which contains three modules: \emph{randomization module}, \emph{storage module}, and \emph{response module}. The \emph{randomization module} runs on the user side to randomize the users' sensitive streaming data. The \emph{storage module} and the \emph{response module} run on the server side, where the \emph{storage module} utilizes a space-saving data structure. 

In this paper, we aim to adapt and optimize the \emph{HeavyGuardian} ($\mathcal{HG}$) data structure due to its popularity and simplicity, but expect our LDP designs to be generalizable to more sophisticated space-saving data structures in the future. The randomized streaming data continuously reported by users is stored in $\mathcal{HG}$ following the \emph{ED strategy}, and the statistical results are released by the \emph{response module} after \emph{debiasing}. Specifically, the functions of the three modules can be summarized as the following three algorithmic components:

\begin{itemize}
    \item \textsc{Randomize}. It is executed in the \emph{randomization module}. It takes raw data $v_i^t$ of the $i^{\text{th}}$ user at timestamp $t$ as input, and outputs a randomized data $r_i^t$ that satisfies LDP.

\vspace{0.05in}
    
    \item \textsc{Insert}. It is executed in the \emph{storage module}. It inserts the randomized data $r_i^t$ into $\mathcal{HG}$ following the \emph{ED strategy}, and updates the counts of the KV pairs in $\mathcal{HG}$.

\vspace{0.05in}
    
    \item \textsc{Response}. It is executed in the \emph{response module}. It obtains the hot items and their corresponding counts from $\mathcal{HG}$ when receiving a request. Then it maps them to a list for publishing after debiasing all counts.
\end{itemize}

\revision{In the following sections, we first instantiate a baseline scheme, and then propose three advanced schemes based on this framework by elaborately designing algorithms for the three modules.}

\subsection{A Baseline Scheme: BGR}
\label{sec:basicldppd}
We first discuss a baseline scheme BGR (Basic Scheme Combining GRR) that directly integrates an existing LDP scheme: GRR. 

\medskip
\noindent\textbf{Algorithms.} 
The BGR algorithm is outlined in Algorithm \ref{alg:LdpPD}. At timestamp $t$, the data $v_i^t$ of a user $u_i$ is randomized using GRR, and the resulting randomized value $r_i^t$ is then transmitted to the server. Subsequently, the server incorporates $r_i^t$ into the data structure $\mathcal{HG}$ following the \emph{ED strategy}.
Note that the counts stored within $\mathcal{HG}$ are consistently biased noisy values. To mitigate this, the server debiases all counts in the \emph{response module} following the standard GRR debiasing approach \cite{wang2017locally} before publishing the statistical outcomes.


\setlength{\floatsep}{0.1cm}
\setlength{\textfloatsep}{0.1cm}
\begin{algorithm}[!h]
  \caption{BGR (baseline)}
  \label{alg:LdpPD}
  {\small{
  \begin{algorithmic}[1]
        \Require {timestamp $t$, data domain $\Omega$ with size $d$, data structure $\mathcal{HG}$, number of the received data $num$.}
        \Ensure {$ResponseList$}
        \Statex {\textsc{Randomize}}
        \State {Obtain the current raw data $v_i^t$;}
        \State {$r_i^t\leftarrow$ GRR($v_i^t$, $\epsilon$)} \Comment{Randomize data with GRR.}
        \Statex{\textsc{Insert}}
        \State {Receive an incoming data $r_i^t$;}
        \State {$num\leftarrow num+1$}
        \State {Insert $r_i^t$ into $\mathcal{\mathcal{HG}}$ following ED strategy;}
        \If {the least count $\mathcal{HG}[k].C$ $\le 0$}
        \State {Replace the weakest KV pair with new KV pair $<r_i^t, 1>$}
        \EndIf
        \vspace{-0.03in}
        \Statex{\textsc{Response}}
        \State {$p= \frac{e^{\epsilon}}{e^{\epsilon} + d - 1}$, $q= \frac{1}{e^{\epsilon} + d - 1}$}
        \If {receive a Top-$k$ query}
        \For {each $\mathcal{HG}[j]\in \mathcal{HG}$}
        \State{$ResponseList[j].ID\leftarrow\mathcal{HG}[j].ID$}
        \State {$ResponseList[j].C\leftarrow (\mathcal{HG}[j].C-num\cdot q)/(p-q)$}
        \EndFor
        \EndIf\\
        \Return {$ResponseList$}
  \end{algorithmic}}}
\end{algorithm}

\vspace{-0.05in}

\noindent\textbf{Theoretical Analysis.}
Next, we theoretically analyze the error bound of the frequency estimated by BGR. Part of the error comes from the exponential decay of the counts on the server when the coming data is not recorded in $\mathcal{HG}$. Another part of the error comes from the noise introduced by the \emph{randomization module} to perturb the data with the GRR. We first give the error analysis for the \emph{ED strategy} of $\mathcal{HG}$ provided by Yang et al. \cite{DBLP:conf/kdd/0003GZZSL18} in Lemma \ref{lem:pd_error} below.

\begin{lemma}
  \label{lem:pd_error}
  Given a stream prefix $S_t$ with $t$ items in $\Omega$, it obeys an arbitrary distribution and $|\Omega|=d$.
  We assume that there are $w$ buckets to store the hottest $\lambda$ items mapped to them, each item is mapped to a bucket with the probability of $\frac{1}{w}$.
  Let $v_i$ be the $i^{\text{th}}$ hottest item, $f_i$ be the real frequency of $v_i$, and $\tilde f_i$ be the estimated frequency of $v_i$.
  Given a small positive number $\alpha$, we have
  \[Pr[f_i-\tilde f_i\ge \alpha t]\le \frac{1}{2\alpha t}(f_i-\sqrt{f_i^2-\frac{4P_{weak}E(V)}{b-1}})\]
  where $P_{weak}=e^{-(i-1)/w}\times(\frac{i-1}{w})^{l-1}/(l-1)!$, $E(V)=\frac{1}{w}\sum_{j=i+1}^d f_j$.
\end{lemma}

Our theoretical analysis follows the conclusion provided in Lemma \ref{lem:pd_error}.
In fact, Lemma \ref{lem:pd_error} only considers the bias caused by \emph{exponential decays} after the items are recorded as hot items, ignoring the count loss before items are recorded.
However, this count loss is strongly related to the distribution of the data stream and the order of the data arrival, so it's difficult to be theoretically analyzed.
Besides, as we mentioned in Section \ref{sec:hg}, we set the number of buckets $w=1$ in this paper since we only track Top-$k$ heavy hitters and $k$ is a small constant. 
Therefore, we only use the result when $w=1$ in Lemma \ref{lem:pd_error} and we show the error bound of BGR in Theorem \ref{the:LdpPD}.

\begin{figure*}[t]
	\centering
	\includegraphics[width=1\textwidth]{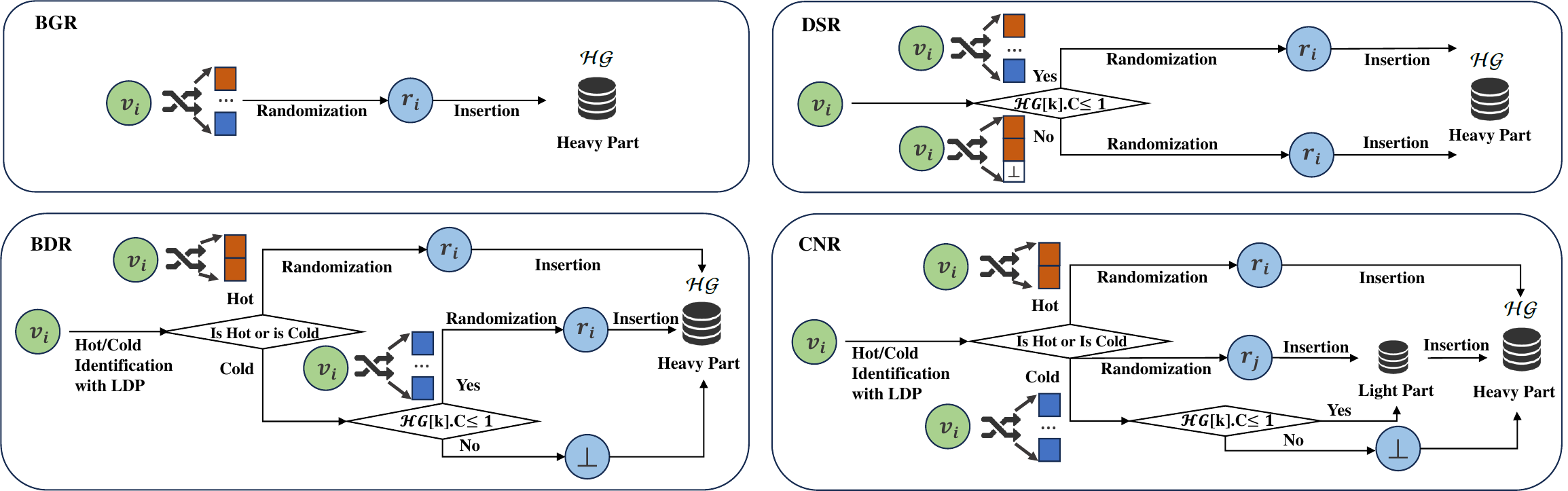}
	\vspace{-0.2in}
	\setlength{\belowcaptionskip}{2ex}
	\caption{The flowcharts of the randomization and storage modules in BGR (baseline) and three advanced schemes. \emph{Note that each subfigure only shows the procedures of one scheme for a single user. The system procedures are more complicated since a large number of users would frequently/concurrently submit data to the server and update the $\mathcal{HG}$ (the domain for randomization frequently changes). A debiasing procedure is also included in the response module of the server}. 
 DSR employs a strategy of randomization within a reduced domain when there are no imminent hot item evictions. BDR mitigates the impact of expansive cold domains on the accuracy of hot item estimates by splitting the privacy budget, which also eliminates the need for switching randomization strategies in DSR. CNR fully utilizes the idle privacy budget in BDR and elects new hot items with more potential to enter $\mathcal{HG}$.\color{black}}\vspace{-0.2in}
	\label{fig:schemes}
\end{figure*}

\begin{theorem}
  \label{the:LdpPD}
  Given a stream prefix $\hat S_t$ with $t$ items randomized by BGR satisfying $\epsilon$-LDP and there is a data structure $\mathcal{HG}$ to store the Top-$k$ items.
  Let $v_i$ be the $i^{\text{th}}$ hottest item, $f_i$ be the real frequency of $v_i$, $\tilde f_i$ be the final estimated frequency of $v_i$.
  We have
  \begin{align*}
    Pr[f_i-\tilde f_i&\le(\sqrt{2t\log (2/\beta)}+\alpha t)\cdot\frac{e^{\epsilon}+d-1}{e^{\epsilon}-1}]\\
    &\ge(1-\beta)(1-\frac{1}{2\alpha}(1-\sqrt{1-\frac{4P_{weak}E(V)}{b-1}}))\\
  \end{align*}
  where $P_{weak}=\frac{(i-1)!(d-k)!}{(d-1)!(i-k)!}$, $E(V)=\sum_{j=i+1}^d f_j$, $\alpha$ and $\beta$ are small positive numbers with $\alpha,\beta\in(0,1)$.
\end{theorem}
\begin{proof}
\revision{
We assume that $\hat f_i$ is the frequency of noisy data recorded in $\mathcal{HG}$ according to the \emph{ED strategy}.
Meanwhile, due to the ED strategy introducing additional errors during the recording of $\hat f_i$, the frequency used for debiasing by the GRR mechanism before publication is denoted as $\bar f_i$.
Then the error bound of final debiased frequency $\tilde f_i$ compared to $f_i$ can be obtained by combining the error bounds of $\hat f_i - \bar f_i$ and $f_i-\frac{\hat f_i-tq}{p-q}$.
The detailed proof is deferred to Appendix \ref{app:ldppd_the}.}
\end{proof}

\noindent\textbf{Problems with Existing LDP mechanisms.}
Theorem \ref{the:LdpPD} shows that the error bound of BGR grows proportionally to the size of the data domain $d$.
While BGR is sufficient for solving the task of finding Top-$k$ items in streaming data with a small data domain, it can inevitably fall into the dilemma that the error is too large when dealing with a large data domain. From the proof of Theorem \ref{the:LdpPD}, we can find that the excessive error caused by the large data domain mainly comes from the randomization process of the GRR.

\revision{
Several LDP mechanisms have been proposed to address data randomization in large data domains. However, directly integrating these mechanisms with \emph{HeavyGuardian} is still problematic in practice. The core concept behind these mechanisms revolves around mapping data from the large data domain to a smaller data domain using techniques such as hash functions \cite{wang2017locally}, Hadamard matrix encoding \cite{acharya2019hadamard, acharya2019communication}, or Bloom filter encoding \cite{erlingsson2014rappor}. Subsequently, data is randomized within this reduced data domain.
}

\revision{
There are several issues with these approaches.
Firstly, decoding a single randomized data on the server side  entails an exhaustive scan of the entire data domain, which becomes computationally expensive for large data domains. 
Furthermore, this approach implies that the server must store the entirety of the data domain, which may contradict the requirement for bounded memory consumption on the server side. 
Additionally, these mechanisms introduce collisions when decoding randomized data for analysis. While such collisions are typically manageable in general frequency estimation tasks due to their uniform distribution, they can render strategies like the ED strategy and other space-saving techniques unusable. Assuming that data is mapped from a large domain of size $d$ to a smaller data domain of size $g$, the average number of collision data generated by decoding a data point is $d/g$. In essence, the arrival density of an item directly impacts its potential to be recorded within the data structure as a hot item. If decoded data is mixed with $d/g-1$ different data points, the true hot item may lose its advantage in being recorded within $\mathcal{HG}$. In scenarios where the domain size $d$ is extremely large, such that $d/g$ surpasses the size $k$ of $\mathcal{HG}$, the entire scheme becomes untenable, and all data points are indiscriminately recorded with equal probability. 
}

\revision{
Consequently, it is desirable to develop novel LDP mechanisms capable of effectively randomizing data within large data domains and addressing challenges posed by the dynamically changing hot/cold items while optimizing the performance of \emph{HeavyGuardian}.}

\section{Advanced LDP Mechanism Designs}
\revision{
In this section, we propose three novel advanced schemes to address the aforementioned problem in BGR by designing new randomization methods, which are outlined 
in Figure \ref{fig:schemes}.}

\subsection{DSR (Domain-Shrinkage Randomization)}
\label{sec:dsr}
Tasks involving heavy hitter estimation in streams often assume that the streaming data follows a Zipf distribution \cite{DBLP:journals/tods/CormodeM05, DBLP:journals/pvldb/CormodeH08, DBLP:conf/vldb/MankuM02}. This assumption aligns well with the distribution observed in various real-world scenarios, such as purchased goods and popular songs.
In these contexts, the data domain predominantly consists of a few frequently occurring hot items, while most items are relatively rare or never appear.
However, the GRR mechanism in BGR randomizes a large number of hot items to these rare items for ensuring LDP, which leads to poor performance of the \emph{ED strategy}. 
Furthermore, the protection of these rare items is critical since they often contain highly sensitive information. 
For instance, an individual might not be concerned about others knowing they have watched popular movies but may be apprehensive about revealing their interest in niche films, as it could inadvertently expose their personal preferences and hobbies. It is based on these observations that we have designed our advanced algorithm, DSR.

Specifically, we refer to the items recorded in $\mathcal{HG}$ as the hot items, and the items not in $\mathcal{HG}$ as the cold items. 
We observe that the \emph{ED strategy} of $\mathcal{HG}$ does not need to know the specific value of the cold items in most cases. It only needs to reduce the count of the KV pair with the lowest frequency in $\mathcal{HG}$ by $1$ with a certain probability when it receives a cold item. The \emph{ED strategy} needs to know the specific value of the cold item to replace the KV pair in $\mathcal{HG}$ only when the weakest KV pair (with the lowest frequency) is going to be evicted. A direct idea is to represent all cold items as ``$\bot$'', and randomize the data on the domain $\{\mathcal{HG}.C\}\cup\{\bot\}$. When the weakest KV pair in $\mathcal{HG}$ is about to be evicted, it changes back to BGR to randomize the data on the entire domain. In this way, the size of the data domain can be reduced from $d$ to $k+1$ when there is no KV pair in $\mathcal{HG}$ going to be replaced, which alleviates the low utility caused by a large domain.


\setlength{\textfloatsep}{0.1cm}
\begin{algorithm}[!h]
  \caption{DSR (\textsc{Randomize})}
  \label{alg:dsr-randomize}
  {\small{
  \begin{algorithmic}[1]
        \Require {timestamp $t$, privacy budget $\epsilon$, data domain $\Omega$ with size $d$, data structure $\mathcal{HG}$.}
        \Ensure {$r_i^t$}
        \State {Obtain the current raw data $v_i^t$;}
        \If {the least count $\mathcal{HG}[k].C\le 1$}
        \State {$r_i^t\leftarrow$ GRR($v_i^t$, $\epsilon$)} \Comment{Randomize data in $\Omega$ with GRR.}
        \Else
        \State {Let $b\leftarrow Ber(\frac{e^{\epsilon}}{e^{\epsilon}+k})$}
        \If {$b==1$} \Comment{Randomize data in reduced domain.}
        \State {$r_i^t=v_i^t$}
        \Else
        \State {$r_i^t=v'$, where $v'\in\{\mathcal{HG}.ID\}\cup\{\bot\}$ and $v'\neq v_i^t$}
        \EndIf
        \EndIf
        \\
        \Return {$r_i^t$}
  \end{algorithmic}}}
\end{algorithm}


\medskip
\noindent\textbf{Algorithms.}
The \textsc{Randomize} Algorithm of DSR is presented in Algorithm \ref{alg:dsr-randomize}. In the general case, the user randomizes data on the shrinking domain ${\mathcal{HG}.C}\cup{\bot}$. If the server receives a " $\bot$ ", it reduces the count of the weakest KV pair by $1$ with a certain probability. However, this approach poses a challenge when the count of the weakest KV pair reduces to $0$, as it becomes uncertain which cold item should replace the weakest KV pair. Furthermore, requiring the user to re-randomize the data across the entire domain can potentially violate $\epsilon$-Local Differential Privacy ($\epsilon$-LDP). To solve this problem, DSR requires users to switch to BGR for randomization on the entire data domain when the count of the weakest KV pair reaches $1$ or less. In this case, as long as the count of the weakest KV pair is reduced by $1$, it can be replaced by a new KV pair with a cold item directly. Users can subsequently switch back to randomizing data on the reduced domain once the new KV pair stabilizes (i.e., reaches a count $>1$).

\setlength{\floatsep}{0.1cm}
\setlength{\textfloatsep}{0.1cm}
\begin{algorithm}[t]
  \caption{DSR ($\textsc{Insert}$)}
  \label{alg:dsr-insert}
  {\small{
  \begin{algorithmic}[1]
        \Require {timestamp $t$, privacy budget $\epsilon$, data domain $\Omega$ with size $d$, reduced domain $\Omega_s=\{\mathcal{HG}.ID\}\cup\{\bot\}$, data structure $\mathcal{HG}$, number of the received data $\in\Omega$ $num_{entire}$, number of the received data $\in\Omega_s$ $num_{reduced}$.}
        \Ensure {Updated $\mathcal{HG}$}
        \State {$p_1= \frac{e^{\epsilon}}{e^{\epsilon} + d - 1}$, $q_1= \frac{1}{e^{\epsilon} + d - 1}$, $p_2= \frac{e^{\epsilon}}{e^{\epsilon} + k}$, $q_2= \frac{1}{e^{\epsilon} + k}$}
        \State {Receive an incoming data $r_i^t$;}
        \State {$\mathcal{HG}\leftarrow\textsf{DSR\_Insert}$($r_i^t$, $\mathcal{HG}$, $p_1$, $q_1$, $p_2$, $q_2$, $num_{entire}$, $num_{reduced}$)}\\
        \Return {Updated $\mathcal{HG}$}
  \end{algorithmic}}}
\end{algorithm}


\setlength{\floatsep}{0.1cm}
\begin{algorithm}[t]
  \caption{DSR ($\textsc{Response}$)}
  \label{alg:dsr-response}
  {\small{
  \begin{algorithmic}[1]
        \Require {timestamp $t$, privacy budget $\epsilon$, data domain $\Omega$ with size $d$, data structure $\mathcal{HG}$, number of the received data $num$.}
        \Ensure {$ResponseList$}
        \State {$p_1= \frac{e^{\epsilon}}{e^{\epsilon} + d - 1}$, $q_1= \frac{1}{e^{\epsilon} + d - 1}$, $p_2= \frac{e^{\epsilon}}{e^{\epsilon} + k}$, $q_2= \frac{1}{e^{\epsilon} + k}$}
        \If {receive a Top-$k$ query}
        \State {$\mathcal{HG}\leftarrow$ $\textsf{DSR\_FinalDebias}(\mathcal{HG},\ p_1,\ q_1,\ p_2,\ q_2)$}
        \For {each $\mathcal{HG}[j]\in \mathcal{HG}$}
        \State{$ResponseList[j].ID\leftarrow\mathcal{HG}[j].ID$}
        \State {$ResponseList[j].C\leftarrow\mathcal{HG}[j].C$}
        \EndFor
        \EndIf\\
        \Return {$ResponseList$}
  \end{algorithmic}}}
\end{algorithm}


The \textsc{Insert} and \textsc{Response} algorithms are shown in Algorithms \ref{alg:dsr-insert} and \ref{alg:dsr-response}, respectively. 
Due to the switch between two mechanisms with different parameters in the \textsc{Randomize} algorithm, a complex debiasing process is initiated during the insertion and response phases.
Each switch between mechanisms necessitates debiasing of all the counts of KV pairs stored in $\mathcal{HG}$ using the debiasing formula of the current mechanism.
To prevent redundant debiasing of cumulative counts, it is imperative to multiply all the counts by the denominator of the debiasing formula of the new mechanism.
For the sake of readability, the debiasing functions $\textsf{DSR\_Insert}$ in the \textsc{Insert} algorithm and $\textsf{DSR\_FinalDebias}$ in the \textsc{Response} algorithm are deferred to Appendix \ref{app:func_dsr_insert}.

\vspace{0.05in}

\noindent\textbf{Theoretical Analysis.}
Theoretical analysis demonstrates that the error bound of DSR in the worst-case scenario aligns with that of BGR, as illustrated in Theorem \ref{the:LdpPD}. This can be attributed to the frequent replacement of the weakest KV pair for certain data distributions, compelling users to randomize data over the entire data domain for the majority of instances. 
However, DSR's improvement over BGR is expected to be more substantial for datasets exhibiting a more concentrated data distribution.

\subsection{BDR (Budget-Division Randomization)}
\label{sec:advldppd}

We present a novel scheme, BDR, that further enhances accuracy beyond DSR. Although DSR demonstrates improvement over BGR, it still predominantly randomizes data similarly to BGR when there are frequent changes to the items in $\mathcal{HG}$. Additionally, the complexity of debiasing is increased due to the transition between two randomization mechanisms with distinct parameters. 
Since the current cold value cannot be randomized and sent repeatedly, resulting in the waste of the privacy budget while awaiting new cold items.
To address this problem, we designed a budget-division-based scheme (BDR) that efficiently avoids switching between different randomization mechanisms and mixing randomized data from different output data domains.
Besides, we observe that the hot items stored by $\mathcal{HG}$ after initialization may not be true hot items.
Through adjustments in the allocation of the privacy budget, BDR reduces the impact of the initial $\mathcal{HG}$ on the final result, with the probability of $\frac{k}{e^\epsilon+k}$ that any other item be randomized to the current "hot items".

\begin{algorithm}[t]
  \caption{BDR ($\textsc{Randomize}$)}
  \label{alg:advLdpPD-randomize}
  {\small{
  \begin{algorithmic}[1]
        \Require {timestamp $t$, privacy budget $\epsilon$, data domain $\Omega$ with size $d$, data structure $\mathcal{HG}$.}
        \Ensure {$r_i^t$}
        \State {Divide $\epsilon$ into $\epsilon_1$ and $\epsilon_2$, where $\epsilon_1+\epsilon_2=\epsilon$;}
        \State {Obtain the current raw data $v_i^t$;}
        \State {$Flag\leftarrow \mathcal{M}_{judge}(v_i^t, \epsilon_1)$} \Comment {Determine whether $v_i^t$ is hot or cold}
        \If {Flag==1} \Comment {$v_i^t$ is determined as hot}
        \State {$r_i^t\leftarrow\mathcal{M}_{hot}(v_i^t, \epsilon_2)$}
        \ElsIf {$Flag==0$ and $\mathcal{HG}[k].C\le 1$} \Comment {$v_i^t$ is determined as cold and an item in $\mathcal{HG}$ is about to be evicted}
        \State {$r_i^t\leftarrow\mathcal{M}_{cold}(v_i^t, \epsilon_2)$}
        \Else
        \State {$r_i^t=\bot$}
        \EndIf\\
        \Return {$r_i^t$}
  \end{algorithmic}}}
\end{algorithm}


\begin{algorithm}[t]
  \caption{BDR ($\textsc{Randomize}$-$\mathcal{M}_{judge}$)}
  \label{alg:advLdpPD-randomize1}
  {\small{
  \begin{algorithmic}[1]
        \Require {raw data $v_i^t$, privacy budget $\epsilon_1$, data structure $\mathcal{HG}$}
        \Ensure {$Flag$}
        \State {Let $b\leftarrow Ber(\frac{e^{\epsilon_1}}{e^{\epsilon_1}+1})$}
        \If {$b==1$}
        \If {$v_i^t\in \mathcal{HG}$}
        \State{$Flag=1$}
        \Else
        \State {$Flag=0$}
        \EndIf
        \Else
        \If {$v_i^t\in \mathcal{HG}$}
        \State{$Flag=0$}
        \Else
        \State {$Flag=1$}
        \EndIf
        \EndIf\\
        \Return {$Flag$}
  \end{algorithmic}}}
\end{algorithm}


We divide the privacy budget into two parts and run three sub-randomization mechanisms $\mathcal{M}_{judge}$, $\mathcal{M}_{hot}$, and $\mathcal{M}_{cold}$.
Specifically, the $\mathcal{M}_{judge}$ mechanism is used to randomize whether the data is a hot item.
If the $\mathcal{M}_{judge}$ mechanism determines that the data is a hot item, the $\mathcal{M}_{hot}$ mechanism is used to randomize the data in the data domain covered by items recorded in $\mathcal{HG}$.  
The $\mathcal{M}_{cold}$ mechanism randomizes the items determined to be cold by the $\mathcal{M}_{judge}$ mechanism when an item in $\mathcal{HG}$ is about to be evicted.
We show the overall flow of the \textsc{Randomize} algorithm in Algorithm \ref{alg:advLdpPD-randomize}, and the $\mathcal{M}_{judge}$, $\mathcal{M}_{hot}$, and $\mathcal{M}_{cold}$ mechanisms in Algorithm \ref{alg:advLdpPD-randomize1}, Algorithm \ref{alg:advLdpPD-randomize2}, and Algorithm \ref{alg:advLdpPD-randomize3}, respectively.

\begin{algorithm}[t]
  \caption{BDR ($\textsc{Randomize}$-$\mathcal{M}_{hot}$)}
  \label{alg:advLdpPD-randomize2}
  {\small{
  \begin{algorithmic}[1]
        \Require {raw data $v_i^t$, $Flag$, privacy budget $\epsilon_2$, data structure $\mathcal{HG}$.}
        \Ensure {$r_i^t$}
        \State {Let $b\leftarrow Ber(\frac{e^{\epsilon_2}}{e^{\epsilon_2}+k-1})$}
        \If {$v_i^t\in\mathcal{HG}$}
        \If {$b==1$}
        \State {$r_i^t=v_i^t$}
        \Else
        \State {$r_i^t=v'$, where $v'\in \mathcal{HG}$ and $v'\neq v_i^t$}
        \EndIf
        \Else
        \State {$r_i^t=v'$, where $v'$ is uniform random sampled from $\mathcal{HG}$}
        \EndIf\\
        \Return {$r_i^t$}
  \end{algorithmic}}}
\end{algorithm}

\begin{algorithm}[t]
  \caption{BDR ($\textsc{Randomize}$-$\mathcal{M}_{cold}$)}
  \label{alg:advLdpPD-randomize3}
  {\small{
  \begin{algorithmic}[1]
        \Require {raw data $v_i^t$, privacy budget $\epsilon_2$, data domain $\Omega$ with size $d$, data structure $\mathcal{HG}$.}
        \Ensure {$r_i^t$}
        \State {Let $b\leftarrow Ber(\frac{e^{\epsilon_2}}{e^{\epsilon_2}+d-k-1})$}
        \If {$v_i^t\notin\mathcal{HG}$}
        \If {$b==1$}
        \State {$r_i^t=v_i^t$}
        \Else
        \State {$r_i^t=v'$, where $v'\in \Omega/\mathcal{HG}$ and $v'\neq v_i^t$}
        \EndIf
        \Else
        \State {$r_i^t=v'$, where $v'$ is uniform random sampled from $\Omega/\mathcal{HG}$}
        \EndIf\\
        \Return {$r_i^t$}
  \end{algorithmic}}}
\end{algorithm}

\medskip
\noindent\textbf{Algorithms.}
At timestamp $t$, the user obtains the current raw data $v_i^t$, which is a hot item or a cold item.
Note that the server can write the currently recorded hot items to a bulletin board in real time or the users can obtain the set of hot items from the \emph{response module} at any time.
Therefore, users can always know the current hot items and cold items when randomizing their data.
Firstly, the user randomizes whether $v_i^t$ is a hot item using the $\mathcal{M}_{judge}$ mechanism, which is a binary flip.
The error introduced by the $\mathcal{M}_{judge}$ mechanism is independent of the size of the data domain.
If the $\mathcal{M}_{judge}$ mechanism determines that $v_i^t$ is a hot item, then $v_i^t$ needs to be randomized on the data domain covered by items recorded in $\mathcal{HG}$ with the $\mathcal{M}_{hot}$ mechanism.
If $v_i^t$ is a hot item, $\mathcal{M}_{hot}$ mechanism randomizes it in the data domain covered by items recorded in $\mathcal{HG}$ as the general GRR.
If $v_i^t$ is actually a cold item, $\mathcal{M}_{hot}$ mechanism uniformly and randomly maps it to any item contained in $\mathcal{HG}$.
Otherwise, the user sends ``$\bot$'' to the server if the $\mathcal{M}_{judge}$ mechanism determines that $v_i^t$ is a cold item.

We consider a special case where the $\mathcal{M}_{judge}$ mechanism determines that $v_i^t$ is a cold item, but the count of the weakest KV pair in $\mathcal{HG}$ is reduced to $0$ by the \emph{ED strategy}.
Then the server would need this cold value to replace the item in $\mathcal{HG}$. 
Therefore, we provide the $\mathcal{M}_{cold}$ mechanism, similar to the $\mathcal{M}_{hot}$ mechanism, randomizing the data in the data domain covered by the cold items.
When the user observes that the count of the weakest KV pair in $\mathcal{HG}$ is equal to or smaller than $1$, the user uses $\mathcal{M}_{cold}$ mechanism to randomize $v_i^t$ and then sends it to the server when $v_i^t$ is determined to be cold.
The $\mathcal{HG}$ has a high probability of replacing the weakest item with a cold item in this case. 
Note that the privacy budget consumed by $\mathcal{M}_{cold}$ is the remaining budget $\epsilon_2$ at timestamp $t$, and the total privacy budget for $v_i^t$ is still limited to $\epsilon$. 
Figure \ref{fig:advLdpPD-examp} shows an example at $6$ timestamps to illustrate the randomization process.

Next, we discuss how the \emph{response module} on the server debiases the counts of hot items stored in $\mathcal{HG}$.
Denote $p_1$ as the probability $\frac{e^{\epsilon_1}}{e^{\epsilon_1}+1}$, $q_1$ as the probability $\frac{1}{e^{\epsilon_1}+1}$, $p_2$ as the probability $\frac{e^{\epsilon_2}}{e^{\epsilon_2}+k-1}$, $q_2$ as the probability $\frac{1}{e^{\epsilon_2}+k-1}$.
Let $num$ denote the total number of data received by the server from the beginning of the statistics to the current timestamp, and $\gamma_h$ denote the proportion of hot items.
Let $\bar f_v$ be the noisy recorded count of item $v$, then the debiased estimation result $\tilde f_v$ is calculated as
\begin{equation}
  \label{eq:adv_debias}
  \begin{aligned}
    &\tilde f_v=\frac{\bar f_v-\gamma_h\cdot num(p_1q_2-q_1/k)-num\cdot q_1/k}{p_1(p_2-q_2)}
  \end{aligned}
\end{equation}
Here, $\gamma_h$ can be obtained from the warm-up round or the prior knowledge of data distribution, which is discussed in detail in Section \ref{subsec:discuss}.
We show the details of the \textsc{Response} algorithm in Algorithm \ref{alg:advLdpPD-response}.
Besides, we omit the details of the \textsc{Insert} algorithm here since it is the same as that of BGR shown in Algorithm \ref{alg:LdpPD}.

\begin{algorithm}[t]
  \caption{BDR ($\textsc{Response}$)}
  \label{alg:advLdpPD-response}
  {\small{
  \begin{algorithmic}[1]
        \Require {timestamp $t$, privacy budget $\epsilon_1$, $\epsilon_2$, data domain $\Omega$ with size $d$, data structure $\mathcal{HG}$, number of the received data $num$.}
        \Ensure {$ResponseList$}
        \State {$p_1=\frac{e^{\epsilon_1}}{e^{\epsilon_1}+1}$, $p_2= \frac{e^{\epsilon_2}}{e^{\epsilon_2} + k - 1}$, $q_1=\frac{1}{e^{\epsilon_1}+1}$, $q_2= \frac{1}{e^{\epsilon_2} + k - 1}$}
        \If {receive a Top-$k$ query}
        \For {each $\mathcal{HG}[j]\in \mathcal{HG}$}
        \State{$ResponseList[j].ID\leftarrow\mathcal{HG}[j].ID$}
        \State {$ResponseList[j].C\leftarrow \mathcal{HG}[j].C-num\cdot(\gamma_h p_1q_2+(1-\gamma_h)q_1/k)/(p_1(p_2-q_2))$}
        \EndFor
        \EndIf\\
        \Return {$ResponseList$}
  \end{algorithmic}}}
\end{algorithm}

\begin{figure}[t]
	\centering
	\includegraphics[width=0.4\textwidth]{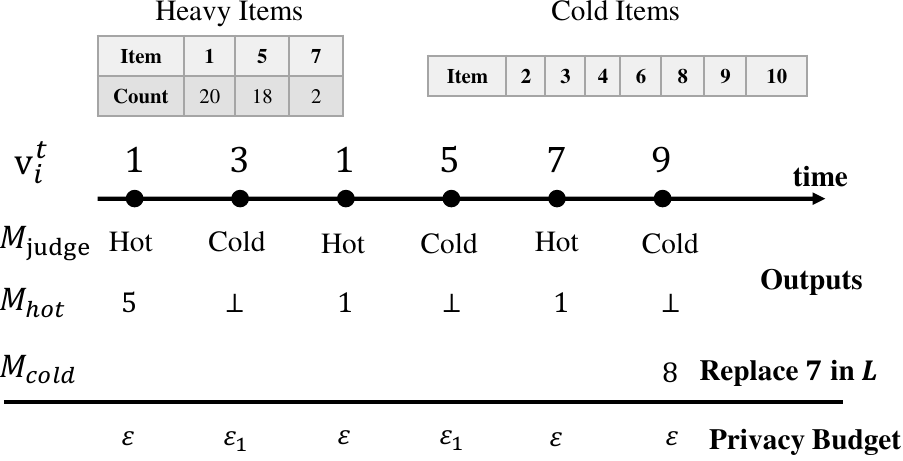}
	\vspace{-0.05in}
	\caption{An example of BDR.}
	\label{fig:advLdpPD-examp}
\end{figure}

\medskip
\noindent\textbf{Theoretical Analysis.} We show that BDR satisfies $\epsilon$-LDP as below.
\begin{theorem}
  \label{the:advLdpPD-priv}
  BDR satisfies $\epsilon$-LDP.
\end{theorem}

\vspace{-0.1in}

\begin{proof}
\revision{Firstly, $\mathcal{M}_{judge}$ satisfies $\epsilon$-LDP since $p_1/q_1=e^{\epsilon_1}$. Secondly, $M_{hot}$ satisfies $\epsilon_2$-LDP since $p_2k\le p_2/q_2=e^{\epsilon_2}$. Similarly, $M_{cold}$ also satisfies $\epsilon_2$-LDP. Therefore, BDR satisfies $(\epsilon_1+\epsilon_2)$-LDP.
The detailed proof is deferred to Appendix \ref{sec:Adv_privacy}.}
\end{proof}

\vspace{-0.05in}

Then we show the error bound of BDR in Theorem \ref{the:AdvLdpPD}.

\vspace{-0.05in}

\begin{theorem}
  \label{the:AdvLdpPD}
  Given a stream prefix $\hat S_t$ with $t$ items randomized by BDR satisfying $\epsilon$-LDP and there is a data structure $\mathcal{HG}$ to store the Top-$k$ items.
  Let $v_i$ be the $i^{\text{th}}$ hottest item, $f_i$ be the real frequency of $v_i$, $\tilde f_i$ be the final estimated frequency of $v_i$.
  We have
  \begin{align*}
    Pr[f_i-\tilde{f_i}&\le(3\sqrt{\frac{t\log(3/\beta)}{2}}+\alpha t)\cdot\frac{(e^{\epsilon_1}+1)(e^{\epsilon_2}+k-1)}{e^{\epsilon_1}(e^{\epsilon_2}-1)}]
    \end{align*}
    \begin{align*}
    &\ge(1-\beta)(1-\frac{1}{2\alpha}(1-\sqrt{1-\frac{4P_{weak}E(V)}{b-1}}))
  \end{align*}
  where $P_{weak}=\frac{(i-1)!(d-k)!}{(d-1)!(i-k)!}$, $E(V)=\sum_{j=i+1}^d f_j$, $\alpha$ and $\beta$ are small positive numbers with $\alpha,\beta\in(0,1)$.
\end{theorem}
\begin{proof}
\revision{The approach of the proof is similar to that of Theorem \ref{the:LdpPD}, the error bound of final debiased frequency $\tilde f_i$ compared to $f_i$ can be obtained by combining the error bounds of $\hat f_i - \bar f_i$ and $f_i-\frac{\hat f_i-N_hp_1q_2-(t-N_h)\cdot\frac{q_1}{k}}{p_1(p_2-q_2)}$, where $N_h$ is the number of hot items and $N_h\le t$.
The detailed proof is deferred to Appendix \ref{sec:Adv_the}.}
\end{proof}
The result of Theorem \ref{the:AdvLdpPD} shows that BDR significantly reduces the impact of the large data domain on the accuracy of the statistical results compared to BGR and DSR (Theorem \ref{the:LdpPD}).

\subsection{CNR (Cold-Nomination Randomization)}
\label{sec:cnr}
In BDR, we find that the privacy budget $\epsilon_2$ is unexploited when the data is determined to be a cold item and there is no item in $\mathcal{HG}$ that is about to be evicted, which can be observed in Figure \ref{fig:advLdpPD-examp}.
Besides, there is a light part in the original data structure of $\mathcal{HG}$ used to store the counts of cold items (see Figure \ref{fig:heavyguardsys}).
The length of this part $\lambda_l$ is set to $0$ in BGR, DSR, and BDR. 
Driven by these observations, we propose a new scheme CNR, which uses these two idle resources to further improve the accuracy over BDR.

\begin{algorithm}[t]
  \caption{CNR ($\textsc{Randomize}$)}
  \label{alg:cnr-randomize}
  {\small{
  \begin{algorithmic}[1]
        \Require {timestamp $t$, privacy budget $\epsilon$, data domain $\Omega$ with size $d$, data structure $\mathcal{HG}$.}
        \Ensure {$r_i^t$}
        \State {Divide $\epsilon$ into $\epsilon_1$ and $\epsilon_2$, where $\epsilon_1+\epsilon_2=\epsilon$;}
        \State {Obtain the current raw data $v_i^t$;}
        \State {$Flag\leftarrow \mathcal{M}_{judge}(v_i^t, \epsilon_1)$} \Comment {Determine whether $v_i^t$ is hot or cold}
        \If {Flag==1} \Comment {$v_i^t$ is determined as hot}
        \State {$r_i^t\leftarrow\mathcal{M}_{hot}(v_i^t, \epsilon_2)$}
        \Else \Comment {$v_i^t$ is determined as cold}
        \State {$r_i^t\leftarrow\mathcal{M}_{cold}(v_i^t, \epsilon_2)$}
        \EndIf\\
        \Return {$r_i^t$}
  \end{algorithmic}}}
\end{algorithm}

\setlength{\floatsep}{0.1cm}
\setlength{\textfloatsep}{0.1cm}
\begin{algorithm}[t]
  \caption{CNR ($\textsc{Insert}$)}
  \label{alg:cnr-insert}
  {\small{
  \begin{algorithmic}[1]
        \Require {timestamp $t$, data domain $\Omega$ with size $d$, data structure $\mathcal{HG}$, number of the received data $num$.}
        \Ensure {Updated $\mathcal{HG}$}
        \State {Receive an incoming data $r_i^t$;}
        \State {$num\leftarrow num+1$}
        \State {Insert $r_i^t$ into Heavy part of $\mathcal{\mathcal{HG}}$ following ED strategy;}
        \If {$r_i^t$ not in Heavy part of $\mathcal{\mathcal{HG}}$}
        \State {Insert $r_i^t$ into Light part of $\mathcal{\mathcal{HG}}$ following ED strategy;}
        \EndIf
        \If {the least count in Heavy part $\mathcal{HG}[k].C\le 0$}
        \State {Replace the weakest KV pair in Heavy part with the king KV pair in Light part, where their counts are set to $1$.}
        \EndIf\\
        \Return {Updated $\mathcal{HG}$}
  \end{algorithmic}}}
\end{algorithm}

\medskip
\noindent\textbf{Algorithms.}
Algorithm \ref{alg:cnr-randomize} shows the \textsc{Randomize} algorithm of CNR, similar to that of BDR. All the data determined as cold items by $\mathcal{M}_{judge}$ mechanism are randomized to specific cold items on the cold domain using $\mathcal{M}_{cold}$ mechanism, rather than calling $\mathcal{M}_{cold}$ mechanism only when there is a hot item to be evicted. Here, $\mathcal{M}_{judge}$, $\mathcal{M}_{hot}$, and $\mathcal{M}_{cold}$ are the same as Algorithms \ref{alg:advLdpPD-randomize1}, \ref{alg:advLdpPD-randomize2}, and \ref{alg:advLdpPD-randomize3} in BDR. When inserting the randomized items into $\mathcal{\mathcal{HG}}$, the cold items that cannot be inserted into the heavy part are inserted into the light part following the ED strategy.
Then the light part helps to provide a more accurate potential hot item to become a new hot item when a value in the heavy part is about to be evicted.
Note that the light part only provides selected cold items, and its count is set to $1$ when a cold item enters the heavy part, just the same as BDR.
Thus, the debiasing formula of the counts in the heavy part is the same as that of the BDR, avoiding debiasing the randomized counts from different output domains like DSR. We show the \textsc{Insert} in Algorithm \ref{alg:cnr-insert}, and the \textsc{Response} is the same as Algorithm \ref{alg:advLdpPD-response}.

\medskip
\noindent\textbf{Theoretical Analysis.}
Firstly, CNR still satisfies $\epsilon$-LDP, and the privacy budget consumed by randomizing data is $\epsilon_1+\epsilon_2=\epsilon$.
Then, the error bound of counts recorded in the heavy part is the same as Theorem \ref{the:AdvLdpPD} shown in BDR, since CNR only provides a better cold item to become a new hot item when there is an item to be evicted.
Note that all theoretical analyses for the error bound of the counts we provide only consider the error of the recorded counts without considering whether the items are true hot items.
Since the accuracy of the hot items tracked by the scheme is influenced by both initial $\mathcal{\mathcal{HG}}$ and data distributions, we evaluate it by conducting a comprehensive evaluation in Section \ref{sec:exp}.

Besides, CNR has no specific requirement for the length of the light part $\lambda_l$, as long as it satisfies $\lambda_l>0$.
The longer light part can provide more accurate new hot items to the heavy part.
The setting of $\lambda_h$ can refer to the original $\mathcal{\mathcal{HG}}$ \cite{DBLP:conf/kdd/0003GZZSL18}, or set a small constant according to the specific requirements.
In our experiments, setting $\lambda_l=5$ for finding Top-$20$ items on a concentrated data distribution can observe a significant improvement for small $\epsilon$.
Furthermore, the counters in the light part of $\mathcal{\mathcal{HG}}$ are tailored for cold items, and the counter size is very small, e.g., 4 bits.
Therefore, CNR does not increase too much additional memory consumption compared to the other schemes and still meets high memory efficiency.

\section{Experimental Evaluation}
\label{sec:exp}
In this section, we design experiments to evaluate our proposed schemes. The evaluation mainly includes four aspects:
(1) the accuracy of the heavy hitters via the proposed schemes; (2) the accuracy achieved by the proposed schemes compared with the baselines;
(3) the impact of the key parameters on the accuracy of the proposed schemes;
(4) the memory size consumed by the proposed schemes compared with the baselines.
Towards these goals, we conduct experiments on both synthetic and real-world datasets, and simulate to collect streaming data from users at continuous timestamps for heavy hitter analysis.
Besides, we introduce different metrics to evaluate the accuracy of the results from three different aspects.

To better guide the application of the schemes in practice, we also conduct supplementary experiments on more datasets and test the computation and communication overheads. Please refer to Appendix \ref{app:supp_exp} for details.
\begin{figure*}[t]
  \centering
  \includegraphics[width=0.9\textwidth]{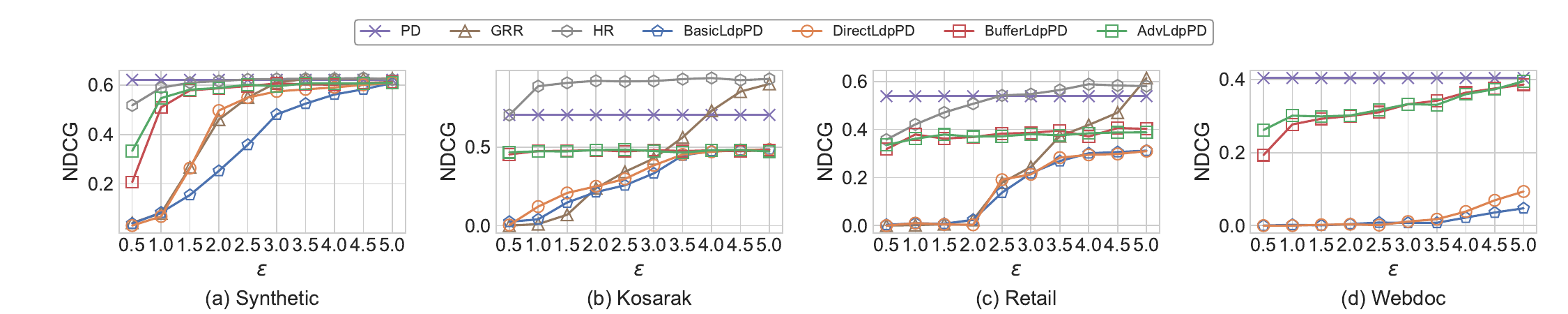}
  \vspace{-2ex}
  \setlength{\belowcaptionskip}{-1ex}
  \caption{Evaluation of NDCG for Top-$20$ on both synthetic and real-world datasets while taking $1\%$ data for warm-up stage.}\vspace{-0.05in}
  \label{fig:ndcg_comp}
\end{figure*}

\begin{figure*}[t]
  \centering
  \includegraphics[width=0.9\textwidth]{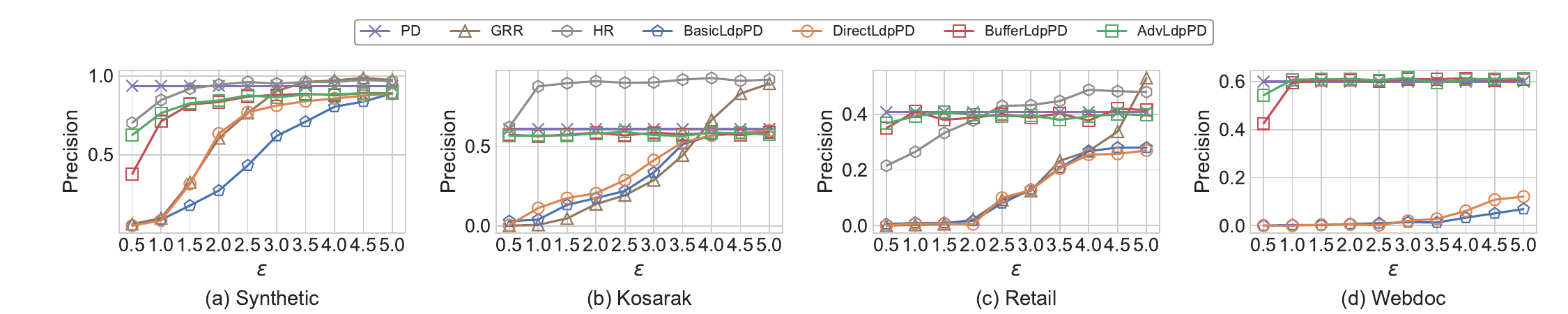}
  \vspace{-2ex}
  \setlength{\belowcaptionskip}{-1ex}
  \caption{Evaluation of Precision for Top-$20$ on both synthetic and real-world datasets while taking $1\%$ data for warm-up stage.}\vspace{-0.05in}
  \label{fig:precision_comp}
\end{figure*}

\begin{figure*}[t]
  \centering
  \includegraphics[width=0.9\textwidth]{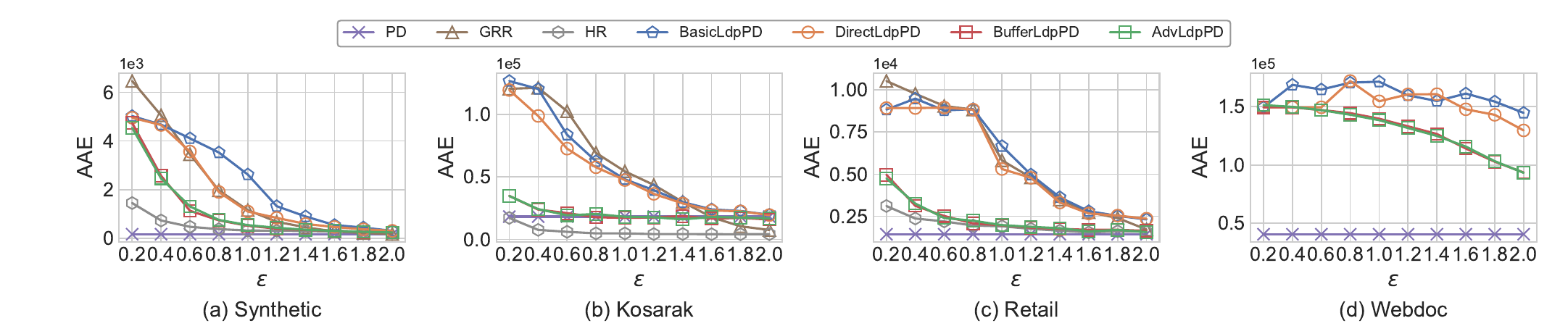}
  \vspace{-2ex}
  \setlength{\belowcaptionskip}{-1ex}
  \caption{Evaluation of AAE for Top-$20$ on both synthetic and real-world datasets while taking $1\%$ data for warm-up stage.}\vspace{-0.1in}
  \label{fig:re_comp}
\end{figure*}




\subsection{Setup}
\noindent\textbf{Datasets.}
We run experiments on the following datasets:
\begin{itemize}
  \item Several synthetic datasets are generated with two different distributions and three domain sizes. One kind of datasets are generated by randomly sampling data from a Normal distribution with variance $\sigma=5$, and others are generated from an Exponential distribution with variance $\sigma=10$. There are $n=100,000$ values in each dataset.
\vspace{0.05in}
  
  \item Retail dataset \cite{retail} contains the retail market basket data from an anonymous Belgian retail store with around $0.9$ million values and $16k$ distinct items.

  \vspace{0.05in}

  \item Kosarak dataset \cite{kosarak} contains the click streams on a Hungarian website, with around $8$ million values and $42k$ URLs.

\vspace{0.05in}

  \item Webdocs dataset \cite{webdocs} is constructed from a collection of web HTML documents, which comprises around 300 million records, and 5.26 million distinct items.
\end{itemize}

\noindent\textbf{Metrics.}
In reality, various applications focus on different aspects of the heavy hitter estimation results.
Therefore, we have to comprehensively evaluate the quality of the results from three aspects:
(1) how accurately that $\mathcal{HG}$ captures the actual heavy hitters;
(2) how accurately that the ordering of the heavy hitters in $\mathcal{HG}$;
(3) how accurately that $\mathcal{HG}$ captures the actual counts of heavy hitters.
We use the following three metrics to cover each aspect:

\emph{Precision.} 
It measures the accuracy of the actual heavy hitters captured by $\mathcal{HG}$.
It is the number of actual heavy hitters divided by the number of all items in $\mathcal{HG}$, as given by
\[Precision=\frac{\#Actual\ heavy\ hitters\ in\ \mathcal{HG}}{\# Heavy\ hitters}.\]

\emph{Normalized Discounted Cumulative Gain (NDCG).}
It measures the ordering quality of the heavy hitters captured by $\mathcal{HG}$, which is a common effectiveness in recommendation systems and other related applications.
$NDCG$ is between $0$ and $1$ for all $k$, and the closer it is to $1$ means the ordering quality of $\mathcal{HG}$ is higher.
The formulas for calculating NDCG is deferred to Appendix \ref{app:ndcg}.

\emph{Average Absolute Error (AAE).}
It measures the error of the counts of the actual Top-$k$ items with their estimated counts recorded in $\mathcal{HG}$, which can be calculated as
\[AAE_k=\frac{1}{k}\sum_{i=1}^k |f_{actual}(v_i)-f_{estimated}(v_i)|.\]
If an actual hot item is not recorded by $\mathcal{HG}$, its AAE is calculated by setting the estimated count as $0$.
For consistent and fair comparisons, we post-process all counts recorded by $\mathcal{HG}$ to $0$ when calculating AAE. 
All results in experiments are averaged with $20$ repeats.

\subsection{Implementation Details}

We fully implemented our schemes and all baselines in Java to provide unified concrete performance comparisons. For all schemes, we separately implement the server and the client side, and the perturb data for communication are serialized to `byte[]'. This makes our implementation easier to be deployed in practice, in which the server and clients would communicate via network channels using byte strings. 
In our experiments, focus more on the effectiveness of our schemes so that we run the server and the client on a single process. 
All experiments are run on Ubuntu 20.04 with 96 Intel Xeon 2.20 GHz CPU and 256 GB RAM. Our source code is available for public request.
Besides, we have some improvements compared with the original implementation in our re-implementation for both LDP mechanisms and original \emph{HeavyGuardian}. More implementation details are deferred to Appendix \ref{app:imple-detail}.
\begin{figure*}[t]
  \centering
  \includegraphics[width=0.9\textwidth]{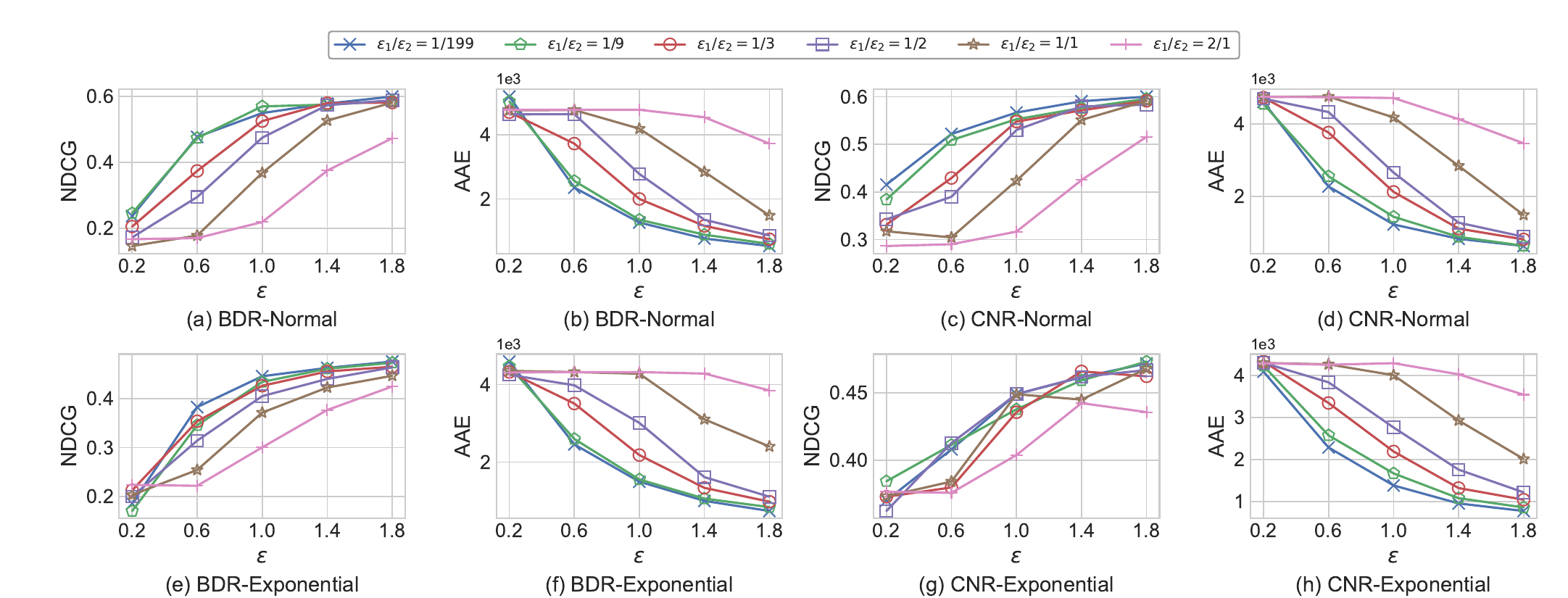}
  \vspace{-2ex}
  \setlength{\belowcaptionskip}{0ex}
  \caption{Accuracy of BDR and CNR vs. different allocations of privacy budget; Conducted on the two synthetic datasets with Normal distribution and Exponential distribution, where the domain size $d=1000$, and taking $1\%$ data for warm-up stage.}\vspace{-0.15in}
  \label{fig:adv_para}
\end{figure*}



\subsection{Analysis of Experimental Results}
\noindent\textbf{Comparison of Accuracy.}
We compare the accuracy of the baseline scheme and three advanced schemes with the non-private \emph{HeavyGuardian} and two LDP mechanisms: Generalized Randomized Response (GRR) and Hadamard Response (HR) (HR performs the best in our evaluation, see Figure \ref{fig:LDP_mech} in Appendix \ref{app:imple-detail}).
We evaluate all schemes on the Synthetic, Retail, Kosarak, and Webdocs datasets. 
The results for three metrics: NDCG, Precision, and AAE are shown in Figure \ref{fig:ndcg_comp}, Figure \ref{fig:precision_comp}, and Figure \ref{fig:re_comp}, respectively.
Since running GRR and HR exceeds the computing or storage capabilities of our server, we only show the results of our schemes on the Webdocs dataset.
In each figure, we vary the privacy budget $\epsilon$ within a range of $[0.5, 5]$.
All schemes involve a warm-up stage for fairness of the comparison.

Firstly, we observe that the accuracy of the proposed schemes BGR, DSR, BDR, and CNR improves sequentially.
The improvement of DSR compared with BGR is more obvious as $\epsilon$ increases, and the advantage of CNR over BDR is more significant as $\epsilon$ decreases.
We think the reason is that when $\epsilon$ is large, i.e., $\epsilon>1$, the randomized hot items are still concentrated and there are fewer times to randomize on the entire domain to provide specific cold items for replacing with the weakest hot items in $\mathcal{HG}$, thus the improvement achieved by DSR is relatively significant.
When $\epsilon$ is small, i.e., $\epsilon<1$, the distribution of the randomized data is relatively uniform, thus the weakest hot item in $\mathcal{HG}$ always need to be replaced.
In this case, the advantage of CNR compared to BDR in providing more potential cold items to enter $\mathcal{HG}$ can be more obvious. 
Besides, we find that these observations are not pronounced on two real-world datasets. 
The reason is that those real-world datasets have large data domains and irregular data distributions.
Therefore, $\mathcal{HG}$ needs to replace the items frequently even if the $\epsilon$ is relatively large.
This means that DSR always randomizes the data on the entire domain in the same way as BGR.
In addition, the large data domain can also lead to low accuracy in the light part of $\mathcal{HG}$.
Then the performance of the CNR is similar to BDR in this case.

Secondly, compared with the non-private \emph{HeavyGuardian} and memory-unlimited LDP randomization mechanisms, BDR and CNR outperform GRR on all datasets in terms of all metrics when $\epsilon<3$.
Moreover, their accuracy on the synthetic dataset is close to HR, and the accuracy on all datasets is close to non-private \emph{HeavyGuardian}.
In all three datasets, BDR and CNR are set to $\epsilon_1/\epsilon_2=0.5$, and their parameter $\gamma_h$ is calculated during the warm-up stage.
We also observe that the performance of BGR and DSR gradually dominates that of GRR as the size of the data domain increases when $\epsilon<3.5$.
However, their accuracy is much lower than that of BDR and CNR when the domain size is extremely large.

Finally, we observe that the NDCG of all schemes is slightly lower than their Precision on all datasets.
The main reason is that NDCG considers the ordering weights of the hit items in addition to whether the true hot items are hit or not.
Besides, the comparison results of all schemes in terms of AAE on all datasets are consistent with the comparison of NDCG and Precision.
The AAE of the statistical results of BGR, DSR, BDR, and CNR decreases in turn.

\vspace{-0.1in}

\begin{figure}[!h]
  \centering
  \includegraphics[width=0.45\textwidth]{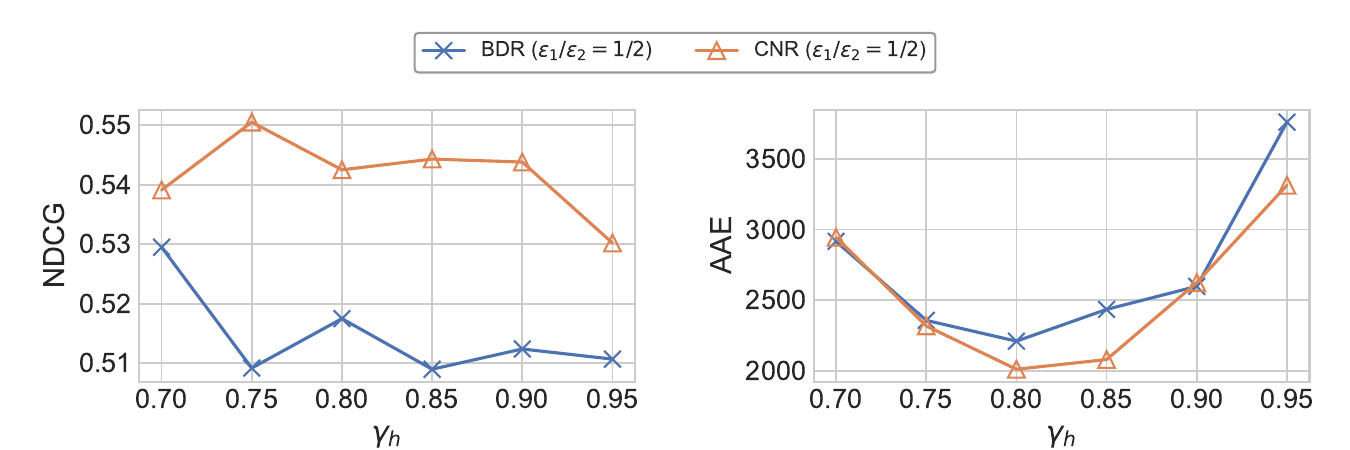}
  \vspace{-2ex}
  \setlength{\belowcaptionskip}{0ex}
  \caption{The impact of parameter $\gamma_h$ on accuracy of BDR and CNR. The evaluation is conducted on the synthetic dataset with normal distribution, where taking $1\%$ data for warm-up stage and $\epsilon=0.6$.}\vspace{-0.2in}
  \label{fig:gammah}
\end{figure}

\begin{figure}[!h]
  \centering
  \includegraphics[width=0.45\textwidth]{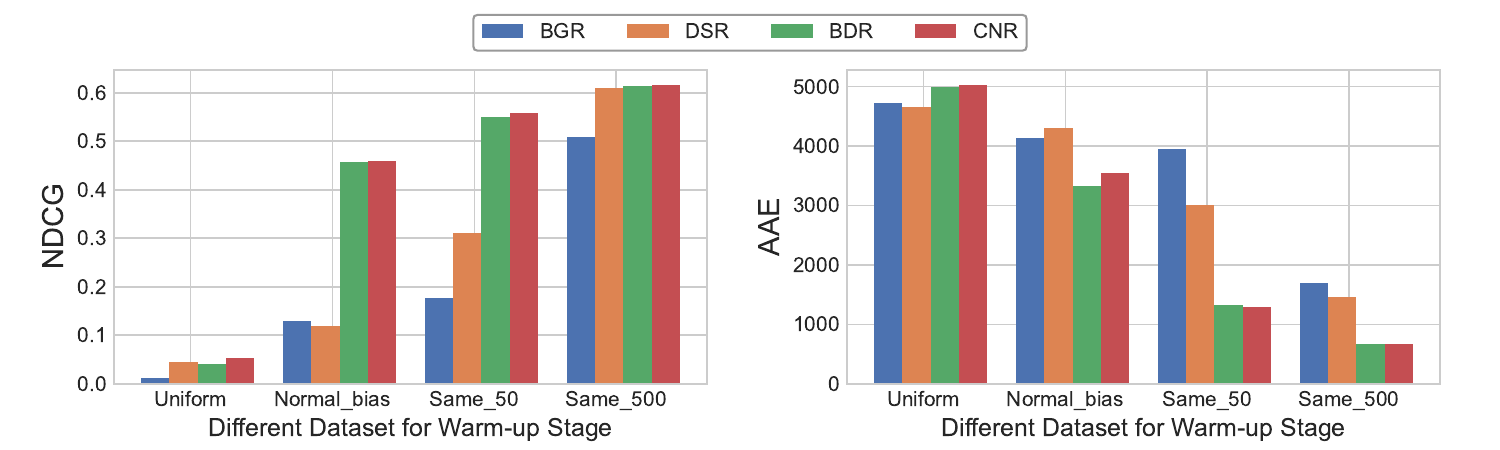}
  \vspace{-2ex}
  \setlength{\belowcaptionskip}{0ex}
  \caption{The impact of the warm-up stage on the accuracy of the proposed schemes for tracking Top-$k$ items on synthetic dataset with normal distribution, where $\epsilon=2$.}
  \label{fig:warmup}
\end{figure}

\medskip
\noindent\textbf{Impact of Key Parameters on the Accuracy.}
We evaluate the impact of key parameters on the accuracy of the proposed schemes by varying them within a certain range.
In order to eliminate the interference of irregular distribution on the evaluation, we conduct experiments on several synthetic datasets.
Due to limited space, we only present the NDCG and Precision of the statistical results on two synthetic datasets with Normal distribution and Exponential distribution.
The results of all metrics on more synthetic datasets with different domain size are deferred to Appendix \ref{app:supp_exp}. 

Firstly, Figure \ref{fig:adv_para} shows the impact of the allocation method of privacy budget $\epsilon$ on the accuracy of BDR and CNR.
We observe that BDR and CNR allocate less privacy budget to $\epsilon_1$ and more privacy budget to $\epsilon_2$ can obtain higher accuracy of the statistical results.
The improvement of NDCG is significant when $\epsilon_1/\epsilon_2$ decreases from $2/1$ to $1/9$, and the increase slows down after $\epsilon_1/\epsilon_2$ is less than $1/9$.
We think the reason is that a hot item recorded in $\mathcal{HG}$ is randomized to a cold item with a greater probability when $\epsilon_1$ is small, and the number of data that is a hot item is larger than data that is a cold item, which leads to the items in $\mathcal{HG}$ are easier to be evicted.
Meanwhile, increasing $\epsilon_2$ can improve the correctness of the orders of the items recorded in $\mathcal{HG}$.
Therefore, reducing $\epsilon_1/\epsilon_2$ can increase the probability that the hot items with a small count are replaced by other cold items, so that the real hot items can occupy the $\mathcal{HG}$ faster.
This is also consistent with the experimental results in original $\mathcal{HG}$ (Figure 4(a), \cite{DBLP:conf/kdd/0003GZZSL18}).
The accuracy of the result increases when the parameter $b$ is reduced to make it easier for the new item to enter $\mathcal{HG}$, but the improvement becomes no longer obvious when $b$ is reduced to a certain small value.
Therefore, we recommend setting $\epsilon_1/\epsilon_2=1/9$ to get near-optimal accuracy in the actual deployment of BDR and CNR.
We also conduct the evaluations on synthetic datasets with different domain sizes, and obtain the consistent observations with the above.
The results are shown in Figure \ref{fig:adv_ndcg}-Figure \ref{fig:buff_re} in Appendix \ref{app:supp_exp}.
Moreover, we find that increasing the domain size has some impact on the accuracy of the schemes, but we can still improve the accuracy by adjusting the privacy budget allocation.

Then Figure \ref{fig:gammah} shows the impact of the parameter $\gamma_h$ on the accuracy of the BDR and CNR. We calculate the exact $\gamma_h\simeq 0.92$.
As a debiasing parameter, $\gamma_h$ directly affects the counts of the statistical result, so the impact of $\gamma_h$ can be clearly observed from the AAE of the result.
However, the indirect impact on NDCG is not obvious, the lines in the figure are fluctuating.
An interesting phenomenon can be observed from the AAE of the results.
More accurate $\gamma_h$ does not necessarily give more accurate count of the result.
The reason is that the ED strategy continuously reduces the counts of the weakest item with a certain probability, which causes the statistical results to be underestimated.
According to debiasing Equation \ref{eq:adv_debias}, reducing $\gamma_h$ can cause the debiased result to be over-estimated, thereby offsetting part of the bias introduced by the ED strategy. 

\begin{table}[t]
  \small
  \centering
  \caption{Comparison of memory size (KB) consumed by schemes on different datasets.}\vspace{-0.1in}
  \label{tab:memory}
  \begin{tabular}{ccccc}
      \hline
       \diagbox{Scheme}{Dataset} & Synthesize & Kosarak & Webdocs\\
       \hline
       \emph{HeavyGuardian} & 2.40 & 2.40 & 2.41\\
       \hline
       GRR & 153.33 & 6154.99 & -\\
       \hline
       HR & 153.41 & 6249.77 & -\\
       \hline
       \textbf{BGR} & \textbf{2.66} & \textbf{2.66} & \textbf{2.67}\\
       \hline
       \textbf{DSR} & \textbf{2.73} & \textbf{2.73} & \textbf{2.75}\\
       \hline
       \textbf{BDR} & \textbf{2.69} & \textbf{2.68} & \textbf{2.70}\\
       \hline
       \textbf{CNR} & \textbf{3.43} & \textbf{3.09} & \textbf{3.43}\\
      \hline
  \end{tabular}
\end{table}

Finally, Figure \ref{fig:warmup} shows the impact of the warm-up stage on the accuracy of the baseline BGR and the proposed three schemes. We compared their accuracy using five different datasets for the warm-up stage.
The five datasets include a uniformly random dataset with the size of $50$, a dataset with the size of $50$ and distribution skewed from the true normal distribution, and two datasets with the true normal distribution with sizes of $50$ and $500$.
We can observe that their accuracy increases as the distribution of the dataset used in the warm-up stage approaches the true distribution and as the size of the dataset increases.
Specifically, BDR and CNR set $\epsilon_1/\epsilon_2=0.5$, and they are least affected by the warm-up stage among all schemes.
We think the reason is that the current cold items in BDR and CNR are easier to enter $\mathcal{HG}$ to become new hot items, which reduces the impact of the accuracy of the initial $\mathcal{HG}$ on the final statistical result.
Similar observations can be obtained on the real-world datasets, and the results are shown in Figure \ref{fig:warmup_kosarak} in Appendix \ref{app:supp_exp}.

\medskip
\noindent\textbf{Comparison of Memory Consumption.}
We then evaluate the total memory size consumed by all schemes when tracking Top-$20$ heavy hitters on the four different datasets.
We present the results in Table \ref{tab:memory}.
The proposed schemes show a significant advantage in memory consumption when the data domain is large, such as in the Kosarak and Webdocs datasets.
We can observe that the memory size consumed by the GRR and HR increases linearly as the domain size $d$ increases. 
In contrast, the memory consumed by all space-saving schemes is only related to the number of tracked heavy hitters $k$, and $k$ is usually much smaller than $d$.
Note that the GRR and HR do not have memory consumption available for the Webdocs dataset since the computation or memory requirements of these schemes exceeded the capacity of the server used for testing.
Additionally, we conduct tests to evaluate the computation and communication overhead of all schemes. The detailed results of these tests can be found in Appendix \ref{app:supp_exp}.

\section{Discussion}
\label{subsec:discuss}
In this section, we supplementally discuss more details about the practical implementation and the potential extension of these schemes.

\subsection{System Parameters \& Implementations}
\noindent\textbf{Warm-up Stage.}
In practice, our framework along with the theoretical designs and analyses are applied to a steady state where the $\mathcal{HG}$ are filled during previous timestamps, rather than dealing with the cold-start scenario where $\mathcal{HG}$ is empty. Therefore, to simulate such a steady state where $\mathcal{HG}$ is properly warm-started, we consider all the proposed schemes include a warm-up round at the beginning of the statistics.
Note that although CNR needs to use the light part in $\mathcal{\mathcal{HG}}$ structure, only the heavy part should be filled in the warm-up stage like other schemes. 
The data for the warm-up stage can be a priori dataset stored on the server or data voluntarily contributed by users in the first round of the statistics.
There is no specific requirement for the number of data in the warm-up stage.
The only requirement is that the data should at least be able to fill the $\mathcal{HG}$.
Besides, the closer the distribution of the priori dataset to the real data distribution, or the larger the number of data that users voluntarily contribute, the higher the accuracy of $\mathcal{HG}$ in the subsequent statistics.

\medskip
\noindent\textbf{Parameter $\gamma_h$ in BDR and CNR.}
The debiasing formula for both BDR and CNR contains a parameter $\lambda_h$, which is the proportion of data that is the hot item in the stream.
The server actually does not know the specific value of $\gamma_h$, but it can be theoretically calculated based on prior knowledge about the data distribution.
If the server has no prior knowledge about the data distribution, $\gamma_h$ can also be statistically obtained from the initial $\mathcal{HG}$ after the warm-up stage.
Certainly, $\gamma_h$ obtained by the above two methods both inevitably introduce additional errors to the estimated results, and the impact is evaluated in the experiments.
However, the current design of schemes cannot avoid it, and we leave it for future work.

\medskip
\noindent\textbf{Privacy Parameters $\epsilon_1,\ \epsilon_2$ in BDR and CNR.}
\revision{
Next, we analyze how to split the privacy budget $\epsilon$ into $\epsilon_1$ and $\epsilon_2$ in BDR and CNR, based on insights from our theoretical and experimental results. }

\revision{
Our theoretical analysis in Theorem \ref{the:AdvLdpPD}, provides an error bound for estimating the count of hot items in BDR, which is equally applicable to CNR. 
It shows that allocating a larger portion of the privacy budget to $\epsilon_2$ leads to a reduced error bound, which is further corroborated by our experimental results in Figure \ref{fig:adv_para}(b)(d)(f)(h). 
In fact, the count error only focuses on the accuracy of counts for hot items already identified by the data structure. This calculation excludes errors coming from the misclassification of hot items due to randomization with $\epsilon_1$. }

\revision{
However, the estimation is complex when considering the impact of $\epsilon_1$ and $\epsilon_2$ on the precision of the data structure $\mathcal{HG}$ in capturing the true hot items. Increasing the privacy budget allocated to $\epsilon_1$ does reduce the probability of determining hot data as cold and simultaneously enhances the probability that currently recorded items remain within $\mathcal{HG}$. 
Nevertheless, this does not necessarily get an improved precision in capturing items within $\mathcal{HG}$. 
The setting of parameter $b$ in $\mathcal{HG}$ \cite{DBLP:conf/kdd/0003GZZSL18} faces the same dilemma.
Increasing $b$ will reduce the probability of the current cold values entering $\mathcal{HG}$, and vice versa.
Multiple factors collaboratively impact the precision of $\mathcal{HG}$ in capturing hot items. 
For instance, when the initial $\mathcal{HG}$ captures inaccurate hot items, a higher probability of eviction among recorded items within $\mathcal{HG}$ can lead to improved precision; if the true hot items are concentrated in the first half of the data stream, a higher probability of retention for items within $\mathcal{HG}$ can result in higher precision.
It can also be observed from the experimental results that the Precision and NDCG of the results on some data streams are not as regular as those of AAE as $\epsilon_1$ and $\epsilon_2$ change, i.e., when hot items are distributed in a more dispersed manner within the Exponential distribution as opposed to the Normal distribution, the NDCG depicted in Figure \ref{fig:adv_para} emphasize that allocating a smaller fraction of $\epsilon_1$ does not confer any discernible advantage.
In \cite{DBLP:conf/kdd/0003GZZSL18}, they provide an empirical value, i.e., $b=1.08$.
Based on our comprehensive evaluations, we suggest setting $\epsilon_1/\epsilon_2=0.5$ in most scenarios can achieve promising accuracy.
}

\medskip
\noindent\textbf{Guidance on Scheme Selection.}
\revision{
In this paper, we introduce three enhanced schemes, each making distinct trade-offs between accuracy, computational overhead, and memory usage.
According to our theoretical and experimental results, we summarize a table, as detailed in Table \ref{tab:comparision_four}, including three advanced designs and a baseline in terms of accuracy, computation overhead, and memory consumption.
From the baseline BGR to DSR, BDR, and CNR, there is a sequential improvement in the accuracy of the results.
Meanwhile, this enhancement comes at the cost of increased computational complexity on the client side or memory consumption on the server side.
In practical deployment, we recommend selecting a scheme based on the specific performance requirements of the task.
}
\begin{table}[!h]
  \small
  \centering
  \caption{Performance comparison of baseline method and proposed schemes.}\vspace{-0.1in}
  \label{tab:comparision_four}
  \begin{tabular}{ccccc}
      \hline
        & BGR & DSR & BDR & CNR\\
       \hline
       Accuracy & 4th & 3rd & 2nd & 1st\\
       \hline
       Computation Overheads & 1st & 2nd & 3rd & 4th\\
       \hline
       Memory Consumption & 1st & 1st & 1st & 2nd\\
       \hline
  \end{tabular}
\end{table}
\subsection{Extensions}
\noindent\textbf{$w$-Event-Level and User-Level Privacy.}
\revision{
While the schemes proposed in this paper offer event-level privacy guarantees, they possess the flexibility to be extended to offer enhanced privacy protection, including $w$-event-level privacy and user-level privacy. Specifically, $w$-event-level privacy ensures $\epsilon$-LDP within any sliding window of size $w$, while user-level privacy guarantees $\epsilon$-LDP for all streaming data contributed by an individual user.}

\revision{
To achieve $w$-event-level privacy and user-level privacy for finite data streams, we could distribute the privacy budget evenly across each timestamp. 
This entails changing the privacy budget used for randomizing each streaming data point from $\epsilon$ to $\epsilon/w$ and $\epsilon/l$, where $l$ represents the length of the finite data stream.
We have to mention that while there are proposed methods for privacy budget allocation that outperform the average allocation approach \cite{farokhi2020temporally, DBLP:journals/pvldb/KellarisPXP14}, applying them to our proposed schemes presents certain challenges.
The primary obstacle lies in the variation of privacy budgets used to randomize each streaming data, which can impede the server to debias the accumulated counts in the heavy list. 
This complication also obstructs the application of the schemes to provide user-level privacy for infinite data streams.
An intuitive approach to address this issue is that the server to independently debias each incoming streaming data point using the privacy budget transmitted by the user concurrently.
However, this approach may introduce increased computational complexity on the server's end and heightened communication complexity for the user. 
We leave this challenge for future research and exploration.
}

\medskip
\noindent\textbf{Other Tasks.}
\revision{
Since the proposed framework HG-LDP focuses on the heavy hitter estimation task, only CNR involves the Light part of the data structure $\mathcal{HG}$ to store the counts of part of cold items.
When CNR extends its functionality to store the counts of all cold items in the Light part as in \cite{DBLP:conf/kdd/0003GZZSL18}, it can also support other tasks supported in \cite{DBLP:conf/kdd/0003GZZSL18}, such as frequency estimation and frequency distribution estimation.
It's essential to note that these tasks, even functionally supported, encounter a challenge related to accuracy when randomizing within large data domains. The new LDP randomization mechanisms in this paper are designed by utilizing the characteristics of the heavy hitter tasks to only ensure the accuracy of hot items.
We intend to delve deeper into this aspect as part of our future research efforts.
}

\section{Related Work}
An extended Related Work is in Appendix \ref{app:related}.

\vspace{0.05in}

\noindent\textbf{Differential Private Data Stream Collection}
The earliest studies in differential privacy for streaming data collection originate from continuous observation of private data \cite{evfimievski2003limiting, bansal2008improved, erlingsson2014rappor, ding2017collecting, dwork2010differential}.
Recent works on differential private data stream collection mainly focus on Centralized Differential Privacy (CDP).
Some works study how to publish the summation of the streaming data privately \cite{perrier2018private, wang2021continuous, DBLP:journals/pvldb/KellarisPXP14, farokhi2020temporally}.
Some works study the release of correlated streaming data \cite{wang2017cts, bao2021cgm} propose a correlated Gaussian noise mechanism.
Some recent works focus on data stream collection with Local Differential Privacy (LDP) \cite{joseph2018local, DBLP:conf/sigmod/RenSYYZX22, wang2021continuous}.

\vspace{0.05in}

\noindent\textbf{Tracking Heavy Hitters in Data stream}
Mining streaming data faces three principal challenges: \emph{volume}, \emph{velocity}, and \emph{volatility} \cite{krempl2014open}.
The existing heavy hitters estimation algorithms in the data stream can be divided into three classes: Counter-based algorithms, Quantile algorithms, and Sketch algorithms \cite{DBLP:journals/pvldb/CormodeH08}. 
Counter-based algorithms track the subset of items in the stream, and they quickly determine whether to record and how to record with each new arrival data \cite{DBLP:conf/vldb/MankuM02, DBLP:conf/icdt/MetwallyAA05, DBLP:conf/kdd/0003GZZSL18, DBLP:conf/sigmod/Zhou0J0YLU18}.
The Quantile algorithms \cite{greenwald2001space, shrivastava2004medians} focus on finding the item which is the smallest item that dominates $\phi n$ items from the data stream.
Sketch algorithms \cite{alon1999space, cormode2005improved, DBLP:conf/kdd/LiLXJ00DZ20, chi2004moment} record items with a data structure, which can be thought of as a linear projection of the input, hash functions are usually used to define that. 
\revision{However, the sketch algorithms involve a large number of hash operations, which cannot meet the timeliness requirements of streaming data.
Besides, all items are recorded and additional information needs to be stored for retrieval, which leads to unnecessary memory consumption \cite{DBLP:journals/pvldb/CormodeH08}}. Our design is based on Counter-based algorithms with an extended setting where streaming data is protected by LDP.

\section{Conclusion}
In this paper, we proposed a framework HG-LDP for tracking the Top-$k$ heavy hitters on data streams at bounded memory expense, while providing rigorous LDP protection.
A baseline and three advanced schemes with new LDP randomization mechanisms are designed under the hood of the framework.
We implement all the proposed schemes and evaluate them on both synthetic and real-world datasets in terms of accuracy and memory consumption.
The experimental results demonstrated that the proposed schemes achieve a satisfactory ``accuracy-privacy-memory efficiency'' tradeoff.
For future work, we will extend the framework to be compatible with more diverse selections of memory-efficiency data structures as well as broader types of statistical tasks to enhance its flexibility. 



\begin{acks}
This work is supported by the National Key Research and Development Program of China under Grant 2021YFB3100300, and the National Natural Science Foundation of China under Grant U20A20178 and 62072395. Weiran Liu is supported in part by the Major Programs of the National Social Science Foundation of China under Grant 22\&ZD147. Yuan Hong is supported in part by the National Science Foundation under Grants CNS-2308730, CNS-2302689, CNS-2319277, CMMI-2326341 and the Cisco Research Award.
\end{acks}
\bibliographystyle{ACM-Reference-Format}
\bibliography{sample-base}


\begin{thebibliography}{64}


\ifx \showCODEN    \undefined \def \showCODEN     #1{\unskip}     \fi
\ifx \showDOI      \undefined \def \showDOI       #1{#1}\fi
\ifx \showISBNx    \undefined \def \showISBNx     #1{\unskip}     \fi
\ifx \showISBNxiii \undefined \def \showISBNxiii  #1{\unskip}     \fi
\ifx \showISSN     \undefined \def \showISSN      #1{\unskip}     \fi
\ifx \showLCCN     \undefined \def \showLCCN      #1{\unskip}     \fi
\ifx \shownote     \undefined \def \shownote      #1{#1}          \fi
\ifx \showarticletitle \undefined \def \showarticletitle #1{#1}   \fi
\ifx \showURL      \undefined \def \showURL       {\relax}        \fi
\providecommand\bibfield[2]{#2}
\providecommand\bibinfo[2]{#2}
\providecommand\natexlab[1]{#1}
\providecommand\showeprint[2][]{arXiv:#2}

\bibitem[kos({[n.\,d.]})]%
        {kosarak}
 \bibinfo{year}{[n.\,d.]}\natexlab{}.
\newblock \bibinfo{title}{Frequent Itemset Mining Dataset Repository, kosarak.}
\newblock \bibinfo{howpublished}{\url{http://fimi.uantwerpen.be/data/}}.
\newblock


\bibitem[ret({[n.\,d.]})]%
        {retail}
 \bibinfo{year}{[n.\,d.]}\natexlab{}.
\newblock \bibinfo{title}{Frequent Itemset Mining Dataset Repository, retail.}
\newblock \bibinfo{howpublished}{\url{http://fimi.uantwerpen.be/data/}}.
\newblock


\bibitem[web({[n.\,d.]})]%
        {webdocs}
 \bibinfo{year}{[n.\,d.]}\natexlab{}.
\newblock \bibinfo{title}{Frequent Itemset Mining Dataset Repository, webdocs.}
\newblock \bibinfo{howpublished}{\url{http://fimi.uantwerpen.be/data/webdocs.dat.gz}}.
\newblock


\bibitem[you({[n.\,d.]})]%
        {youtube}
 \bibinfo{year}{[n.\,d.]}\natexlab{}.
\newblock \bibinfo{title}{How Many Videos Are on YouTube: Exploring the Vast Digital Landscape.}
\newblock \bibinfo{howpublished}{\url{https://www.techpluto.com/how-many-videos-are-on-youtube/How_Many_Videos_Are_There_on_YouTube}}.
\newblock


\bibitem[Acharya and Sun(2019)]%
        {acharya2019communication}
\bibfield{author}{\bibinfo{person}{Jayadev Acharya} {and} \bibinfo{person}{Ziteng Sun}.} \bibinfo{year}{2019}\natexlab{}.
\newblock \showarticletitle{Communication complexity in locally private distribution estimation and heavy hitters}. In \bibinfo{booktitle}{\emph{International Conference on Machine Learning}}. PMLR, \bibinfo{pages}{51--60}.
\newblock


\bibitem[Acharya et~al\mbox{.}(2019)]%
        {acharya2019hadamard}
\bibfield{author}{\bibinfo{person}{Jayadev Acharya}, \bibinfo{person}{Ziteng Sun}, {and} \bibinfo{person}{Huanyu Zhang}.} \bibinfo{year}{2019}\natexlab{}.
\newblock \showarticletitle{Hadamard response: Estimating distributions privately, efficiently, and with little communication}. In \bibinfo{booktitle}{\emph{The 22nd International Conference on Artificial Intelligence and Statistics}}. PMLR, \bibinfo{pages}{1120--1129}.
\newblock


\bibitem[Alon et~al\mbox{.}(1999)]%
        {alon1999space}
\bibfield{author}{\bibinfo{person}{Noga Alon}, \bibinfo{person}{Yossi Matias}, {and} \bibinfo{person}{Mario Szegedy}.} \bibinfo{year}{1999}\natexlab{}.
\newblock \showarticletitle{The space complexity of approximating the frequency moments}.
\newblock \bibinfo{journal}{\emph{Journal of Computer and system sciences}} \bibinfo{volume}{58}, \bibinfo{number}{1} (\bibinfo{year}{1999}), \bibinfo{pages}{137--147}.
\newblock


\bibitem[Bansal et~al\mbox{.}(2008)]%
        {bansal2008improved}
\bibfield{author}{\bibinfo{person}{Nikhil Bansal}, \bibinfo{person}{Don Coppersmith}, {and} \bibinfo{person}{Maxim Sviridenko}.} \bibinfo{year}{2008}\natexlab{}.
\newblock \showarticletitle{Improved approximation algorithms for broadcast scheduling}.
\newblock \bibinfo{journal}{\emph{SIAM J. Comput.}} \bibinfo{volume}{38}, \bibinfo{number}{3} (\bibinfo{year}{2008}), \bibinfo{pages}{1157--1174}.
\newblock


\bibitem[Bao et~al\mbox{.}(2021)]%
        {bao2021cgm}
\bibfield{author}{\bibinfo{person}{Ergute Bao}, \bibinfo{person}{Yin Yang}, \bibinfo{person}{Xiaokui Xiao}, {and} \bibinfo{person}{Bolin Ding}.} \bibinfo{year}{2021}\natexlab{}.
\newblock \showarticletitle{CGM: an enhanced mechanism for streaming data collection with local differential privacy}.
\newblock \bibinfo{journal}{\emph{Proceedings of the VLDB Endowment}} \bibinfo{volume}{14}, \bibinfo{number}{11} (\bibinfo{year}{2021}), \bibinfo{pages}{2258--2270}.
\newblock


\bibitem[Basat et~al\mbox{.}(2022)]%
        {DBLP:journals/ton/BasatEKOVW22}
\bibfield{author}{\bibinfo{person}{Ran~Ben Basat}, \bibinfo{person}{Gil Einziger}, \bibinfo{person}{Isaac Keslassy}, \bibinfo{person}{Ariel Orda}, \bibinfo{person}{Shay Vargaftik}, {and} \bibinfo{person}{Erez Waisbard}.} \bibinfo{year}{2022}\natexlab{}.
\newblock \showarticletitle{Memento: Making Sliding Windows Efficient for Heavy Hitters}.
\newblock \bibinfo{journal}{\emph{{IEEE/ACM} Trans. Netw.}} \bibinfo{volume}{30}, \bibinfo{number}{4} (\bibinfo{year}{2022}), \bibinfo{pages}{1440--1453}.
\newblock
\urldef\tempurl%
\url{https://doi.org/10.1109/TNET.2021.3132385}
\showDOI{\tempurl}


\bibitem[Bassily et~al\mbox{.}(2017)]%
        {DBLP:conf/nips/BassilyNST17}
\bibfield{author}{\bibinfo{person}{Raef Bassily}, \bibinfo{person}{Kobbi Nissim}, \bibinfo{person}{Uri Stemmer}, {and} \bibinfo{person}{Abhradeep~Guha Thakurta}.} \bibinfo{year}{2017}\natexlab{}.
\newblock \showarticletitle{Practical Locally Private Heavy Hitters}. In \bibinfo{booktitle}{\emph{Advances in Neural Information Processing Systems 30: Annual Conference on Neural Information Processing Systems 2017, December 4-9, 2017, Long Beach, CA, {USA}}}, \bibfield{editor}{\bibinfo{person}{Isabelle Guyon}, \bibinfo{person}{Ulrike von Luxburg}, \bibinfo{person}{Samy Bengio}, \bibinfo{person}{Hanna~M. Wallach}, \bibinfo{person}{Rob Fergus}, \bibinfo{person}{S.~V.~N. Vishwanathan}, {and} \bibinfo{person}{Roman Garnett}} (Eds.). \bibinfo{pages}{2288--2296}.
\newblock


\bibitem[Ben{-}Basat et~al\mbox{.}(2016)]%
        {DBLP:conf/infocom/Ben-BasatEFK16}
\bibfield{author}{\bibinfo{person}{Ran Ben{-}Basat}, \bibinfo{person}{Gil Einziger}, \bibinfo{person}{Roy Friedman}, {and} \bibinfo{person}{Yaron Kassner}.} \bibinfo{year}{2016}\natexlab{}.
\newblock \showarticletitle{Heavy hitters in streams and sliding windows}. In \bibinfo{booktitle}{\emph{35th Annual {IEEE} International Conference on Computer Communications, {INFOCOM} 2016, San Francisco, CA, USA, April 10-14, 2016}}. \bibinfo{publisher}{{IEEE}}, \bibinfo{pages}{1--9}.
\newblock
\urldef\tempurl%
\url{https://doi.org/10.1109/INFOCOM.2016.7524364}
\showDOI{\tempurl}


\bibitem[Braverman et~al\mbox{.}(2016)]%
        {DBLP:conf/stoc/BravermanCIW16}
\bibfield{author}{\bibinfo{person}{Vladimir Braverman}, \bibinfo{person}{Stephen~R. Chestnut}, \bibinfo{person}{Nikita Ivkin}, {and} \bibinfo{person}{David~P. Woodruff}.} \bibinfo{year}{2016}\natexlab{}.
\newblock \showarticletitle{Beating CountSketch for heavy hitters in insertion streams}. In \bibinfo{booktitle}{\emph{Proceedings of the 48th Annual {ACM} {SIGACT} Symposium on Theory of Computing, {STOC} 2016, Cambridge, MA, USA, June 18-21, 2016}}, \bibfield{editor}{\bibinfo{person}{Daniel Wichs} {and} \bibinfo{person}{Yishay Mansour}} (Eds.). \bibinfo{publisher}{{ACM}}, \bibinfo{pages}{740--753}.
\newblock
\urldef\tempurl%
\url{https://doi.org/10.1145/2897518.2897558}
\showDOI{\tempurl}


\bibitem[Canonne et~al\mbox{.}(2020)]%
        {DBLP:conf/nips/Canonne0S20}
\bibfield{author}{\bibinfo{person}{Cl{\'{e}}ment~L. Canonne}, \bibinfo{person}{Gautam Kamath}, {and} \bibinfo{person}{Thomas Steinke}.} \bibinfo{year}{2020}\natexlab{}.
\newblock \showarticletitle{The Discrete Gaussian for Differential Privacy}. In \bibinfo{booktitle}{\emph{Advances in Neural Information Processing Systems 33: Annual Conference on Neural Information Processing Systems 2020, NeurIPS 2020, December 6-12, 2020, virtual}}, \bibfield{editor}{\bibinfo{person}{Hugo Larochelle}, \bibinfo{person}{Marc'Aurelio Ranzato}, \bibinfo{person}{Raia Hadsell}, \bibinfo{person}{Maria{-}Florina Balcan}, {and} \bibinfo{person}{Hsuan{-}Tien Lin}} (Eds.).
\newblock


\bibitem[Chan et~al\mbox{.}(2010)]%
        {DBLP:conf/icalp/ChanSS10}
\bibfield{author}{\bibinfo{person}{T.{-}H.~Hubert Chan}, \bibinfo{person}{Elaine Shi}, {and} \bibinfo{person}{Dawn Song}.} \bibinfo{year}{2010}\natexlab{}.
\newblock \showarticletitle{Private and Continual Release of Statistics}. In \bibinfo{booktitle}{\emph{Automata, Languages and Programming, 37th International Colloquium, {ICALP} 2010, Bordeaux, France, July 6-10, 2010, Proceedings, Part {II}}} \emph{(\bibinfo{series}{Lecture Notes in Computer Science}, Vol.~\bibinfo{volume}{6199})}, \bibfield{editor}{\bibinfo{person}{Samson Abramsky}, \bibinfo{person}{Cyril Gavoille}, \bibinfo{person}{Claude Kirchner}, \bibinfo{person}{Friedhelm~Meyer auf~der Heide}, {and} \bibinfo{person}{Paul~G. Spirakis}} (Eds.). \bibinfo{publisher}{Springer}, \bibinfo{pages}{405--417}.
\newblock
\urldef\tempurl%
\url{https://doi.org/10.1007/978-3-642-14162-1\_34}
\showDOI{\tempurl}


\bibitem[Chandola et~al\mbox{.}(2009)]%
        {DBLP:journals/csur/ChandolaBK09}
\bibfield{author}{\bibinfo{person}{Varun Chandola}, \bibinfo{person}{Arindam Banerjee}, {and} \bibinfo{person}{Vipin Kumar}.} \bibinfo{year}{2009}\natexlab{}.
\newblock \showarticletitle{Anomaly detection: {A} survey}.
\newblock \bibinfo{journal}{\emph{{ACM} Comput. Surv.}} \bibinfo{volume}{41}, \bibinfo{number}{3} (\bibinfo{year}{2009}), \bibinfo{pages}{15:1--15:58}.
\newblock
\urldef\tempurl%
\url{https://doi.org/10.1145/1541880.1541882}
\showDOI{\tempurl}


\bibitem[Chen and Cong(2015)]%
        {DBLP:conf/sigmod/ChenC15}
\bibfield{author}{\bibinfo{person}{Lisi Chen} {and} \bibinfo{person}{Gao Cong}.} \bibinfo{year}{2015}\natexlab{}.
\newblock \showarticletitle{Diversity-Aware Top-k Publish/Subscribe for Text Stream}. In \bibinfo{booktitle}{\emph{Proceedings of the 2015 {ACM} {SIGMOD} International Conference on Management of Data, Melbourne, Victoria, Australia, May 31 - June 4, 2015}}, \bibfield{editor}{\bibinfo{person}{Timos~K. Sellis}, \bibinfo{person}{Susan~B. Davidson}, {and} \bibinfo{person}{Zachary~G. Ives}} (Eds.). \bibinfo{publisher}{{ACM}}, \bibinfo{pages}{347--362}.
\newblock
\urldef\tempurl%
\url{https://doi.org/10.1145/2723372.2749451}
\showDOI{\tempurl}


\bibitem[Chen et~al\mbox{.}(2017)]%
        {DBLP:conf/ccs/ChenMHM17}
\bibfield{author}{\bibinfo{person}{Yan Chen}, \bibinfo{person}{Ashwin Machanavajjhala}, \bibinfo{person}{Michael Hay}, {and} \bibinfo{person}{Gerome Miklau}.} \bibinfo{year}{2017}\natexlab{}.
\newblock \showarticletitle{PeGaSus: Data-Adaptive Differentially Private Stream Processing}. In \bibinfo{booktitle}{\emph{Proceedings of the 2017 {ACM} {SIGSAC} Conference on Computer and Communications Security, {CCS} 2017, Dallas, TX, USA, October 30 - November 03, 2017}}, \bibfield{editor}{\bibinfo{person}{Bhavani Thuraisingham}, \bibinfo{person}{David Evans}, \bibinfo{person}{Tal Malkin}, {and} \bibinfo{person}{Dongyan Xu}} (Eds.). \bibinfo{publisher}{{ACM}}, \bibinfo{pages}{1375--1388}.
\newblock
\urldef\tempurl%
\url{https://doi.org/10.1145/3133956.3134102}
\showDOI{\tempurl}


\bibitem[Chi et~al\mbox{.}(2004)]%
        {chi2004moment}
\bibfield{author}{\bibinfo{person}{Yun Chi}, \bibinfo{person}{Haixun Wang}, \bibinfo{person}{Philip~S Yu}, {and} \bibinfo{person}{Richard~R Muntz}.} \bibinfo{year}{2004}\natexlab{}.
\newblock \showarticletitle{Moment: Maintaining closed frequent itemsets over a stream sliding window}. In \bibinfo{booktitle}{\emph{Fourth IEEE International Conference on Data Mining (ICDM'04)}}. IEEE, \bibinfo{pages}{59--66}.
\newblock


\bibitem[Cormode(2011)]%
        {cormode2011sketch}
\bibfield{author}{\bibinfo{person}{Graham Cormode}.} \bibinfo{year}{2011}\natexlab{}.
\newblock \showarticletitle{Sketch techniques for approximate query processing}.
\newblock \bibinfo{journal}{\emph{Foundations and Trends in Databases. NOW publishers}} (\bibinfo{year}{2011}), \bibinfo{pages}{15}.
\newblock


\bibitem[Cormode and Hadjieleftheriou(2008)]%
        {DBLP:journals/pvldb/CormodeH08}
\bibfield{author}{\bibinfo{person}{Graham Cormode} {and} \bibinfo{person}{Marios Hadjieleftheriou}.} \bibinfo{year}{2008}\natexlab{}.
\newblock \showarticletitle{Finding frequent items in data streams}.
\newblock \bibinfo{journal}{\emph{Proc. {VLDB} Endow.}} \bibinfo{volume}{1}, \bibinfo{number}{2} (\bibinfo{year}{2008}), \bibinfo{pages}{1530--1541}.
\newblock
\urldef\tempurl%
\url{https://doi.org/10.14778/1454159.1454225}
\showDOI{\tempurl}


\bibitem[Cormode et~al\mbox{.}(2003)]%
        {DBLP:conf/vldb/CormodeKMS03}
\bibfield{author}{\bibinfo{person}{Graham Cormode}, \bibinfo{person}{Flip Korn}, \bibinfo{person}{S. Muthukrishnan}, {and} \bibinfo{person}{Divesh Srivastava}.} \bibinfo{year}{2003}\natexlab{}.
\newblock \showarticletitle{Finding Hierarchical Heavy Hitters in Data Streams}. In \bibinfo{booktitle}{\emph{Proceedings of 29th International Conference on Very Large Data Bases, {VLDB} 2003, Berlin, Germany, September 9-12, 2003}}, \bibfield{editor}{\bibinfo{person}{Johann~Christoph Freytag}, \bibinfo{person}{Peter~C. Lockemann}, \bibinfo{person}{Serge Abiteboul}, \bibinfo{person}{Michael~J. Carey}, \bibinfo{person}{Patricia~G. Selinger}, {and} \bibinfo{person}{Andreas Heuer}} (Eds.). \bibinfo{publisher}{Morgan Kaufmann}, \bibinfo{pages}{464--475}.
\newblock
\urldef\tempurl%
\url{https://doi.org/10.1016/B978-012722442-8/50048-3}
\showDOI{\tempurl}


\bibitem[Cormode et~al\mbox{.}(2021)]%
        {DBLP:journals/pvldb/CormodeMM21}
\bibfield{author}{\bibinfo{person}{Graham Cormode}, \bibinfo{person}{Samuel Maddock}, {and} \bibinfo{person}{Carsten Maple}.} \bibinfo{year}{2021}\natexlab{}.
\newblock \showarticletitle{Frequency Estimation under Local Differential Privacy}.
\newblock \bibinfo{journal}{\emph{Proc. {VLDB} Endow.}} \bibinfo{volume}{14}, \bibinfo{number}{11} (\bibinfo{year}{2021}), \bibinfo{pages}{2046--2058}.
\newblock
\urldef\tempurl%
\url{https://doi.org/10.14778/3476249.3476261}
\showDOI{\tempurl}


\bibitem[Cormode and Muthukrishnan(2005a)]%
        {cormode2005improved}
\bibfield{author}{\bibinfo{person}{Graham Cormode} {and} \bibinfo{person}{Shan Muthukrishnan}.} \bibinfo{year}{2005}\natexlab{a}.
\newblock \showarticletitle{An improved data stream summary: the count-min sketch and its applications}.
\newblock \bibinfo{journal}{\emph{Journal of Algorithms}} \bibinfo{volume}{55}, \bibinfo{number}{1} (\bibinfo{year}{2005}), \bibinfo{pages}{58--75}.
\newblock


\bibitem[Cormode and Muthukrishnan(2005b)]%
        {DBLP:journals/tods/CormodeM05}
\bibfield{author}{\bibinfo{person}{Graham Cormode} {and} \bibinfo{person}{S. Muthukrishnan}.} \bibinfo{year}{2005}\natexlab{b}.
\newblock \showarticletitle{What's hot and what's not: tracking most frequent items dynamically}.
\newblock \bibinfo{journal}{\emph{{ACM} Trans. Database Syst.}} \bibinfo{volume}{30}, \bibinfo{number}{1} (\bibinfo{year}{2005}), \bibinfo{pages}{249--278}.
\newblock
\urldef\tempurl%
\url{https://doi.org/10.1145/1061318.1061325}
\showDOI{\tempurl}


\bibitem[Das et~al\mbox{.}(2007)]%
        {DBLP:conf/vldb/DasGKS07}
\bibfield{author}{\bibinfo{person}{Gautam Das}, \bibinfo{person}{Dimitrios Gunopulos}, \bibinfo{person}{Nick Koudas}, {and} \bibinfo{person}{Nikos Sarkas}.} \bibinfo{year}{2007}\natexlab{}.
\newblock \showarticletitle{Ad-hoc Top-k Query Answering for Data Streams}. In \bibinfo{booktitle}{\emph{Proceedings of the 33rd International Conference on Very Large Data Bases, University of Vienna, Austria, September 23-27, 2007}}, \bibfield{editor}{\bibinfo{person}{Christoph Koch}, \bibinfo{person}{Johannes Gehrke}, \bibinfo{person}{Minos~N. Garofalakis}, \bibinfo{person}{Divesh Srivastava}, \bibinfo{person}{Karl Aberer}, \bibinfo{person}{Anand Deshpande}, \bibinfo{person}{Daniela Florescu}, \bibinfo{person}{Chee~Yong Chan}, \bibinfo{person}{Venkatesh Ganti}, \bibinfo{person}{Carl{-}Christian Kanne}, \bibinfo{person}{Wolfgang Klas}, {and} \bibinfo{person}{Erich~J. Neuhold}} (Eds.). \bibinfo{publisher}{{ACM}}, \bibinfo{pages}{183--194}.
\newblock


\bibitem[Ding et~al\mbox{.}(2017)]%
        {ding2017collecting}
\bibfield{author}{\bibinfo{person}{Bolin Ding}, \bibinfo{person}{Janardhan Kulkarni}, {and} \bibinfo{person}{Sergey Yekhanin}.} \bibinfo{year}{2017}\natexlab{}.
\newblock \showarticletitle{Collecting telemetry data privately}.
\newblock \bibinfo{journal}{\emph{Advances in Neural Information Processing Systems}}  \bibinfo{volume}{30} (\bibinfo{year}{2017}).
\newblock


\bibitem[Dwork(2006)]%
        {DBLP:conf/icalp/Dwork06}
\bibfield{author}{\bibinfo{person}{Cynthia Dwork}.} \bibinfo{year}{2006}\natexlab{}.
\newblock \showarticletitle{Differential Privacy}. In \bibinfo{booktitle}{\emph{Automata, Languages and Programming, 33rd International Colloquium, {ICALP} 2006, Venice, Italy, July 10-14, 2006, Proceedings, Part {II}}} \emph{(\bibinfo{series}{Lecture Notes in Computer Science}, Vol.~\bibinfo{volume}{4052})}, \bibfield{editor}{\bibinfo{person}{Michele Bugliesi}, \bibinfo{person}{Bart Preneel}, \bibinfo{person}{Vladimiro Sassone}, {and} \bibinfo{person}{Ingo Wegener}} (Eds.). \bibinfo{publisher}{Springer}, \bibinfo{pages}{1--12}.
\newblock
\urldef\tempurl%
\url{https://doi.org/10.1007/11787006\_1}
\showDOI{\tempurl}


\bibitem[Dwork et~al\mbox{.}(2010a)]%
        {DBLP:conf/stoc/DworkNPR10}
\bibfield{author}{\bibinfo{person}{Cynthia Dwork}, \bibinfo{person}{Moni Naor}, \bibinfo{person}{Toniann Pitassi}, {and} \bibinfo{person}{Guy~N. Rothblum}.} \bibinfo{year}{2010}\natexlab{a}.
\newblock \showarticletitle{Differential privacy under continual observation}. In \bibinfo{booktitle}{\emph{Proceedings of the 42nd {ACM} Symposium on Theory of Computing, {STOC} 2010, Cambridge, Massachusetts, USA, 5-8 June 2010}}, \bibfield{editor}{\bibinfo{person}{Leonard~J. Schulman}} (Ed.). \bibinfo{publisher}{{ACM}}, \bibinfo{pages}{715--724}.
\newblock
\urldef\tempurl%
\url{https://doi.org/10.1145/1806689.1806787}
\showDOI{\tempurl}


\bibitem[Dwork et~al\mbox{.}(2010b)]%
        {dwork2010differential}
\bibfield{author}{\bibinfo{person}{Cynthia Dwork}, \bibinfo{person}{Moni Naor}, \bibinfo{person}{Toniann Pitassi}, {and} \bibinfo{person}{Guy~N Rothblum}.} \bibinfo{year}{2010}\natexlab{b}.
\newblock \showarticletitle{Differential privacy under continual observation}. In \bibinfo{booktitle}{\emph{Proceedings of the forty-second ACM symposium on Theory of computing}}. \bibinfo{pages}{715--724}.
\newblock


\bibitem[Dwork and Roth(2014)]%
        {DBLP:journals/fttcs/DworkR14}
\bibfield{author}{\bibinfo{person}{Cynthia Dwork} {and} \bibinfo{person}{Aaron Roth}.} \bibinfo{year}{2014}\natexlab{}.
\newblock \showarticletitle{The Algorithmic Foundations of Differential Privacy}.
\newblock \bibinfo{journal}{\emph{Found. Trends Theor. Comput. Sci.}} \bibinfo{volume}{9}, \bibinfo{number}{3-4} (\bibinfo{year}{2014}), \bibinfo{pages}{211--407}.
\newblock
\urldef\tempurl%
\url{https://doi.org/10.1561/0400000042}
\showDOI{\tempurl}


\bibitem[Erlingsson et~al\mbox{.}(2014a)]%
        {DBLP:conf/ccs/ErlingssonPK14}
\bibfield{author}{\bibinfo{person}{{\'{U}}lfar Erlingsson}, \bibinfo{person}{Vasyl Pihur}, {and} \bibinfo{person}{Aleksandra Korolova}.} \bibinfo{year}{2014}\natexlab{a}.
\newblock \showarticletitle{{RAPPOR:} Randomized Aggregatable Privacy-Preserving Ordinal Response}. In \bibinfo{booktitle}{\emph{Proceedings of the 2014 {ACM} {SIGSAC} Conference on Computer and Communications Security, Scottsdale, AZ, USA, November 3-7, 2014}}, \bibfield{editor}{\bibinfo{person}{Gail{-}Joon Ahn}, \bibinfo{person}{Moti Yung}, {and} \bibinfo{person}{Ninghui Li}} (Eds.). \bibinfo{publisher}{{ACM}}, \bibinfo{pages}{1054--1067}.
\newblock
\urldef\tempurl%
\url{https://doi.org/10.1145/2660267.2660348}
\showDOI{\tempurl}


\bibitem[Erlingsson et~al\mbox{.}(2014b)]%
        {erlingsson2014rappor}
\bibfield{author}{\bibinfo{person}{{\'U}lfar Erlingsson}, \bibinfo{person}{Vasyl Pihur}, {and} \bibinfo{person}{Aleksandra Korolova}.} \bibinfo{year}{2014}\natexlab{b}.
\newblock \showarticletitle{Rappor: Randomized aggregatable privacy-preserving ordinal response}. In \bibinfo{booktitle}{\emph{Proceedings of the 2014 ACM SIGSAC conference on computer and communications security}}. \bibinfo{pages}{1054--1067}.
\newblock


\bibitem[Evfimievski et~al\mbox{.}(2003)]%
        {evfimievski2003limiting}
\bibfield{author}{\bibinfo{person}{Alexandre Evfimievski}, \bibinfo{person}{Johannes Gehrke}, {and} \bibinfo{person}{Ramakrishnan Srikant}.} \bibinfo{year}{2003}\natexlab{}.
\newblock \showarticletitle{Limiting privacy breaches in privacy preserving data mining}. In \bibinfo{booktitle}{\emph{Proceedings of the twenty-second ACM SIGMOD-SIGACT-SIGART symposium on Principles of database systems}}. \bibinfo{pages}{211--222}.
\newblock


\bibitem[Farokhi(2020)]%
        {farokhi2020temporally}
\bibfield{author}{\bibinfo{person}{Farhad Farokhi}.} \bibinfo{year}{2020}\natexlab{}.
\newblock \showarticletitle{Temporally discounted differential privacy for evolving datasets on an infinite horizon}. In \bibinfo{booktitle}{\emph{2020 ACM/IEEE 11th International Conference on Cyber-Physical Systems (ICCPS)}}. IEEE, \bibinfo{pages}{1--8}.
\newblock


\bibitem[Gillenwater et~al\mbox{.}(2022)]%
        {DBLP:conf/icml/GillenwaterJMD22}
\bibfield{author}{\bibinfo{person}{Jennifer Gillenwater}, \bibinfo{person}{Matthew Joseph}, \bibinfo{person}{Andres~Mu{\~{n}}oz Medina}, {and} \bibinfo{person}{M{\'{o}}nica~Ribero Diaz}.} \bibinfo{year}{2022}\natexlab{}.
\newblock \showarticletitle{A Joint Exponential Mechanism For Differentially Private Top-k}. In \bibinfo{booktitle}{\emph{International Conference on Machine Learning, {ICML} 2022, 17-23 July 2022, Baltimore, Maryland, {USA}}} \emph{(\bibinfo{series}{Proceedings of Machine Learning Research}, Vol.~\bibinfo{volume}{162})}, \bibfield{editor}{\bibinfo{person}{Kamalika Chaudhuri}, \bibinfo{person}{Stefanie Jegelka}, \bibinfo{person}{Le~Song}, \bibinfo{person}{Csaba Szepesv{\'{a}}ri}, \bibinfo{person}{Gang Niu}, {and} \bibinfo{person}{Sivan Sabato}} (Eds.). \bibinfo{publisher}{{PMLR}}, \bibinfo{pages}{7570--7582}.
\newblock


\bibitem[Greenberg(2016)]%
        {greenberg2016apple}
\bibfield{author}{\bibinfo{person}{Andy Greenberg}.} \bibinfo{year}{2016}\natexlab{}.
\newblock \showarticletitle{Apple’s ‘differential privacy’is about collecting your data—but not your data}.
\newblock \bibinfo{journal}{\emph{Wired, June}}  \bibinfo{volume}{13} (\bibinfo{year}{2016}).
\newblock


\bibitem[Greenwald and Khanna(2001)]%
        {greenwald2001space}
\bibfield{author}{\bibinfo{person}{Michael Greenwald} {and} \bibinfo{person}{Sanjeev Khanna}.} \bibinfo{year}{2001}\natexlab{}.
\newblock \showarticletitle{Space-efficient online computation of quantile summaries}.
\newblock \bibinfo{journal}{\emph{ACM SIGMOD Record}} \bibinfo{volume}{30}, \bibinfo{number}{2} (\bibinfo{year}{2001}), \bibinfo{pages}{58--66}.
\newblock


\bibitem[Jin et~al\mbox{.}(2010)]%
        {DBLP:journals/vldb/JinYCYL10}
\bibfield{author}{\bibinfo{person}{Cheqing Jin}, \bibinfo{person}{Ke Yi}, \bibinfo{person}{Lei Chen}, \bibinfo{person}{Jeffrey~Xu Yu}, {and} \bibinfo{person}{Xuemin Lin}.} \bibinfo{year}{2010}\natexlab{}.
\newblock \showarticletitle{Sliding-window top-\emph{k} queries on uncertain streams}.
\newblock \bibinfo{journal}{\emph{{VLDB} J.}} \bibinfo{volume}{19}, \bibinfo{number}{3} (\bibinfo{year}{2010}), \bibinfo{pages}{411--435}.
\newblock
\urldef\tempurl%
\url{https://doi.org/10.1007/s00778-009-0171-0}
\showDOI{\tempurl}


\bibitem[Joseph et~al\mbox{.}(2018)]%
        {joseph2018local}
\bibfield{author}{\bibinfo{person}{Matthew Joseph}, \bibinfo{person}{Aaron Roth}, \bibinfo{person}{Jonathan Ullman}, {and} \bibinfo{person}{Bo Waggoner}.} \bibinfo{year}{2018}\natexlab{}.
\newblock \showarticletitle{Local differential privacy for evolving data}.
\newblock \bibinfo{journal}{\emph{Advances in Neural Information Processing Systems}}  \bibinfo{volume}{31} (\bibinfo{year}{2018}).
\newblock


\bibitem[Kasiviswanathan et~al\mbox{.}(2008)]%
        {DBLP:conf/focs/KasiviswanathanLNRS08}
\bibfield{author}{\bibinfo{person}{Shiva~Prasad Kasiviswanathan}, \bibinfo{person}{Homin~K. Lee}, \bibinfo{person}{Kobbi Nissim}, \bibinfo{person}{Sofya Raskhodnikova}, {and} \bibinfo{person}{Adam~D. Smith}.} \bibinfo{year}{2008}\natexlab{}.
\newblock \showarticletitle{What Can We Learn Privately?}. In \bibinfo{booktitle}{\emph{49th Annual {IEEE} Symposium on Foundations of Computer Science, {FOCS} 2008, October 25-28, 2008, Philadelphia, PA, {USA}}}. \bibinfo{publisher}{{IEEE} Computer Society}, \bibinfo{pages}{531--540}.
\newblock
\urldef\tempurl%
\url{https://doi.org/10.1109/FOCS.2008.27}
\showDOI{\tempurl}


\bibitem[Kellaris et~al\mbox{.}(2014)]%
        {DBLP:journals/pvldb/KellarisPXP14}
\bibfield{author}{\bibinfo{person}{Georgios Kellaris}, \bibinfo{person}{Stavros Papadopoulos}, \bibinfo{person}{Xiaokui Xiao}, {and} \bibinfo{person}{Dimitris Papadias}.} \bibinfo{year}{2014}\natexlab{}.
\newblock \showarticletitle{Differentially Private Event Sequences over Infinite Streams}.
\newblock \bibinfo{journal}{\emph{Proc. {VLDB} Endow.}} \bibinfo{volume}{7}, \bibinfo{number}{12} (\bibinfo{year}{2014}), \bibinfo{pages}{1155--1166}.
\newblock
\urldef\tempurl%
\url{https://doi.org/10.14778/2732977.2732989}
\showDOI{\tempurl}


\bibitem[Krempl et~al\mbox{.}(2014)]%
        {krempl2014open}
\bibfield{author}{\bibinfo{person}{Georg Krempl}, \bibinfo{person}{Indre {\v{Z}}liobaite}, \bibinfo{person}{Dariusz Brzezi{\'n}ski}, \bibinfo{person}{Eyke H{\"u}llermeier}, \bibinfo{person}{Mark Last}, \bibinfo{person}{Vincent Lemaire}, \bibinfo{person}{Tino Noack}, \bibinfo{person}{Ammar Shaker}, \bibinfo{person}{Sonja Sievi}, \bibinfo{person}{Myra Spiliopoulou}, {et~al\mbox{.}}} \bibinfo{year}{2014}\natexlab{}.
\newblock \showarticletitle{Open challenges for data stream mining research}.
\newblock \bibinfo{journal}{\emph{ACM SIGKDD explorations newsletter}} \bibinfo{volume}{16}, \bibinfo{number}{1} (\bibinfo{year}{2014}), \bibinfo{pages}{1--10}.
\newblock


\bibitem[Li et~al\mbox{.}(2022)]%
        {DBLP:journals/pvldb/LiCZYC22}
\bibfield{author}{\bibinfo{person}{Haoyu Li}, \bibinfo{person}{Qizhi Chen}, \bibinfo{person}{Yixin Zhang}, \bibinfo{person}{Tong Yang}, {and} \bibinfo{person}{Bin Cui}.} \bibinfo{year}{2022}\natexlab{}.
\newblock \showarticletitle{Stingy Sketch: {A} Sketch Framework for Accurate and Fast Frequency Estimation}.
\newblock \bibinfo{journal}{\emph{Proc. {VLDB} Endow.}} \bibinfo{volume}{15}, \bibinfo{number}{7} (\bibinfo{year}{2022}), \bibinfo{pages}{1426--1438}.
\newblock


\bibitem[Li et~al\mbox{.}(2020)]%
        {DBLP:conf/kdd/LiLXJ00DZ20}
\bibfield{author}{\bibinfo{person}{Jizhou Li}, \bibinfo{person}{Zikun Li}, \bibinfo{person}{Yifei Xu}, \bibinfo{person}{Shiqi Jiang}, \bibinfo{person}{Tong Yang}, \bibinfo{person}{Bin Cui}, \bibinfo{person}{Yafei Dai}, {and} \bibinfo{person}{Gong Zhang}.} \bibinfo{year}{2020}\natexlab{}.
\newblock \showarticletitle{WavingSketch: An Unbiased and Generic Sketch for Finding Top-k Items in Data Streams}. In \bibinfo{booktitle}{\emph{{KDD} '20: The 26th {ACM} {SIGKDD} Conference on Knowledge Discovery and Data Mining, Virtual Event, CA, USA, August 23-27, 2020}}, \bibfield{editor}{\bibinfo{person}{Rajesh Gupta}, \bibinfo{person}{Yan Liu}, \bibinfo{person}{Jiliang Tang}, {and} \bibinfo{person}{B.~Aditya Prakash}} (Eds.). \bibinfo{publisher}{{ACM}}, \bibinfo{pages}{1574--1584}.
\newblock
\urldef\tempurl%
\url{https://doi.org/10.1145/3394486.3403208}
\showDOI{\tempurl}


\bibitem[Liu et~al\mbox{.}(2010)]%
        {DBLP:journals/tods/LiuWY10}
\bibfield{author}{\bibinfo{person}{Hongyan Liu}, \bibinfo{person}{Xiaoyu Wang}, {and} \bibinfo{person}{Yinghui Yang}.} \bibinfo{year}{2010}\natexlab{}.
\newblock \showarticletitle{Comments on "an integrated efficient solution for computing frequent and top-k elements in data streams"}.
\newblock \bibinfo{journal}{\emph{{ACM} Trans. Database Syst.}} \bibinfo{volume}{35}, \bibinfo{number}{2} (\bibinfo{year}{2010}), \bibinfo{pages}{15:1--15:4}.
\newblock
\urldef\tempurl%
\url{https://doi.org/10.1145/1735886.1735894}
\showDOI{\tempurl}


\bibitem[Manku and Motwani(2002)]%
        {DBLP:conf/vldb/MankuM02}
\bibfield{author}{\bibinfo{person}{Gurmeet~Singh Manku} {and} \bibinfo{person}{Rajeev Motwani}.} \bibinfo{year}{2002}\natexlab{}.
\newblock \showarticletitle{Approximate Frequency Counts over Data Streams}. In \bibinfo{booktitle}{\emph{Proceedings of 28th International Conference on Very Large Data Bases, {VLDB} 2002, Hong Kong, August 20-23, 2002}}. \bibinfo{publisher}{Morgan Kaufmann}, \bibinfo{pages}{346--357}.
\newblock
\urldef\tempurl%
\url{https://doi.org/10.1016/B978-155860869-6/50038-X}
\showDOI{\tempurl}


\bibitem[Metwally et~al\mbox{.}(2005)]%
        {DBLP:conf/icdt/MetwallyAA05}
\bibfield{author}{\bibinfo{person}{Ahmed Metwally}, \bibinfo{person}{Divyakant Agrawal}, {and} \bibinfo{person}{Amr~El Abbadi}.} \bibinfo{year}{2005}\natexlab{}.
\newblock \showarticletitle{Efficient Computation of Frequent and Top-k Elements in Data Streams}. In \bibinfo{booktitle}{\emph{Database Theory - {ICDT} 2005, 10th International Conference, Edinburgh, UK, January 5-7, 2005, Proceedings}} \emph{(\bibinfo{series}{Lecture Notes in Computer Science}, Vol.~\bibinfo{volume}{3363})}, \bibfield{editor}{\bibinfo{person}{Thomas Eiter} {and} \bibinfo{person}{Leonid Libkin}} (Eds.). \bibinfo{publisher}{Springer}, \bibinfo{pages}{398--412}.
\newblock
\urldef\tempurl%
\url{https://doi.org/10.1007/978-3-540-30570-5\_27}
\showDOI{\tempurl}


\bibitem[Perrier et~al\mbox{.}(2018)]%
        {perrier2018private}
\bibfield{author}{\bibinfo{person}{Victor Perrier}, \bibinfo{person}{Hassan~Jameel Asghar}, {and} \bibinfo{person}{Dali Kaafar}.} \bibinfo{year}{2018}\natexlab{}.
\newblock \showarticletitle{Private continual release of real-valued data streams}.
\newblock \bibinfo{journal}{\emph{arXiv preprint arXiv:1811.03197}} (\bibinfo{year}{2018}).
\newblock


\bibitem[Qin et~al\mbox{.}(2016)]%
        {DBLP:conf/ccs/QinYYKXR16}
\bibfield{author}{\bibinfo{person}{Zhan Qin}, \bibinfo{person}{Yin Yang}, \bibinfo{person}{Ting Yu}, \bibinfo{person}{Issa Khalil}, \bibinfo{person}{Xiaokui Xiao}, {and} \bibinfo{person}{Kui Ren}.} \bibinfo{year}{2016}\natexlab{}.
\newblock \showarticletitle{Heavy Hitter Estimation over Set-Valued Data with Local Differential Privacy}. In \bibinfo{booktitle}{\emph{Proceedings of the 2016 {ACM} {SIGSAC} Conference on Computer and Communications Security, Vienna, Austria, October 24-28, 2016}}, \bibfield{editor}{\bibinfo{person}{Edgar~R. Weippl}, \bibinfo{person}{Stefan Katzenbeisser}, \bibinfo{person}{Christopher Kruegel}, \bibinfo{person}{Andrew~C. Myers}, {and} \bibinfo{person}{Shai Halevi}} (Eds.). \bibinfo{publisher}{{ACM}}, \bibinfo{pages}{192--203}.
\newblock
\urldef\tempurl%
\url{https://doi.org/10.1145/2976749.2978409}
\showDOI{\tempurl}


\bibitem[Ren et~al\mbox{.}(2022)]%
        {DBLP:conf/sigmod/RenSYYZX22}
\bibfield{author}{\bibinfo{person}{Xuebin Ren}, \bibinfo{person}{Liang Shi}, \bibinfo{person}{Weiren Yu}, \bibinfo{person}{Shusen Yang}, \bibinfo{person}{Cong Zhao}, {and} \bibinfo{person}{Zongben Xu}.} \bibinfo{year}{2022}\natexlab{}.
\newblock \showarticletitle{{LDP-IDS:} Local Differential Privacy for Infinite Data Streams}. In \bibinfo{booktitle}{\emph{{SIGMOD} '22: International Conference on Management of Data, Philadelphia, PA, USA, June 12 - 17, 2022}}, \bibfield{editor}{\bibinfo{person}{Zachary Ives}, \bibinfo{person}{Angela Bonifati}, {and} \bibinfo{person}{Amr~El Abbadi}} (Eds.). \bibinfo{publisher}{{ACM}}, \bibinfo{pages}{1064--1077}.
\newblock
\urldef\tempurl%
\url{https://doi.org/10.1145/3514221.3526190}
\showDOI{\tempurl}


\bibitem[Roy et~al\mbox{.}(2016)]%
        {DBLP:conf/sigmod/RoyKA16}
\bibfield{author}{\bibinfo{person}{Pratanu Roy}, \bibinfo{person}{Arijit Khan}, {and} \bibinfo{person}{Gustavo Alonso}.} \bibinfo{year}{2016}\natexlab{}.
\newblock \showarticletitle{Augmented Sketch: Faster and More Accurate Stream Processing}. In \bibinfo{booktitle}{\emph{Proceedings of the 2016 International Conference on Management of Data, {SIGMOD} Conference 2016, San Francisco, CA, USA, June 26 - July 01, 2016}}, \bibfield{editor}{\bibinfo{person}{Fatma {\"{O}}zcan}, \bibinfo{person}{Georgia Koutrika}, {and} \bibinfo{person}{Sam Madden}} (Eds.). \bibinfo{publisher}{{ACM}}, \bibinfo{pages}{1449--1463}.
\newblock
\urldef\tempurl%
\url{https://doi.org/10.1145/2882903.2882948}
\showDOI{\tempurl}


\bibitem[Shrivastava et~al\mbox{.}(2004)]%
        {shrivastava2004medians}
\bibfield{author}{\bibinfo{person}{Nisheeth Shrivastava}, \bibinfo{person}{Chiranjeeb Buragohain}, \bibinfo{person}{Divyakant Agrawal}, {and} \bibinfo{person}{Subhash Suri}.} \bibinfo{year}{2004}\natexlab{}.
\newblock \showarticletitle{Medians and beyond: new aggregation techniques for sensor networks}. In \bibinfo{booktitle}{\emph{Proceedings of the 2nd international conference on Embedded networked sensor systems}}. \bibinfo{pages}{239--249}.
\newblock


\bibitem[Steinfield et~al\mbox{.}(2005)]%
        {DBLP:journals/electronicmarkets/SteinfieldAL05}
\bibfield{author}{\bibinfo{person}{Charles Steinfield}, \bibinfo{person}{Thomas Adelaar}, {and} \bibinfo{person}{Fang Liu}.} \bibinfo{year}{2005}\natexlab{}.
\newblock \showarticletitle{Click and Mortar Strategies Viewed from the Web: {A} Content Analysis of Features Illustrating Integration Between Retailers' Online and Offline Presence}.
\newblock \bibinfo{journal}{\emph{Electron. Mark.}} \bibinfo{volume}{15}, \bibinfo{number}{3} (\bibinfo{year}{2005}), \bibinfo{pages}{199--212}.
\newblock
\urldef\tempurl%
\url{https://doi.org/10.1080/10196780500208632}
\showDOI{\tempurl}


\bibitem[Wang and Xu(2017)]%
        {wang2017cts}
\bibfield{author}{\bibinfo{person}{Hao Wang} {and} \bibinfo{person}{Zhengquan Xu}.} \bibinfo{year}{2017}\natexlab{}.
\newblock \showarticletitle{CTS-DP: publishing correlated time-series data via differential privacy}.
\newblock \bibinfo{journal}{\emph{Knowledge-Based Systems}}  \bibinfo{volume}{122} (\bibinfo{year}{2017}), \bibinfo{pages}{167--179}.
\newblock


\bibitem[Wang et~al\mbox{.}(2017)]%
        {wang2017locally}
\bibfield{author}{\bibinfo{person}{Tianhao Wang}, \bibinfo{person}{Jeremiah Blocki}, \bibinfo{person}{Ninghui Li}, {and} \bibinfo{person}{Somesh Jha}.} \bibinfo{year}{2017}\natexlab{}.
\newblock \showarticletitle{Locally differentially private protocols for frequency estimation}. In \bibinfo{booktitle}{\emph{USENIX Security Symposium}}. \bibinfo{pages}{729--745}.
\newblock


\bibitem[Wang et~al\mbox{.}(2021a)]%
        {wang2021continuous}
\bibfield{author}{\bibinfo{person}{Tianhao Wang}, \bibinfo{person}{Joann~Qiongna Chen}, \bibinfo{person}{Zhikun Zhang}, \bibinfo{person}{Dong Su}, \bibinfo{person}{Yueqiang Cheng}, \bibinfo{person}{Zhou Li}, \bibinfo{person}{Ninghui Li}, {and} \bibinfo{person}{Somesh Jha}.} \bibinfo{year}{2021}\natexlab{a}.
\newblock \showarticletitle{Continuous release of data streams under both centralized and local differential privacy}. In \bibinfo{booktitle}{\emph{Proceedings of the 2021 ACM SIGSAC Conference on Computer and Communications Security}}. \bibinfo{pages}{1237--1253}.
\newblock


\bibitem[Wang et~al\mbox{.}(2018)]%
        {DBLP:conf/sp/WangLJ18}
\bibfield{author}{\bibinfo{person}{Tianhao Wang}, \bibinfo{person}{Ninghui Li}, {and} \bibinfo{person}{Somesh Jha}.} \bibinfo{year}{2018}\natexlab{}.
\newblock \showarticletitle{Locally Differentially Private Frequent Itemset Mining}. In \bibinfo{booktitle}{\emph{2018 {IEEE} Symposium on Security and Privacy, {SP} 2018, Proceedings, 21-23 May 2018, San Francisco, California, {USA}}}. \bibinfo{publisher}{{IEEE} Computer Society}, \bibinfo{pages}{127--143}.
\newblock
\urldef\tempurl%
\url{https://doi.org/10.1109/SP.2018.00035}
\showDOI{\tempurl}


\bibitem[Wang et~al\mbox{.}(2021b)]%
        {DBLP:journals/tdsc/0001LJ21}
\bibfield{author}{\bibinfo{person}{Tianhao Wang}, \bibinfo{person}{Ninghui Li}, {and} \bibinfo{person}{Somesh Jha}.} \bibinfo{year}{2021}\natexlab{b}.
\newblock \showarticletitle{Locally Differentially Private Heavy Hitter Identification}.
\newblock \bibinfo{journal}{\emph{{IEEE} Trans. Dependable Secur. Comput.}} \bibinfo{volume}{18}, \bibinfo{number}{2} (\bibinfo{year}{2021}), \bibinfo{pages}{982--993}.
\newblock
\urldef\tempurl%
\url{https://doi.org/10.1109/TDSC.2019.2927695}
\showDOI{\tempurl}


\bibitem[Warner(1965)]%
        {warner1965randomized}
\bibfield{author}{\bibinfo{person}{Stanley~L Warner}.} \bibinfo{year}{1965}\natexlab{}.
\newblock \showarticletitle{Randomized response: A survey technique for eliminating evasive answer bias}.
\newblock \bibinfo{journal}{\emph{J. Amer. Statist. Assoc.}} \bibinfo{volume}{60}, \bibinfo{number}{309} (\bibinfo{year}{1965}), \bibinfo{pages}{63--69}.
\newblock


\bibitem[Yang et~al\mbox{.}(2018)]%
        {DBLP:conf/kdd/0003GZZSL18}
\bibfield{author}{\bibinfo{person}{Tong Yang}, \bibinfo{person}{Junzhi Gong}, \bibinfo{person}{Haowei Zhang}, \bibinfo{person}{Lei Zou}, \bibinfo{person}{Lei Shi}, {and} \bibinfo{person}{Xiaoming Li}.} \bibinfo{year}{2018}\natexlab{}.
\newblock \showarticletitle{HeavyGuardian: Separate and Guard Hot Items in Data Streams}. In \bibinfo{booktitle}{\emph{Proceedings of the 24th {ACM} {SIGKDD} International Conference on Knowledge Discovery {\&} Data Mining, {KDD} 2018, London, UK, August 19-23, 2018}}, \bibfield{editor}{\bibinfo{person}{Yike Guo} {and} \bibinfo{person}{Faisal Farooq}} (Eds.). \bibinfo{publisher}{{ACM}}, \bibinfo{pages}{2584--2593}.
\newblock
\urldef\tempurl%
\url{https://doi.org/10.1145/3219819.3219978}
\showDOI{\tempurl}


\bibitem[Yi and Zhang(2009)]%
        {DBLP:conf/pods/YiZ09}
\bibfield{author}{\bibinfo{person}{Ke Yi} {and} \bibinfo{person}{Qin Zhang}.} \bibinfo{year}{2009}\natexlab{}.
\newblock \showarticletitle{Optimal tracking of distributed heavy hitters and quantiles}. In \bibinfo{booktitle}{\emph{Proceedings of the Twenty-Eigth {ACM} {SIGMOD-SIGACT-SIGART} Symposium on Principles of Database Systems, {PODS} 2009, June 19 - July 1, 2009, Providence, Rhode Island, {USA}}}, \bibfield{editor}{\bibinfo{person}{Jan Paredaens} {and} \bibinfo{person}{Jianwen Su}} (Eds.). \bibinfo{publisher}{{ACM}}, \bibinfo{pages}{167--174}.
\newblock
\urldef\tempurl%
\url{https://doi.org/10.1145/1559795.1559820}
\showDOI{\tempurl}


\bibitem[Zhang et~al\mbox{.}(2013)]%
        {DBLP:conf/www/ZhangCYNL13}
\bibfield{author}{\bibinfo{person}{Xi Zhang}, \bibinfo{person}{Jian Cheng}, \bibinfo{person}{Ting Yuan}, \bibinfo{person}{Biao Niu}, {and} \bibinfo{person}{Hanqing Lu}.} \bibinfo{year}{2013}\natexlab{}.
\newblock \showarticletitle{TopRec: domain-specific recommendation through community topic mining in social network}. In \bibinfo{booktitle}{\emph{22nd International World Wide Web Conference, {WWW} '13, Rio de Janeiro, Brazil, May 13-17, 2013}}, \bibfield{editor}{\bibinfo{person}{Daniel Schwabe}, \bibinfo{person}{Virg{\'{\i}}lio A.~F. Almeida}, \bibinfo{person}{Hartmut Glaser}, \bibinfo{person}{Ricardo Baeza{-}Yates}, {and} \bibinfo{person}{Sue~B. Moon}} (Eds.). \bibinfo{publisher}{International World Wide Web Conferences Steering Committee / {ACM}}, \bibinfo{pages}{1501--1510}.
\newblock
\urldef\tempurl%
\url{https://doi.org/10.1145/2488388.2488519}
\showDOI{\tempurl}


\bibitem[Zhou et~al\mbox{.}(2018)]%
        {DBLP:conf/sigmod/Zhou0J0YLU18}
\bibfield{author}{\bibinfo{person}{Yang Zhou}, \bibinfo{person}{Tong Yang}, \bibinfo{person}{Jie Jiang}, \bibinfo{person}{Bin Cui}, \bibinfo{person}{Minlan Yu}, \bibinfo{person}{Xiaoming Li}, {and} \bibinfo{person}{Steve Uhlig}.} \bibinfo{year}{2018}\natexlab{}.
\newblock \showarticletitle{Cold Filter: {A} Meta-Framework for Faster and More Accurate Stream Processing}. In \bibinfo{booktitle}{\emph{Proceedings of the 2018 International Conference on Management of Data, {SIGMOD} Conference 2018, Houston, TX, USA, June 10-15, 2018}}, \bibfield{editor}{\bibinfo{person}{Gautam Das}, \bibinfo{person}{Christopher~M. Jermaine}, {and} \bibinfo{person}{Philip~A. Bernstein}} (Eds.). \bibinfo{publisher}{{ACM}}, \bibinfo{pages}{741--756}.
\newblock
\urldef\tempurl%
\url{https://doi.org/10.1145/3183713.3183726}
\showDOI{\tempurl}


\end{thebibliography}

\clearpage
\appendix
\section{Appendices}
\subsection{LDP Mechanisms}
\label{app:ldp-mechanism}
\noindent\textbf{Optimal Local Hash.}
The Optimal Local Hash (OLH) mechanism \cite{wang2017locally} is designed for randomizing private values in a large domain.
It maps the value with a randomly selected hash function to a new data domain with the size of $g<<d$ before randomizing the value.
The randomization method is the same as GRR with $p'$ and $q'$ as follows
\begin{equation}
	\label{eq:olh}
	\left\{
	\begin{aligned}	
		p' & = \frac{e^{\epsilon}}{e^{\epsilon} + g - 1}, \\
		q' & = \frac{1}{e^{\epsilon} + g - 1}.
	\end{aligned}
	\right. 
\end{equation}

It can also use the generic method to estimate the count $\tilde c_{i}$ with $p=p'$ and $q=\frac{1}{g}p'+\frac{g-1}{g}q'=\frac{1}{g}$.
When the size of the new data domain $g=e^\epsilon+1$, the variance of $\tilde c_{i}$ can be minimized as
\begin{equation}
  \label{eq:var_olh}
Var[\tilde c_{i}] = n\cdot\frac{4e^\epsilon}{(e^\epsilon-1)^2}
\end{equation}

\noindent\textbf{Hadamard Response.}
The Hadamard Response (HR) mechanism \cite{acharya2019hadamard, acharya2019communication} encodes private values with a $K\times K$ Hadamard matrix, where $K=2^{\lceil\log_2(d+1)\rceil}$.
Expect for the first row of the matrix (all values are `1'), each other row corresponds to a value in the data domain.
When encoding the $i^{th}$ value in the data domain, the output value is randomly selected from column indices that have `1' in the $(i+1)^{\text{th}}$ row with the probability of $p$, and selected from other indices (columns have `0') with the probability of $q$, where
\begin{equation}
  \label{eq:hr}
	\begin{cases}
		p=\frac{e^{\epsilon}}{1+e^{\epsilon}}\\
		q=\frac{1}{1+e^{\epsilon}}\\
	\end{cases}
\end{equation}

Then $\tilde c_{i}$ can be calculated as follows
\begin{equation}
  \label{eq:est_hr}
  \tilde c_i = \frac{2(e^{\epsilon}+1)}{e^{\epsilon}-1}(\hat c_i-\frac{n}{2})
\end{equation}

The variance of the estimation result $\tilde c_{i}$ is
\begin{equation}
  \label{eq:var_hr}
Var[\tilde c_{i}] = n\cdot\frac{4(e^\epsilon+1)^2}{(e^\epsilon-1)^2}
\end{equation}

\subsection{Proof of Theorem \ref{the:LdpPD}}
\label{app:ldppd_the}
\begin{proof}
  We assume that $\hat f_i$ is the frequency of noisy data recorded in $\mathcal{HG}$ according to the \emph{ED strategy} in BGR.
  Meanwhile, the \emph{ED strategy} would introduce additional error when recording $\hat f_i$, while the frequency actually recorded and used for debiasing by the GRR mechanism before publishing is $\bar f_i$.
  According to Lemma \ref{lem:pd_error}, we have the upper bound of $\hat f_i-\bar f_i$ is
  \begin{align*}
    &Pr[\hat f_i - \bar f_i\ge \alpha t]\le \frac{1}{2\alpha t}(\hat f_i-\sqrt{\hat f_i^2-\frac{4P_{weak}E(V)}{b-1}})\\
    &\Rightarrow Pr[\hat f_i \le \bar f_i + \alpha t]\ge (1-\frac{1}{2\alpha t}(\hat f_i-\sqrt{\hat f_i^2-\frac{4P_{weak}E(V)}{b-1}}))\\
  \end{align*}
  The distribution of $\hat f_i$ can be decomposed into $Bin(f_i, p)+Bin(t-f_i, q)$.
  Denote $Bin(f_i, p)$ as $P_1$ and $Bin(t-f_i, q)$ as $P_2$, according to the hoeffding inequality, we have
  \begin{align*}
    Pr[f_i p-P_1\le \zeta]\ge 1-e^{-\frac{2\zeta^2}{f_i}}, and\\
    Pr[(t-f_i) q-P_2\le \zeta]\ge 1-e^{-\frac{2\zeta^2}{t-f_i}}\\
  \end{align*}
  Then, according to Bonferroni inequality, we have
  \begin{align*}
    &Pr[P_1+P_2\ge f_i p+(t-f_i)q-2\zeta]\ge 1-e^{-\frac{2\zeta^2}{f_i}}-e^{-\frac{2\zeta^2}{t-f_i}}\\
    &\Rightarrow Pr[\hat f_i\ge f_i p+(t-f_i)q-2\zeta]\ge 1-e^{-\frac{2\zeta^2}{f_i}}-e^{-\frac{2\zeta^2}{t-f_i}}\\
    &\Rightarrow Pr[\frac{\hat f_i-tq}{p-q}-f_i\ge\frac{-2\zeta}{p-q}]\ge 1-e^{-\frac{2\zeta^2}{f_i}}-e^{-\frac{2\zeta^2}{t-f_i}}\\
    &\Rightarrow Pr[f_i-\frac{\bar f_i + \alpha t-tq}{p-q}\le f_i-\frac{\hat f_i-tq}{p-q}\le\frac{2\zeta}{p-q}]\\
    &\ge (1-e^{-\frac{2\zeta^2}{f_i}}-e^{-\frac{2\zeta^2}{t-f_i}})(1-\frac{f_i}{2\alpha t}(1-\sqrt{1-\frac{4P_{weak}E(V)}{\hat f_i^2(b-1)}}))
   \end{align*}
   \begin{align*}
    &\Rightarrow Pr[f_i-\tilde f_i\le\frac{2\zeta+\alpha t}{p-q}]\\
    &\ge (1-2e^{-\frac{2\zeta^2}{t}})(1-\frac{1}{2\alpha}(1-\sqrt{1-\frac{4P_{weak}E(V)}{b-1}}))\\
    &\Rightarrow Pr[f_i-\tilde f_i\le(\sqrt{2t\log (2/\beta)}+\alpha t)\cdot\frac{e^{\epsilon+d-1}}{e^{\epsilon}-1}]\\
    &\ge(1-\beta)(1-\frac{1}{2\alpha}(1-\sqrt{1-\frac{4P_{weak}E(V)}{b-1}}))\\
  \end{align*}
  where $\zeta=\sqrt{t\log (2/\beta)/2}$, $P_{weak}=\frac{(i-1)!(d-k)!}{(d-1)!(i-k)!}$, and $E(V)=\sum_{j=i+1}^d f_j$.
\end{proof}

\subsection{Proof of Theorem \ref{the:AdvLdpPD}}
\label{sec:Adv_the}
\begin{proof}
  We assume that $\hat f_i$ is the frequency of noisy data recorded in $\mathcal{HG}$ according to the \emph{ED strategy} in BDR.
  Meanwhile, the \emph{ED strategy} would introduce additional error when recording $\hat f_i$, while the frequency actually recorded and used for debiasing by the randomization mechanism before publishing is $\bar f_i$.
  According to Lemma \ref{lem:pd_error}, we have the upper bound of $\hat f_i-\bar f_i$ is
  \begin{align*}
    &Pr[\hat f_i - \bar f_i\ge \alpha t]\le \frac{1}{2\alpha t}(\hat f_i-\sqrt{\hat f_i^2-\frac{4P_{weak}E(V)}{b-1}})\\
    &\Rightarrow Pr[\hat f_i \le \bar f_i + \alpha t]\ge (1-\frac{1}{2\alpha t}(\hat f_i-\sqrt{\hat f_i^2-\frac{4P_{weak}E(V)}{b-1}}))\\
  \end{align*}
  The distribution of $\hat f_i$ can be decomposed into $Bin(f_i,p_1p_2)+Bin(N_h-f_i,p_1q_2)+Bin(t-N_h,q_1/k)$, where $N_h$ is the number of hot items, $N_h\le t$.
  Denote $Bin(f_i,p_1p_2)$ as $P_1$, $Bin(N_h-f_i,p_1q_2)$ as $P_2$, and $Bin(t-N_h,q_1/k)$ as $P_3$, according to the hoeffding inequality, we have
  \begin{align*}
    &Pr[f_ip_1p_2-P_1\le\zeta]\ge 1-e^{-\frac{2\zeta^2}{f_i}},\\
    &Pr[(N_h-f_i)p_1q_2-P_2\le\zeta]\ge 1-e^{-\frac{2\zeta^2}{N_h-f_i}}, and\\
    &Pr[(t-N_h)\cdot\frac{q_1}{k}-P_3\le\zeta]\ge 1-e^{-\frac{2\zeta^2}{t-N_h}}.\\
    \end{align*}
    Then, according to Bonferroni inequality, we have
    \begin{align*}
      &Pr[f_ip_1p_2-P_1+(N_h-f_i)p_1q_2-P_2+(t-N_h)\cdot\frac{q_1}{k}\\
      &-P_3\le 3\zeta]\ge 1-e^{-\frac{2\zeta^2}{f_i}}+1-e^{-\frac{2\zeta^2}{N_h-f_i}}+1-e^{-\frac{2\zeta^2}{t-N_h}}-2\\
      &\Rightarrow Pr[P_1+P_2+P3\ge f_ip_1p_2+(N_h-f_i)p_1q_2\\
      &+(t-N_h)\cdot\frac{q_1}{k}-3\zeta]\ge 1-e^{-\frac{2\zeta^2}{f_i}}-e^{-\frac{2\zeta^2}{N_h-f_i}}-e^{-\frac{2\zeta^2}{t-N_h}}\\
      &\Rightarrow Pr[\hat f_i\ge f_ip_1p_2+(N_h-f_i)p_1q_2+(t-N_h)\cdot\frac{q_1}{k}-3\zeta]\\
      &\ge 1-e^{-\frac{2\zeta^2}{f_i}}-e^{-\frac{2\zeta^2}{N_h-f_i}}-e^{-\frac{2\zeta^2}{t-N_h}}\\
      \end{align*}
     \begin{align*}
      &\Rightarrow Pr[\frac{\hat f_i-N_hp_1q_2-(t-N_h)\cdot\frac{q_1}{k}}{p_1(p_2-q_2)}-f_i\ge\frac{-3\zeta}{p_1(p_2-q_2)}]\\
      &\ge 1-e^{-\frac{2\zeta^2}{f_i}}-e^{-\frac{2\zeta^2}{N_h-f_i}}-e^{-\frac{2\zeta^2}{t-N_h}}\\
      &\Rightarrow Pr[f_i-\frac{\overline {f_i}+\alpha t-N_hp_1q_2-(t-N_h)\cdot\frac{q_1}{k}}{p_1(p_2-q_2)}\\
      &\le f_i-\frac{\hat f_i-N_hp_1q_2-(t-N_h)\cdot\frac{q_1}{k}}{p_1(p_2-q_2)}\le \frac{3\zeta}{p_1(p_2-q_2)}]\\
      &\ge (1-e^{-\frac{2\zeta^2}{f_i}}-e^{-\frac{2\zeta^2}{N_h-f_i}}-e^{-\frac{2\zeta^2}{t-N_h}})\cdot (1\\
      &-\frac{f_i}{2\alpha t}(1-\sqrt{1-\frac{4P_{weak}E(V)}{\hat{f_i}^2(b-1)}}))\\
      &\Rightarrow Pr[f_i-\tilde{f_i}\le \frac{3\zeta+\alpha t}{p_1(p_2-q_2)}]\ge (1-3e^{-\frac{2\zeta^2}{N_h}})\cdot (1\\
      &-\frac{1}{2\alpha}(1-\sqrt{1-\frac{4P_{weak}E(V)}{b-1}}))\\
      &\Rightarrow Pr[f_i-\tilde{f_i}\le(3\sqrt{\frac{N_h\log(3/\beta)}{2}}+\alpha t)\\
      &\cdot\frac{(e^{\epsilon_1}+1)(e^{\epsilon_2}+k-1)}{e^{\epsilon_1}(e^{\epsilon_2}-1)}]\\
      &\ge(1-\beta)(1-\frac{1}{2\alpha}(1-\sqrt{1-\frac{4P_{weak}E(V)}{b-1}}))
      \end{align*}

\end{proof}

\subsection{Proof of Theorem \ref{the:advLdpPD-priv}}
\label{sec:Adv_privacy}
\begin{proof}
  Denote $v$ and $v'$ as two raw data, $o_1$, $o_2$ and $o_3$ as the output of mechanisms, $\Omega_h$ as the domain of hot items, $\Omega_c$ as the domain of cold items.
  Firstly, $\mathcal{M}_{judge}$ satisfies $\epsilon_1$-LDP, we have
  \begin{align*}
    \frac{Pr[o=``Hot"|v\in \Omega_h]}{Pr[o=``Hot"|v'\in \Omega_c]} \le \frac{p_1}{q_1}=e^{\epsilon_1}
  \end{align*}
  The same result can be proved when $o=$``Cold".

  Secondly, we prove that $\mathcal{M}_{hot}$ mechanism satisfies $\epsilon_2$-LDP.
  The probability ratio of $v$ and $v'$ to get the same randomized item $o_2\in \Omega_h$ is
  \begin{align*}
    \frac{Pr[o_2|v\in \Omega_h]}{Pr[o_2|v'\in \Omega_c]}\le \frac{p_2}{1/k}\le \frac{p_2}{q_2}=e^{\epsilon_2}
  \end{align*}
  Similarly, the $\mathcal{M}_{cold}$ mechanism satisfies 
  \begin{align*}
    \frac{Pr[o_3|v\in \Omega_c]}{Pr[o_3|v'\in \Omega_h]}\le e^{\epsilon_2}
  \end{align*}
  According to the composition theorem of DP \cite{DBLP:journals/fttcs/DworkR14}, BDR satisfies $(\epsilon_1+\epsilon_2)$-LDP at each timestamp where $\epsilon_1+\epsilon_2=\epsilon$.
  Therefore, BDR satisfies $\epsilon$-LDP.
\end{proof}

\subsection{Algorithm of Function $\textsf{DSR\_Insert}$ and $\textsf{DSR\_FinalDebias}$}
\label{app:func_dsr_insert}
\setlength{\floatsep}{0.1cm}
\setlength{\textfloatsep}{0.1cm}
\begin{algorithm}[t]
  \caption{$\textsf{DSR\_Insert}$($v$, $\mathcal{HG}$, $p_1$, $q_1$, $p_2$, $q_2$)}
  \label{alg:dsr-insertfunc}
  {\small{
  \begin{algorithmic}[1]
        \State {$\Omega_s=\{\mathcal{HG}.ID\}\cup\{\bot\}$}
        \If {the least count in $\mathcal{HG}$ changed from $> 1$ to $\le 1$}
        \For {each $\mathcal{HG}[j]\in \mathcal{HG}$}
        \State {$\mathcal{HG}[j]\leftarrow (\mathcal{HG}[j].C-q_1)(p_2-q_2)/(p_1-q_1)$}
        \EndFor
        \State {Insert $v$ into $\mathcal{\mathcal{HG}}$ following ED strategy;} \Comment {$r_i^t\in\Omega$}
        \State {$num_{entire}\leftarrow num_{entire}+1$}
        \If {the least count in $\mathcal{HG}$ $\le 0$}
        \State {Replace the weakest KV pair with $<r_i^t, 1-p_1\cdot num_{entire}/(p_1-q_1)-p_2\cdot num_{reduced}/(p_2-q_2)>$}
        \EndIf
        \ElsIf {the least count in $\mathcal{HG}$ changed from $\le 1$ to $> 1$}
        \For {each $\mathcal{HG}[j]\in \mathcal{HG}$}
        \State {$\mathcal{HG}[j]\leftarrow (\mathcal{HG}[j].C-q_2)(p_1-q_1)/(p_2-q_2)$}
        \EndFor
        \State {Insert $v$ into $\mathcal{\mathcal{HG}}$ following ED strategy;} \Comment {$r_i^t\in\Omega_s$}
        \State {$num_{reduced}\leftarrow num_{reduced}+1$}
        \ElsIf {the least count in $\mathcal{HG}>1$}
        \For {each $\mathcal{HG}[j]\in \mathcal{HG}$}
        \State {$\mathcal{HG}[j]\leftarrow \mathcal{HG}[j].C-q_1$}
        \EndFor
        \State {Insert $v$ into $\mathcal{\mathcal{HG}}$ following ED strategy;} \Comment {$r_i^t\in\Omega_s$}
        \State {$num_{reduced}\leftarrow num_{reduced}+1$}
        \Else
        \For {each $\mathcal{HG}[j]\in \mathcal{HG}$}
        \State {$\mathcal{HG}[j]\leftarrow \mathcal{HG}[j].C-q_2$}
        \EndFor
        \State {Insert $v$ into $\mathcal{\mathcal{HG}}$ following ED strategy;} \Comment {$r_i^t\in\Omega$}
        \State {$num_{entire}\leftarrow num_{entire}+1$}
        \If {the least count in $\mathcal{HG}$ $\le 0$}
        \State {Replace the weakest KV pair with $<r_i^t, 1-p_1\cdot num_{entire}/(p_1-q_1)-p_2\cdot num_{reduced}/(p_2-q_2)>$}
        \EndIf
        \EndIf\\
        \Return {Updated $\mathcal{HG}$}
  \end{algorithmic}}}
\end{algorithm}

\setlength{\floatsep}{0.1cm}
\setlength{\textfloatsep}{0.1cm}
\begin{algorithm}[t]
  \caption{$\textsf{DSR\_FinalDebias}(\mathcal{HG},\ p_1,\ q_1,\ p_2,\ q_2)$}
  \label{alg:dsr-finaldebias}
  {\small{
  \begin{algorithmic}[1]
        \State {\# Complete final debias based on the state of the previous $\mathcal{HG}$.}
        \If {the least count in $\mathcal{HG}$ changed from $> 1$ to $\le 1$}
        \For {each $\mathcal{HG}[j]\in \mathcal{HG}$}
        \State {$\mathcal{HG}[j]\leftarrow \mathcal{HG}[j].C/(p_2-q_2)$}
        \EndFor
        \ElsIf {the least count in $\mathcal{HG}$ changed from $\le 1$ to $> 1$}
        \For {each $\mathcal{HG}[j]\in \mathcal{HG}$}
        \State {$\mathcal{HG}[j]\leftarrow \mathcal{HG}[j].C/(p_1-q_1)$}
        \EndFor
        \ElsIf {the least count in $\mathcal{HG}>1$}
        \For {each $\mathcal{HG}[j]\in \mathcal{HG}$}
        \State {$\mathcal{HG}[j]\leftarrow \mathcal{HG}[j].C/(p_1-q_1)$}
        \EndFor
        \Else
        \For {each $\mathcal{HG}[j]\in \mathcal{HG}$}
        \State {$\mathcal{HG}[j]\leftarrow \mathcal{HG}[j].C/(p_2-q_2)$}
        \EndFor
        \EndIf\\
        \Return {Updated $\mathcal{HG}$}
  \end{algorithmic}}}
\end{algorithm}

\subsection{Normalized Discounted Cumulative Gain (NDCG)}
\label{app:ndcg}
It measures the ordering quality of the heavy hitters captured by $\mathcal{HG}$, which is a common effectiveness in recommendation systems and other related applications.
Specifically, let $V=\{v_1,v_2,...,v_k\}$ as the Top-$k$ heavy hitters in $\mathcal{HG}$.
If $v_i$ is one of a true Top-$k$ heavy hitter, the relevance score $rel_i$ is 
\[rel_{v_i}=|k-|rank_{actual}(v_i)-rank_{estimated}(v_i)||.\]
If $v_i$ is not a true Top-$k$ heavy hitter, we directly set its $rel_i$ as $0$.
Then, the Discounted Cumulative Gain ($DCG$) is
\[DCG_k=rel_{v_1}+\sum_{i=2}^{k}\frac{rel_{v_i}}{\log_2(i)}.\]
Finally, we normalize the $DCG$ of $\mathcal{HG}$ by comparing it with the Ideal $DCG$ ($IDCG$), which is the DCG when $\mathcal{HG}$ records an actual list of Top-$k$ heavy hitters.
\[NDCG_{k}=\frac{DCG_k}{IDCG_k}.\]
$NDCG_k$ is between $0$ and $1$ for all $k$, and the closer it is to $1$ means the ordering quality of $\mathcal{HG}$ is higher.

\subsection{Implementation Details}
\label{app:imple-detail}

\subsubsection{Re-implementation for LDP Mechansims.}

We treat existing LDP frequency estimation approaches as privacy-preserving baselines. 
Specifically, we estimate the frequency of all items under LDP, and output the Top-$k$ counts as the heavy hitter. Cormode, Maddock, and Maple \cite{DBLP:journals/pvldb/CormodeMM21} placed various LDP frequency estimation approaches into a common framework, and performed an series of experiments in Python\footnote{\url{https://github.com/Samuel-Maddock/pure-LDP}}. 
Their work offered a starting point of our implementations. 

We carefully studied the source codes, and fully re-implemented all baseline LDP mechanisms with the following optimizations.

\noindent \textbf{Data serialization}. In \cite{DBLP:journals/pvldb/CormodeMM21}, the client outputs the perturbed data as an object, which server takes as its input to do data aggregation. 
In practice, the server and the client would communicate via a network channel. 
This requires object serializations and introduces additional communication and computation costs. 
In our implementation, we manually serialize the perturbed data to `byte[]' based on the underlying approaches. If the client outputs bit strings (e.g., OUE and RAPPOR), we compress the output bit string by representing each 8 bits into 1 byte. If the client outputs integers (e.g., OLH, HR), we represent the integer with the minimal byte length, i.e., 1-4 byte for integers in range $[0, 2^8)$, $[0, 2^{16})$, $[0, 2^{24})$, and $[0, 2^{32})$, respectively.

\noindent \textbf{Choices of the hash}. Some frequency estimation approaches leverage (non-cryptographic) hash to map input to Boolean (BLH) or integer(s) (RAPPOR, OLH, HCMS). 
The performances of these approaches are greately affected by the efficiency of the underlying hash. 
Meanwhile, HeavyGuardian also leverages (non-cryptographic) hash to partition data into buckets. 
Note that \cite{DBLP:journals/pvldb/CormodeMM21} and \cite{DBLP:conf/kdd/0003GZZSL18} respectively use xxHash and BobHash. 
We invoke BobHash in all schemes since our test shows that BobHash are more efficient\footnote{Our test shows that on MacBook M1, xxHash takes about $0.01$us to provide an output while BobHash takes about $0.1$us.}.

Besides, we find that debiasing the randomized data before storing it into \emph{HeavyGuardian} can avoid the bias introduced by the ED strategy being amplified by the debiasing process.
However, if a complete debiasing is performed every time randomized data comes, it can cause the previously accumulated count to be debiased repeatedly, e.g., divided by denominator $p-q$ of the debiasing formula repeatedly in BGR.
Therefore, we perform partial debiasing when collecting and dividing all counts by the denominator of the debiasing formula only before publishing the results.

\begin{figure}[h]
  \centering
  \includegraphics[width=0.4\textwidth]{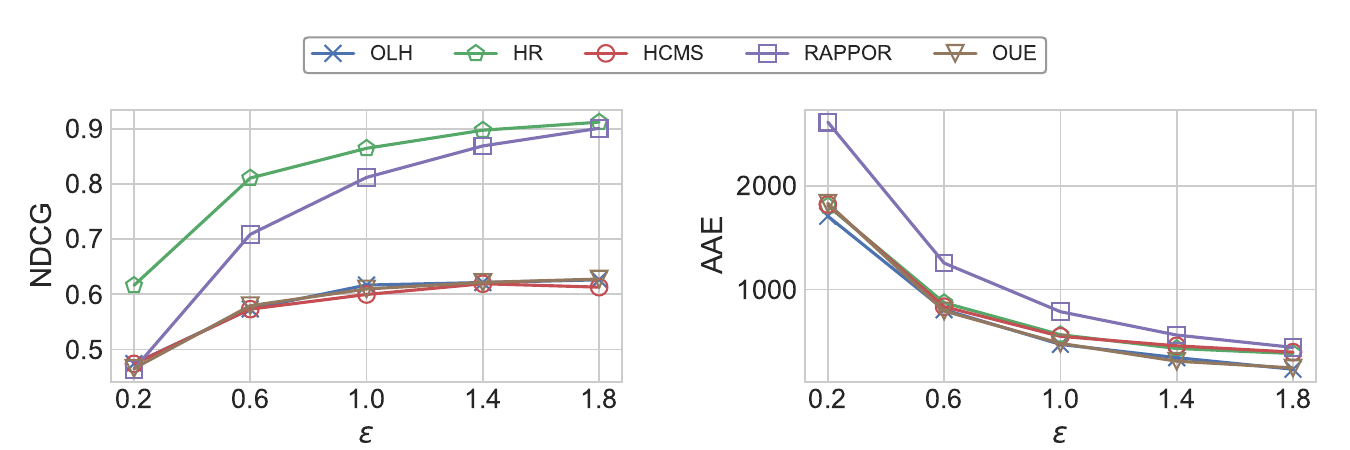}
  \caption{The comparison of different LDP mechanisms on a synthetic dataset with Normal distribution.}
  \label{fig:LDP_mech}
\end{figure}

\subsubsection{Re-implementation for HeavyGuardian}
We treat the original \emph{HeavyGuardian} as the non-private baseline. 
We carefully studied existing open-source C/C++ codes provided by Yang et al.\footnote{\url{https://github.com/Gavindeed/HeavyGuardian}} and fully re-implemented the \emph{HeavyGuardian} using Java. 
Our \emph{HeavyGuardian} re-implementation has some improvements compared with the original implementation.
First, the original implementation contains some hard-coded parameters for different tasks, while our re-implementation allows developers to dynamically config for different tasks. 
Second, the ED strategy in \emph{HeavyGuardian} contains a Bernoulli sampling procedure, i.e., sampling a Boolean value with probability $\mathcal{P} = b^{-C}$ being \textsf{True}, where $b > 1$ is a predefined constant number, and $C$ is a counting value. 
The naive method of sampling used in original \emph{HeavyGuardian} implementation is to randomly sample $r \in [0, 1)$ and test whether $r < b^{-C}$. 
However, since finite computers cannot faithfully represent real numbers, the naive method would not produce the Boolean value with the correct distribution. 
In our implementation, we parse $b^{-C} = \exp(-C \cdot \ln(b))$ and leverage the method of $\mathsf{Bernoulli}(\exp(-\gamma))$ proposed by Ganonne et al. \cite{DBLP:conf/nips/Canonne0S20} to do the sampling with no loss in accuracy.

Recall that the basic version of \emph{HeavyGuardian} is a hash table with $w \geq 1$ buckets storing KV pairs ($\langle ID, count \rangle$). 
Each bucket is divided into the heavy part with size $\lambda_h > 0$ and the light part with size $\lambda_l \geq 0$. 
Because heavy hitter detection focus on only hot items, Yang et al. \cite{DBLP:conf/kdd/0003GZZSL18} recommend setting $\lambda_l = 0$ when using \emph{HeavyGuardian} for heavy hitter detection tasks. 
Our experiments follow this recommendation and set $\lambda_l = 0$ (except CNR).
The basic version of \emph{HeavyGuardian} also allows using different $\lambda_h$ and different numbers of buckets $w$ when counting the most frequent $k$ items in the heavy hitter task. 
Although our implementation also allows setting $w$ and $\lambda_h$, our experiments focus on the basic cases, i.e., $w = 1$ and $\lambda_h = k$, to better demonstrate the effectiveness of our schemes. 

We obtain memory consumption by measuring the deep sizes (i.e., the size of an object including the size of all referred objects, in addition to the size of the object itself) of Objects packaging the \emph{HeavyGuardian} and our schemes. 
The tool we use is the JOL (Java Object Layout) library\footnote{\url{http://hg.openjdk.java.net/code-tools/jol}}. 
Although the error bounds we give in the theoretical analyses are debiasing after storing, the actual error in our implementation is still bounded by and even lower than the theoretical results.

\subsection{Supplementary Experiments}
\label{app:supp_exp}


\begin{figure*}[t]
  \centering
  \includegraphics[width=1\textwidth]{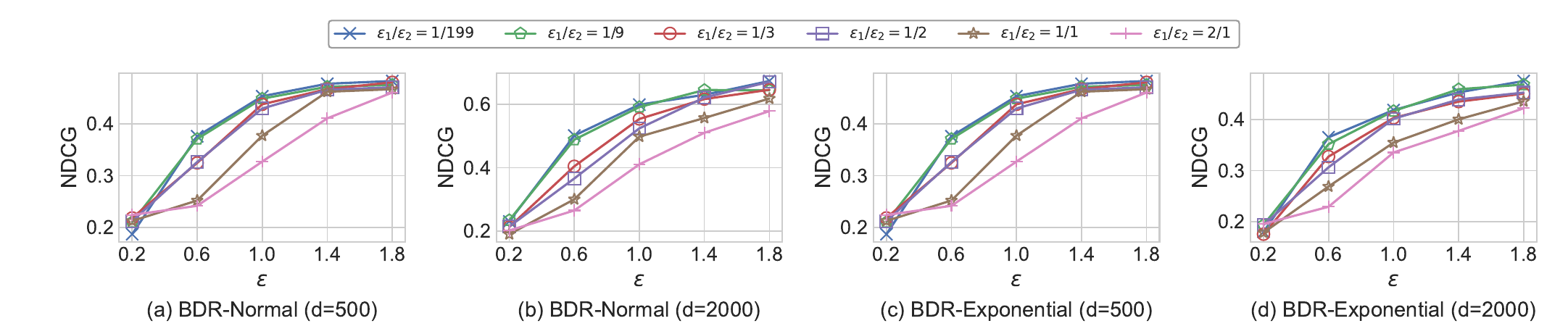}
  \caption{The impact of different allocations of privacy budget on NDCG of BDR. The evaluation is conducted on the Synthetic dataset of Normal distribution and Exponential distribution with domain size $500$ and $2000$, while taking $1\%$ data for warm-up stage.}  
  \label{fig:adv_ndcg}
\end{figure*}

\begin{figure*}[t]
  \centering
  \includegraphics[width=1\textwidth]{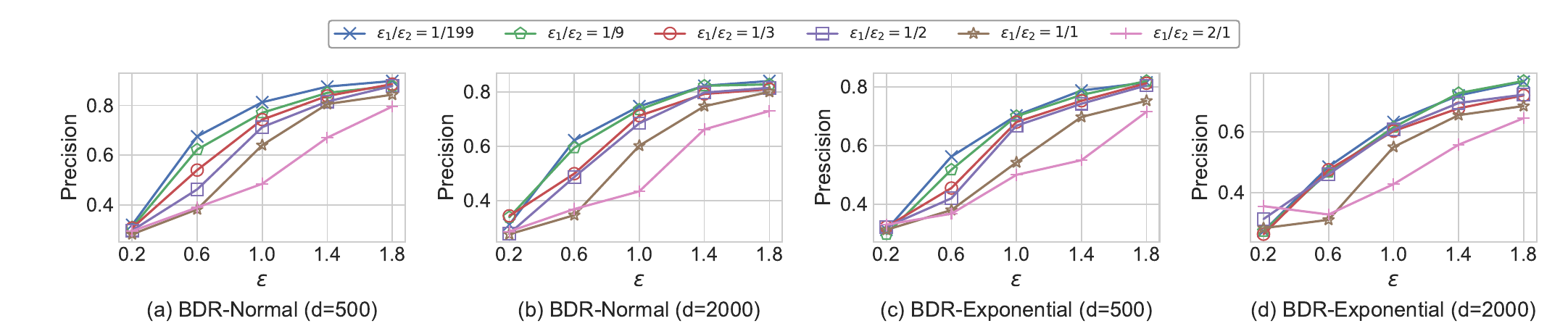}
  \vspace{-0.1in}
  \caption{The impact of different allocations of privacy budget on Precision of BDR. The evaluation is conducted on the Synthetic dataset of Normal distribution and Exponential distribution with domain size $500$ and $2000$, while taking $1\%$ data for warm-up stage.}
  \label{fig:adv_pre}
\end{figure*}

\begin{figure*}[t]
  \centering
  \includegraphics[width=1\textwidth]{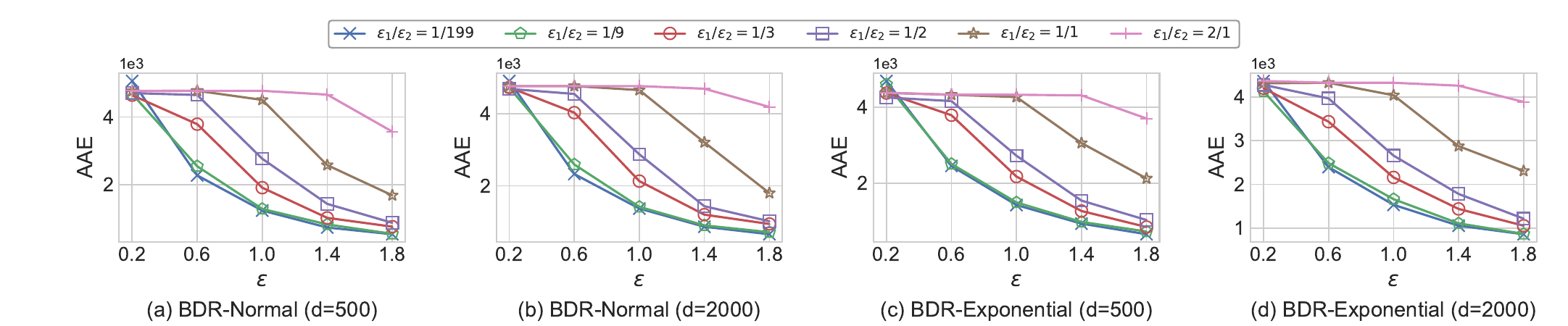}
  \vspace{-0.1in}
  \caption{The impact of different allocations of privacy budget on AAE of BDR. The evaluation is conducted on the Synthetic dataset of Normal distribution and Exponential distribution with domain size $500$ and $2000$, while taking $1\%$ data for warm-up stage.}
  \label{fig:adv_re}
\end{figure*}

\begin{figure*}[t]
  \centering
  \includegraphics[width=1\textwidth]{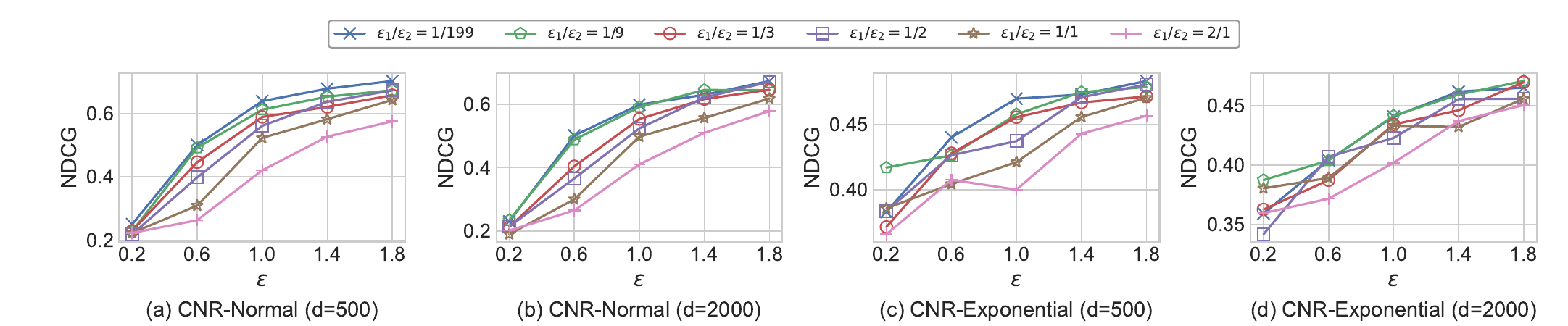}
  \vspace{-0.1in}
  \caption{The impact of different allocations of privacy budget on NDCG of CNR. The evaluation is conducted on the Synthetic dataset of Normal distribution and Exponential distribution with domain size $500$ and $2000$, while taking $1\%$ data for warm-up stage.}
  \label{fig:buff_ndcg}
\end{figure*}

\begin{figure*}[t]
  \centering
  \includegraphics[width=1\textwidth]{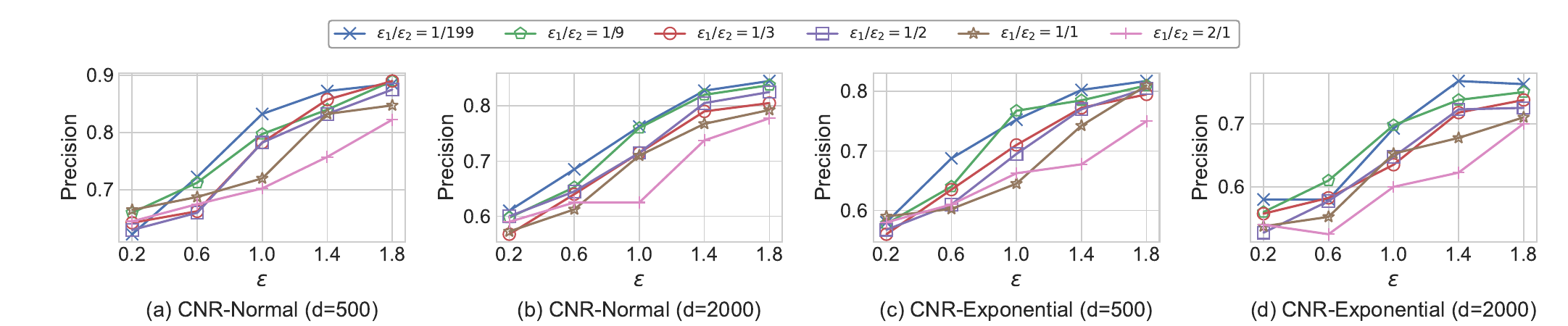}
  \vspace{-0.1in}
  \caption{The impact of different allocations of privacy budget on Precision of CNR. The evaluation is conducted on the Synthetic dataset of Normal distribution and Exponential distribution with domain size $500$ and $2000$, while taking $1\%$ data for warm-up stage.}
  \label{fig:buff_pre}
\end{figure*}

\begin{figure*}[t]
  \centering
  \includegraphics[width=1\textwidth]{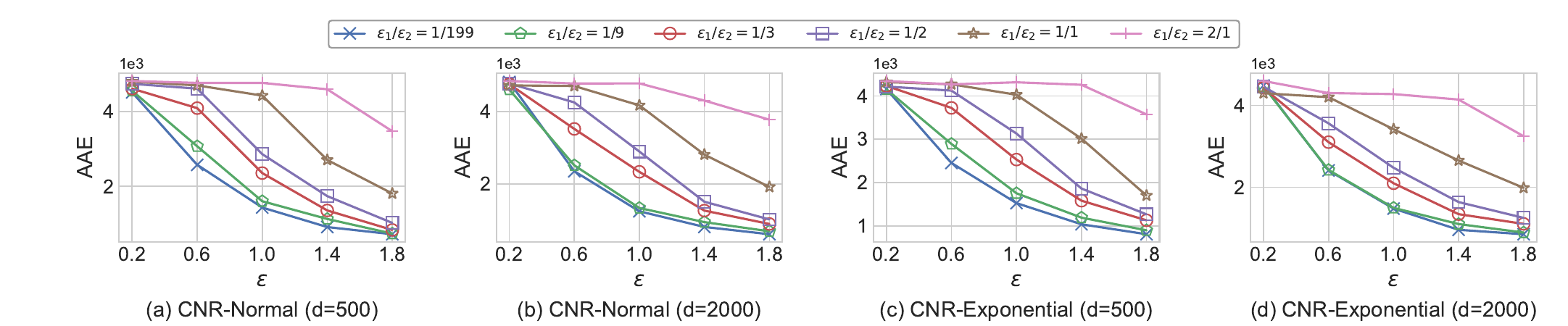}
  \vspace{-0.1in}
  \caption{The impact of different allocations of privacy budget on AAE of CNR. The evaluation is conducted on the Synthetic dataset of Normal distribution and Exponential distribution with domain size $500$ and $2000$, while taking $1\%$ data for warm-up stage.}
  \label{fig:buff_re}
\end{figure*}

\begin{figure*}[t]
  \centering
  \includegraphics[width=0.9\textwidth]{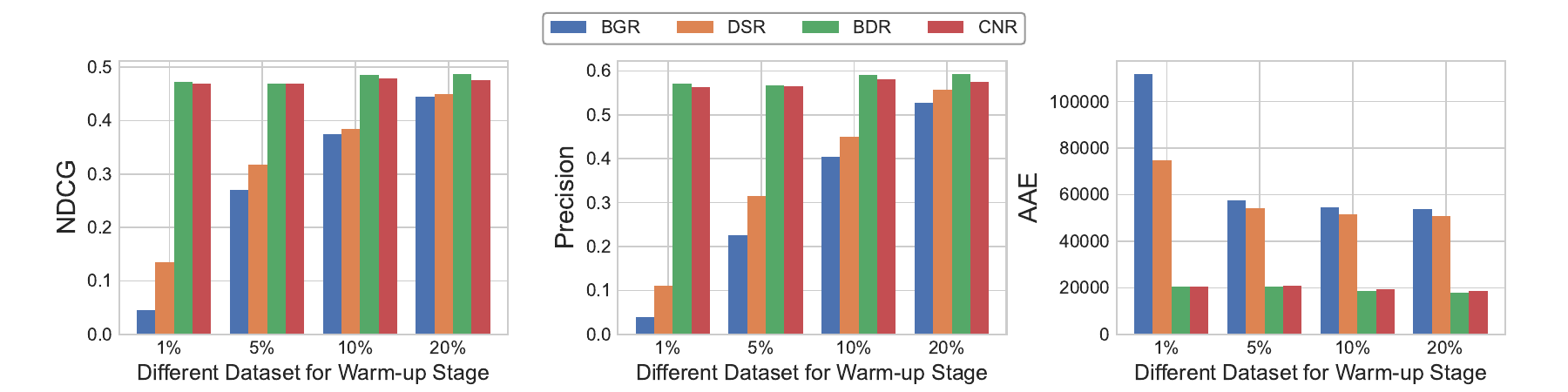}
   \vspace{-0.1in}
  \caption{The impact of the warm-up stage on the accuracy of the proposed schemes for tracking Top-$20$ items on Kosarak dataset, where $\epsilon=2$.}
  \label{fig:warmup_kosarak}
\end{figure*}

\begin{table*}[!h]
  \centering
  \caption{Total communication overheads (MB) consumed by the baseline and proposed three schemes for deriving the Top-$20$ on \textbf{each entire dataset}. The left side of ``/" is the \textbf{uploading} overheads for users, while the right side is the \textbf{downloading} overheads for users.}
  \label{tab:coummunication}

\begin{tabular}{|c|c|c|c|}
  
      \hline
       \diagbox{Scheme}{Dataset} & Kosarak ($30.5M$)& Webdocs ($1413.24M$)\\
       \hline
       \hline
       \textbf{BGR} & \textbf{35.09}/\textbf{0} & \textbf{1883.36}/\textbf{0}\\
       \hline
       \textbf{DSR} & \textbf{35.07}/\textbf{1062.79} & \textbf{1875.19}/\textbf{49891.47}\\
       \hline
       \textbf{BDR} & \textbf{26.67}/\textbf{709.04} & \textbf{1668.71}/\textbf{26295.89}\\
       \hline
       \textbf{CNR} & \textbf{26.49}/\textbf{698.95} & \textbf{1680.25}/\textbf{26861.62}\\
      \hline
  \end{tabular}
\end{table*}

\begin{table*}[!h]
  \centering
  \caption{Total server runtime (second) of the baseline and three proposed schemes for deriving the Top-$20$ on \textbf{each entire dataset}.}
  \label{tab:time-server}
  \begin{tabular}{|c|c|c|c|}
      \hline
       \diagbox{Scheme}{Dataset} & Kosarak & Webdocs \\
       \hline
       \hline
       \textbf{BGR} & \textbf{5.375} & \textbf{215.814}\\
       \hline
       \textbf{DSR} & \textbf{5.586} & \textbf{245.314}\\
       \hline
       \textbf{BDR} & \textbf{5.224} & \textbf{257.32}\\
       \hline
       \textbf{CNR} & \textbf{7.563} & \textbf{417.184}\\
      \hline
  \end{tabular}
\end{table*}

\begin{table*}[!h]
  \centering
  \caption{Total client runtime (second) of the baseline and three proposed schemes for deriving the Top-$20$ on \textbf{each entire dataset}.}
  \label{tab:time-client}
  \begin{tabular}{|c|c|c|c|}
      \hline
       \diagbox{Scheme}{Dataset} & Kosarak & Webdocs \\
       \hline
       \textbf{BGR} & \textbf{1.88} & \textbf{184.481}\\
       \hline
       \textbf{DSR} & \textbf{2.929} & \textbf{226.501}\\
       \hline
       \textbf{BDR} & \textbf{3.549} & \textbf{197.065}\\
       \hline
       \textbf{CNR} & \textbf{3.653} & \textbf{196.88}\\
      \hline
  \end{tabular}
\end{table*}

\subsection{Extended Related Work}
\label{app:related}
\noindent\textbf{Differential Private Data Stream Collection}
The earliest studies in differential privacy for streaming data collection originate from continuous observation of private data \cite{evfimievski2003limiting, bansal2008improved, erlingsson2014rappor, ding2017collecting, dwork2010differential}.
Long-term data collection from users can be regarded as the collection of data streams.
These studies mainly consider the degradation of privacy guarantee due to the repeated appearance of private data when the user's state not changing for a period of time.

Recent works on differential private data stream collection mainly focus on Centralized Differential Privacy (CDP).
Some works study how to publish the summation of the streaming data privately.
To avoid the overestimation of sensitivity of the streaming data caused by outliers, Perrier et al. \cite{perrier2018private} propose truncating the data exceeding a threshold to reduce the sensitivity. 
Wang et al. \cite{wang2021continuous} point out that the threshold should be dynamically adjusted for different distributions instead of using the fixed quantile value.
Therefore, they propose to use an exponential mechanism to get a more reasonable sensitivity.
Some works are devoted to solving the problem that the privacy guarantee is continuously degraded due to the repeated use of data in streams in consecutive periods.
kellaris et al. \cite{DBLP:journals/pvldb/KellarisPXP14} propose $w$-event DP, regarding the statistical results published in a sliding window as an event.
Farokhi et al. \cite{farokhi2020temporally} propose to address the problem of the exploding privacy budget by reducing the privacy guarantee provided for data that appeared in the past over time.
Some works study the release of correlated streaming data. 
Wang et al. \cite{wang2017cts} propose a correlated Laplace noise mechanism, and Bao et al. \cite{bao2021cgm} propose a correlated Gaussian noise mechanism.

There are also some recent works on data stream collection with Local Differential Privacy (LDP) \cite{joseph2018local, DBLP:conf/sigmod/RenSYYZX22, wang2021continuous}.
Joseph et al. \cite{joseph2018local} design a protocol to submit LDP-protected data only when the streaming data has changed and can have a greater impact on the statistical results.
Ren et al. \cite{DBLP:conf/sigmod/RenSYYZX22} propose a privacy budget segmentation framework that provides $w$-event LDP protection to prevent the degradation of privacy due to continuous data collection.
Besides, Wang et al. \cite{wang2021continuous} extend the proposed truncation-based CDP mechanism to LDP for the release of streaming data.

\noindent\textbf{Tracking Heavy Hitters in Data stream}
Mining streaming data faces three principal challenges: \emph{volume}, \emph{velocity}, and \emph{volatility} \cite{krempl2014open}.
The existing heavy hitters estimation algorithms in the data stream can be divided into three classes: Counter-based algorithms, Quantile algorithms, and Sketch algorithms \cite{DBLP:journals/pvldb/CormodeH08}. 
Counter-based algorithms track the subset of items in the stream, and they quickly determine whether to record and how to record with each new arrival data.
Manku et al. \cite{DBLP:conf/vldb/MankuM02} propose two algorithms Sticky Sampling and Lossy Counting, which only record the item and its counts whose estimated counts exceed the threshold.
Metwally et al. \cite{DBLP:conf/icdt/MetwallyAA05} design an algorithm called Space-Saving and record data with a data structure called Stream-Summary, which achieves rapid deletion, update, and insertion for each new arrival data.
Subsequently, Yang et al. \cite{DBLP:conf/kdd/0003GZZSL18} propose a new algorithm called HeavyGuardian to improve Space-Saving.
Zhou et al. \cite{DBLP:conf/sigmod/Zhou0J0YLU18} also propose a framework called Cold Filter (CF) to improve Space-Saving.
The Quantile algorithms such as the GK algorithm \cite{greenwald2001space} and QDigest algorithm \cite{shrivastava2004medians}, focus on finding the item which is the smallest item that dominates $\phi n$ items from the data stream.
Sketch algorithms record items with a data structure, which can be thought of as a linear projection of the input, hash functions are usually used to define the linear projection.
Some existing works include Count Sketch\cite{alon1999space}, CountMin Sketch\cite{cormode2005improved}, WavingSketch\cite{DBLP:conf/kdd/LiLXJ00DZ20}, and Moment\cite{chi2004moment}, etc.
However, the sketch algorithms may involve a large number of hash operations, which cannot meet the timeliness requirements of streaming data.
Besides, the unimportant low-frequent items are all recorded, which leads to unnecessary memory consumption.
Our design is based on Counter-based algorithms with an extended setting where streaming data is protected by LDP.

\end{document}